\documentclass[11pt]{article}

\usepackage[top=50mm, bottom=50mm, left=50mm, right=50mm]{geometry}

\usepackage{lineno}
\usepackage{amssymb}
\usepackage{amsmath}
\usepackage{amsthm}
\usepackage{epsfig}
\usepackage{graphicx}
\usepackage{graphics}
\usepackage{float}
\usepackage{subfigure}
\usepackage{multirow}
\usepackage{color}
\usepackage{lineno}
\usepackage{fullpage}
\usepackage[normalem]{ulem} 
\usepackage{makeidx}
\usepackage{xspace}
\usepackage{wrapfig}

\usepackage{booktabs}
\usepackage{tabularx}
\usepackage{amsfonts}
\usepackage{hyphenat}
\usepackage{hyperref}

\usepackage{times}
\usepackage{psfrag}
\usepackage{ifpdf}
\usepackage{array}

\usepackage{setspace}
\usepackage{latexsym}
\usepackage{enumerate}
\usepackage{mathtools}

\usepackage{dblfloatfix}

\usepackage{jin}

\makeindex

\begin{document}

\title{Boundary Defense against Cyber Threat for Power System Operation}

\date{}
\author{
{Ming Jin\footnote{Department of Industrial Engineering and Operation Research, University of California Berkeley, CA 94720, USA}} \and {Javad Lavaei\footnote{Department of Industrial Engineering and Operation Research, University of California Berkeley, CA 94720, USA}} \and {Somayeh Sojoudi\footnote{Department of Electrical Engineering and Computer Sciences, University of California Berkeley, CA 94720, USA}} \and {Ross Baldick \footnote {Department of Electrical and Computer Engineering, University of Texas at Austin, TX 78712, USA}}}



\maketitle
\setcounter{page}{0}

\begin{abstract}
The operation of power grids is becoming increasingly data-centric. While the abundance of data could improve the efficiency of the system, it poses major reliability challenges. In particular, state estimation aims to learn the behavior of the network from data but an undetected attack on this problem could lead to a large-scale blackout. Nevertheless, understanding vulnerability of state estimation against cyber attacks has been hindered by the lack of tools studying the topological and data-analytic aspects of the network. Algorithmic robustness is of critical need to extract reliable information from abundant but untrusted grid data. We propose a robust state estimation framework that leverages network sparsity and data abundance. For a large-scale power grid, we quantify, analyze, and visualize the regions of the network prone to cyber attacks. We also propose an optimization-based graphical boundary defense mechanism to identify the border of the geographical area whose data has been manipulated. The proposed method does not allow a local attack to have a global effect on the data analysis of the entire network, which enhances the situational awareness of the grid especially in the face of adversity. The developed mathematical framework reveals key geometric and algebraic factors that can affect algorithmic robustness and is used to study the vulnerability of the U.S. power grid in this paper.  

\end{abstract}

\newpage

\begin{figure}[b!]
\centering
\includegraphics[width=\linewidth]{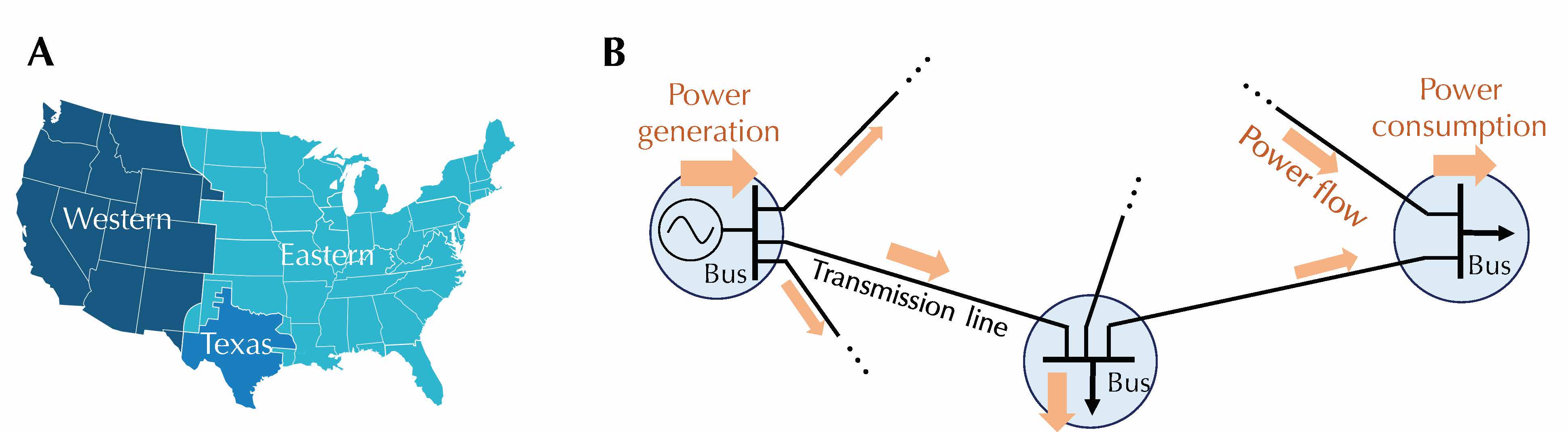}
\caption{\textbf{The U.S. power transmission network.} \textbf{(A)} Map of the Eastern, Western and Texas interconnections. \textbf{(B)} Schematic diagram of a portion of the network. Each blue circle indicates a node (e.g., generator bus or load bus). Nodes are connected by transmission lines. Power is generated, transported and consumed in different locations (the amount of power is shown as the width of the orange arrow).   }
\label{fig:intro}
\end{figure}
While real-world data abound for many complex systems, they are often noisy and corrupted. Acquiring reliable information from abundant but untrusted data is key to enhancing cybersecurity for mission-critical systems, such as transportation and power grid. Since many of these systems are inherently network structured, data analytics cannot be satisfactorily understood without incorporating their underlying graph topologies. Consider the power system state estimation (SE) as an example, which constantly monitors the status of the grid by filtering and fusing a large volume of data every few minutes. It plays a critical role in the economic and reliable operation of the grid because major operational problems such as security-constrained optimal power flow, contingency analysis, and transient stability analysis rely on its output. The current industry practice is based on a set of heuristic iterative algorithms proposed in the 70s, which are known empirically to work properly under normal situations. However, those algorithms become brittle under adverse conditions, such as natural hazards, equipment faults, and even cyber attacks. The significance of functioning SE to operators was illustrated by the 2003 large-scale blackout, in which the failure of SE contributed to the inability of providing real-time diagnostic support.\cite{muir2004final} Despite substantial advances in algorithm design, namely using semidefinite programming, holomorphic embedding load flow methods, and homotopy continuation methods, a major obstacle still remains: the lack of a  framework for the design of a robust and scalable algorithm together with a realistic evaluation of its vulnerability.\cite{abur2004power,huang2012state,molzahn2019survey} Developing such a framework is challenging for three reasons: {(a)} the model of a power system is highly nonlinear and nonconvex due to physical laws, {(b)} computational resources required by the existing algorithms grow rapidly in the size of the system, and {(c)} the number of scenarios for adverse conditions is too large to be enumerated (it is higher than the number of atoms in the observable universe for systems with as low as 500 possible attack points). These challenges have limited the scope of previous studies to simple approximate models or conservative methods that ignore the topology-dependent characterization of vulnerabilities.\cite{huang2012state,molzahn2019survey} Similar hurdles exist in studying vulnerability of data analytics for other large-scale complex graphs, including ecological and social systems,\cite{vespignani2010complex,ganin2016operational} due to the lack of statistical tools for untrusted data with underlying nonlinear and structured (rather than random) graphical models.

Here we focus on the U.S. grid, which is the largest machine on earth with more than 200,000 miles of transmission lines (Figure \ref{fig:intro}). It consists of three large and nearly independent synchronous systems (Eastern, Western, and Texas) that together span the lower 48 United States, most of Canada, and some parts of Mexico. Due to confidentiality requirements on critical infrastructure information, we report our findings on modified grids, which match the size, complexity, and characteristics of actual grids.\cite{birchfield2017grid} Basic properties of the data are listed in Table \ref{tab:summary_lp_socp}. Central to the vulnerability analysis is that we provide formal statistical guarantees that rely on the physical and cyber infrastructure, which can be realistically evaluated on any large-scale system to depict high-granularity characteristics through graph topology, as shown in Figure \ref{fig:inco_US}.

\begin{figure}[th]
\centering
\includegraphics[width=\linewidth]{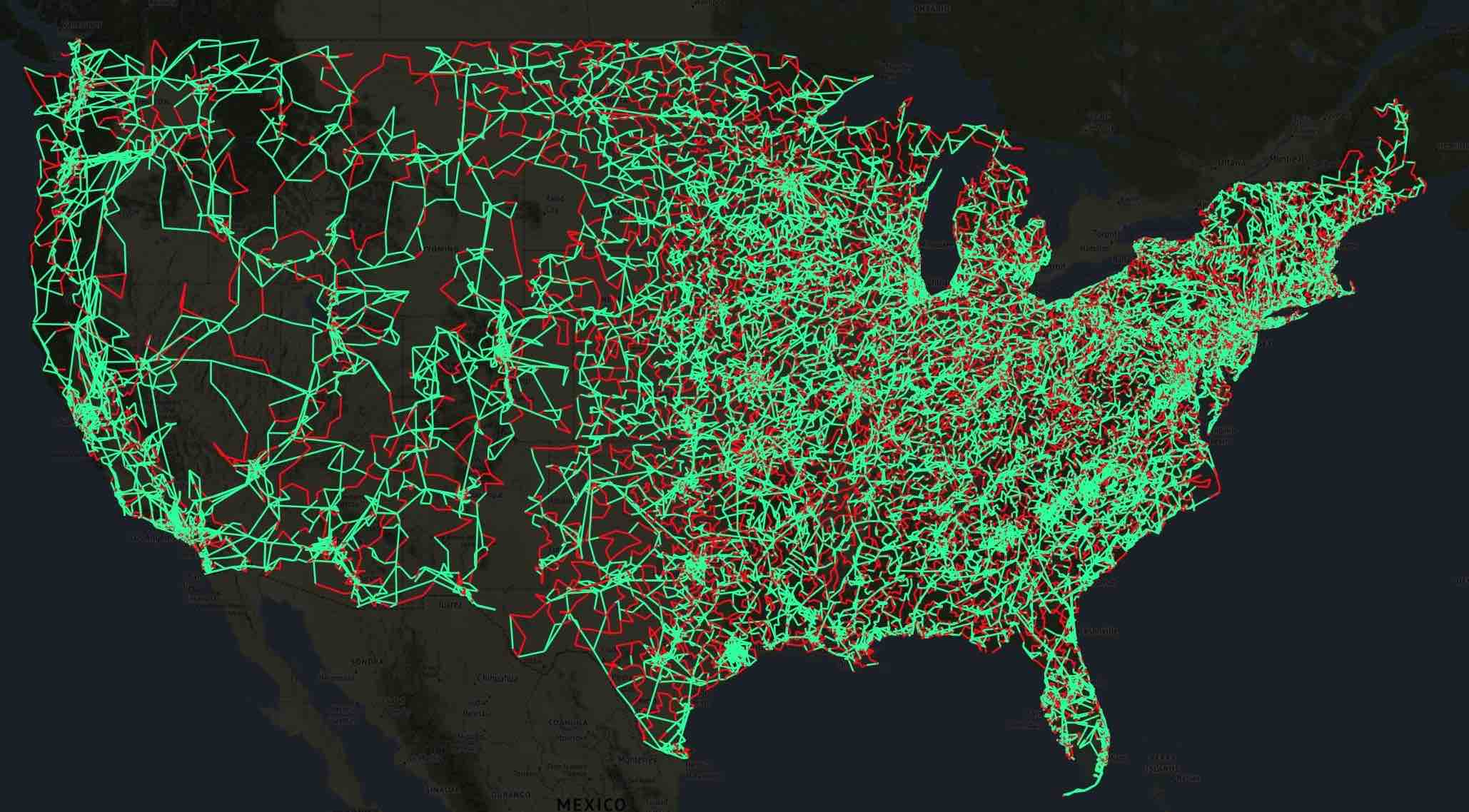}
\caption{\textbf{Vulnerability map of the modified U.S. power grid.} A line is robust (shown in green) if it stops error propagation from one end to the other during state estimation; otherwise, it is vulnerable (shown in red). Because the vulnerability map varies according to the set of available measurements, we demonstrate the map for a single profile with abundant data, which includes voltage magnitude as well as real and reactive power injections per bus and real in addition to reactive power flows per branch. }
\label{fig:inco_US}
\end{figure}

\section*{Data abundance meets algorithmic robustness}

Existing SE software solves nonlinear least squares (NLS) for the set of complex voltage phasors based on power flow measurements. NLS is a nonconvex problem, so even in the absence of measurement errors, local search algorithms such as Newton's method can become ``stuck'' at local minima, which are spurious and do not correspond to a useful estimate of the state. When this occurs, the conventional wisdom is to conclude that the estimations are unduly influenced by bad data, which are subsequently identified, down-weighted and even discarded to rectify the outputs.\cite{merrill1971bad,kotiuga1982bad,monticelli2000electric,huang2012state} Nevertheless, this is misleading and even harmful, especially during unusual or emergency situations when accurate estimates are needed, because erroneously rejecting useful information can further reduce the reliability. Even though advanced convex relaxation techniques, such as semidefinite programming, can partially address this issue,\cite{molzahn2019survey} the primary disadvantage is their heavy computational and memory requirements. 

Compared to the {classic state estimation} where one needs to obtain useful information from limited data, there is a paradigm shift from a \emph{limited-but-trustworthy} data regime to an \emph{abundant-but-untrusted} data regime due to the significant growth in instrumentation and communication in the electric grid,. Hence, a natural question arises: {Can  the additional information from abundant data sources be leveraged to enhance the robustness of the algorithm?} 

In this section, we provide a strong positive answer to this challenging question. We propose a new representation for the common types of measurements, such as the real and reactive power flows and voltage magnitudes, by fully exploiting the sparsity structure of power networks. This representation framework comprises physical quantities such as voltage magnitudes squared and phasor products over the lines. A key advantage is that one can express all power flow measurements as a \emph{linear} combination of these basic parameters. From a computational complexity perspective, this enables depicting the boundary between easy and difficult instances of SE with respect to the number and locations of sensor measurements. Particularly, it is well-known that the SE problem is usually unidentifiable (i.e., there are multiple solutions that are consistent with the measurements) in the traditional power flow setting, where each bus has only 2 sensor measurements. Yet, our analysis shows that the problem becomes solvable as soon as the SE problem is modestly over-determined  (i.e., there exists a method to uniquely recover the true state of the system). In contrast, this guarantee is almost vacuous using existing theoretical tools.\cite{chi2018nonconvex}

While the new parameter representation is effective under clean data, it turns out that it can be used to deal with corrupted and untrusted data as well. To this end, we propose a two-step pipeline. The first step is to solve a convex optimization,
where the objective deals with both dense noise due to measurement error and sparse noise due to bad data. Because the variables correspond to physical quantities, they can be mathematically constrained within a set of second-order cones (SOCs) to improve robustness, though the unconstrained versions based on linear programming (LP) or quadratic programming (QP) are also viable options. Based on the estimations from Step 1, the next step directly reconstructs the voltage phasors from the set of linear bases using elementary algebra. As a general remark, the rationale of the design of the optimization algorithm in Step 1 can be also explained by an interesting connection to the robust statistics literature. It can be shown that the optimization is equivalent to minimizing the Huber loss, which is well-known to be robust to outliers (supplementary material). Previous methods have incorporated Huber loss, but they are either in the setting of the DC approximation or a nonconvex formulation.\cite{mili1996robust,baldick1997implementing} There is also a lack of theoretical understanding of the performance in the literature. Furthermore, the incorporation of conic constraints tightens the relaxation, but the analysis becomes more involved.

\begin{figure}[h]
\centering
\includegraphics[width=.9\linewidth]{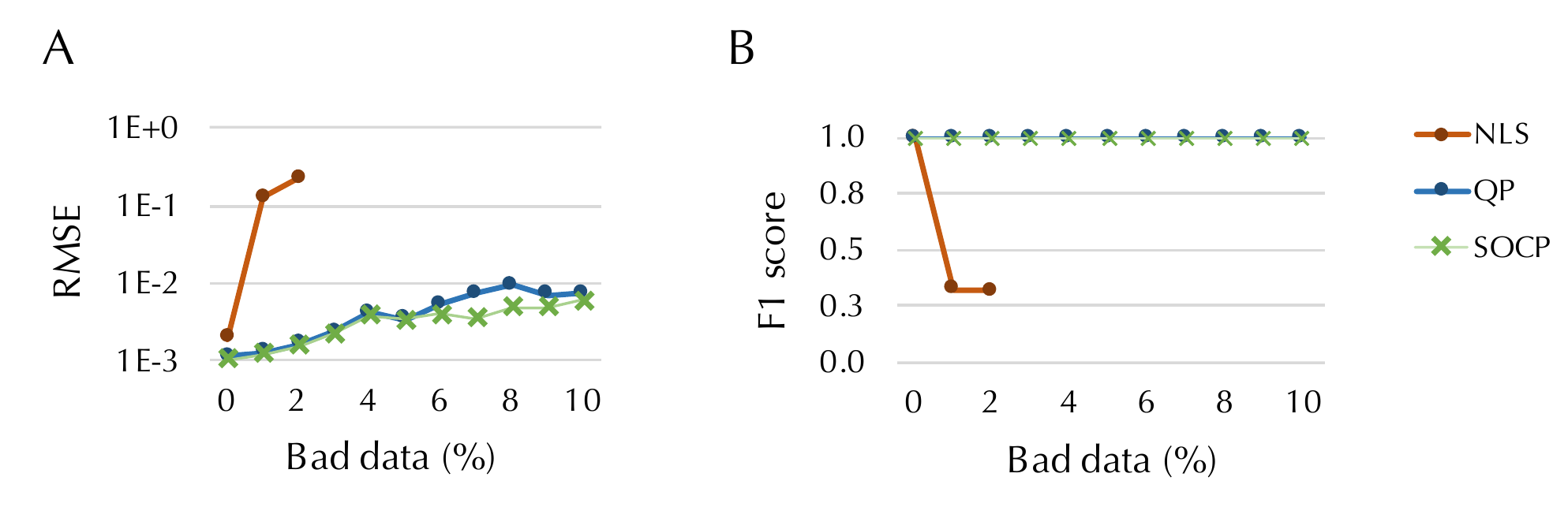}
\caption{\textbf{Evaluation of algorithmic robustness for different levels of bad data.} The bad data are generated by the ``scattered attack'' strategy, where a subset of lines are chosen whose branch measurements are all corrupted. For state estimation, we consider Newton's method to solve nonlinear least squares (NLS) as the baseline, and compare it with the proposed methods based on quadratic programming (QP) and second-order cone programming (SOCP). For each percentage of bad points within dataset, we show \textbf{(A)} the root mean squared error (RMSE) and \textbf{(B)} the F1 score of bad data detection, averaged over 20 independent simulations. For RMSE, a desirable value is any number less than 0.01. The F1 score is the harmonic average of precision and recall. Because NLS deteriorates significantly with the addition of bad data, we only show the simulation results up to 2\% of bad data, which corresponds to about 380 number of bad data. We tested on the synthetic Texas Interconnection with full sensor measurement set.}
\label{fig:BDD}
\end{figure}

Next, we provide a theoretical guarantee for global recovery in the case of sparse bad data. Consider the following ``corrupted sensing model'' where the nonlinearity is hidden within a linear model using our method to be explained later:
\begin{equation}
    \ybf=\Abf\xbf+\wbf+\bbf.
    \label{equ:cor_meas}
\end{equation}
Here, $\ybf$ is the set of $m$ sensor measurements for the vector $\xbf$ consisting of latent variables, $\Abf$ is the sensing matrix, $\wbf$ is the dense random noise due to measurement error, and $\bbf$ is sparse bad data. Let $\Jcal\subset\{1,...,m\}$ denote a set of measurements that are biased by sparse dense noise $\bbf$ (i.e., $\bbf_k\neq 0$ if and only if $k\in\Jcal$) and $\Jcal^c$ be its complement set; let $\Abf_\Jcal$ be the submatrix with rows indexed by $\Jcal$ in $\Abf$. Denote the pseudoinverse of $\Abf_{\Jcal^c}$ as $\Abf^{+}_{\Jcal^c}=(\Abf_{\Jcal^c}^\top \Abf_{\Jcal^c})^{-1}\Abf_{\Jcal^c}^\top$. Then, under some mild ``observability condition,'' the proposed two-step pipeline (the unconstrained version) can simultaneously recover the true state and detect the bad data if the following condition is satisfied:
\begin{equation}
\rho_{\text{GRC}}(\Jcal)=\|\Abf_{\Jcal^c}^{+\top}\Abf_\Jcal^\top\|_{\infty}<1,
\tag{GRC}
\label{equ:GRC}
\end{equation}
where $\|\cdot\|_\infty$ denotes the matrix infinity norm (i.e., the maximum absolute column sum of the matrix). Intuitively, $\rho_{\text{GRC}}(\Jcal)$ measures the alignment of the corrupted data and the clean data. The condition states that if the bad data are not aligned with the benign data, then it is possible to detect them.

We compare the proposed technique to the conventional approach based on Newton's method with bad data detection (BDD) in some empirical evaluations. For Newton's method, measurements with residual larger than a threshold are removed and SE is re-solved. As shown in Figure \ref{fig:BDD}, our method significantly improves on Newton's method in terms of estimation accuracy and bad data detection rates. Contrary to the prevailing wisdom, BDD is not effective because it depends on the quality of the initial estimation from Newton's method, which can be badly influenced by the bad data. In contrast, the key idea of the proposed method is to incorporate a BDD term in the optimization so that the best configuration of the state estimation and bad data vector be detected simultaneously. Since we solve either LP/QP or second-order cone programming (SOCP), the runtime is manageable for real-time applications. We also see that the SOCP version is more robust than the unconstrained LP/QP version (see Figure 3S in the supplementary material for close comparison) .

\section*{From global recovery to boundary defense}

The above results demonstrate that the proposed two-step pipeline can deal with random sparse bad data. More importantly, we are concerned about the scenario where the bad data are engineered, or a whole subregion's data are compromised. This corresponds to adverse conditions such as cyberattacks or natural disasters. In this situation, Newton's method is particularly vulnerable, because by simply solving the nonlinear least squares, the influence of bad data will propagate throughout the system, as shown in Figure \ref{fig:attack_sim_texas}. What can we say about the robustness in this case?

\begin{figure}[b!]
\centering
\includegraphics[width=\linewidth]{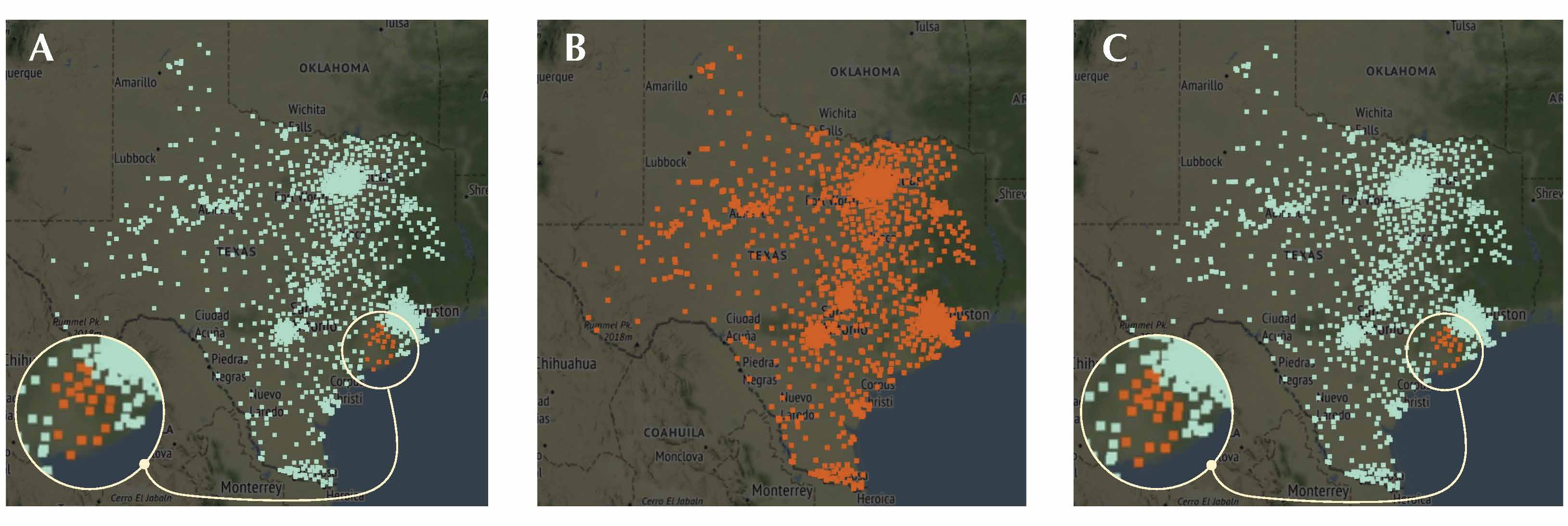}
\caption{\textbf{Evaluation of the boundary defense mechanism.} \textbf{(A)} The grid is under ``zonal attack,'' where the measurements within a zone are corrupted (shown in red). State estimation based on \textbf{(B)} Newton's method for nonlinear least squares, and \textbf{(C)} the proposed method with SOC constraints, where in both cases, buses with an estimation error greater than 0.002 are marked red. The errors propagate throughout the grid in \textbf{(B)}, but are contained within the zonal boundary in \textbf{(C)}.}
\label{fig:attack_sim_texas}
\end{figure}

It turns out that to defend against cyberattacks where the data for a geographical area are attacked, we need to devise a new defense mechanism. Because the basic ``observability'' condition is not satisfied, it is unrealistic to recover the state within the region. A well-defined problem is how to identify the boundary of the attacked region to be able to \emph{limit the spread from and impact of small disruptions to local regions}. 

To this end, we propose a new notion of defense on networks, called ``boundary defense mechanism.'' For a given attack scenario, there is a natural partition of the network into the attacked, inner and outer boundaries, and safe regions (Figure \ref{fig:bound_def}(A)). If boundary defense is successful, then no matter how erroneous the state estimation is within the attacked region, the estimates at the boundary and in the safe region are unaffected. This is a fairly general framework, because it incorporates a wide range of adversarial scenarios that are localized, including line outage, substation down, natural disasters, and cyberattacks. However, a key challenge arises: due to the large number of possible scenarios, it is clearly unrealistic to evaluate the efficacy of boundary defense separately for each case. How can one provide a systematic assessment of the robustness that can be applied to a variety of scenarios?

\begin{figure}[t]
\centering
\includegraphics[width=\linewidth]{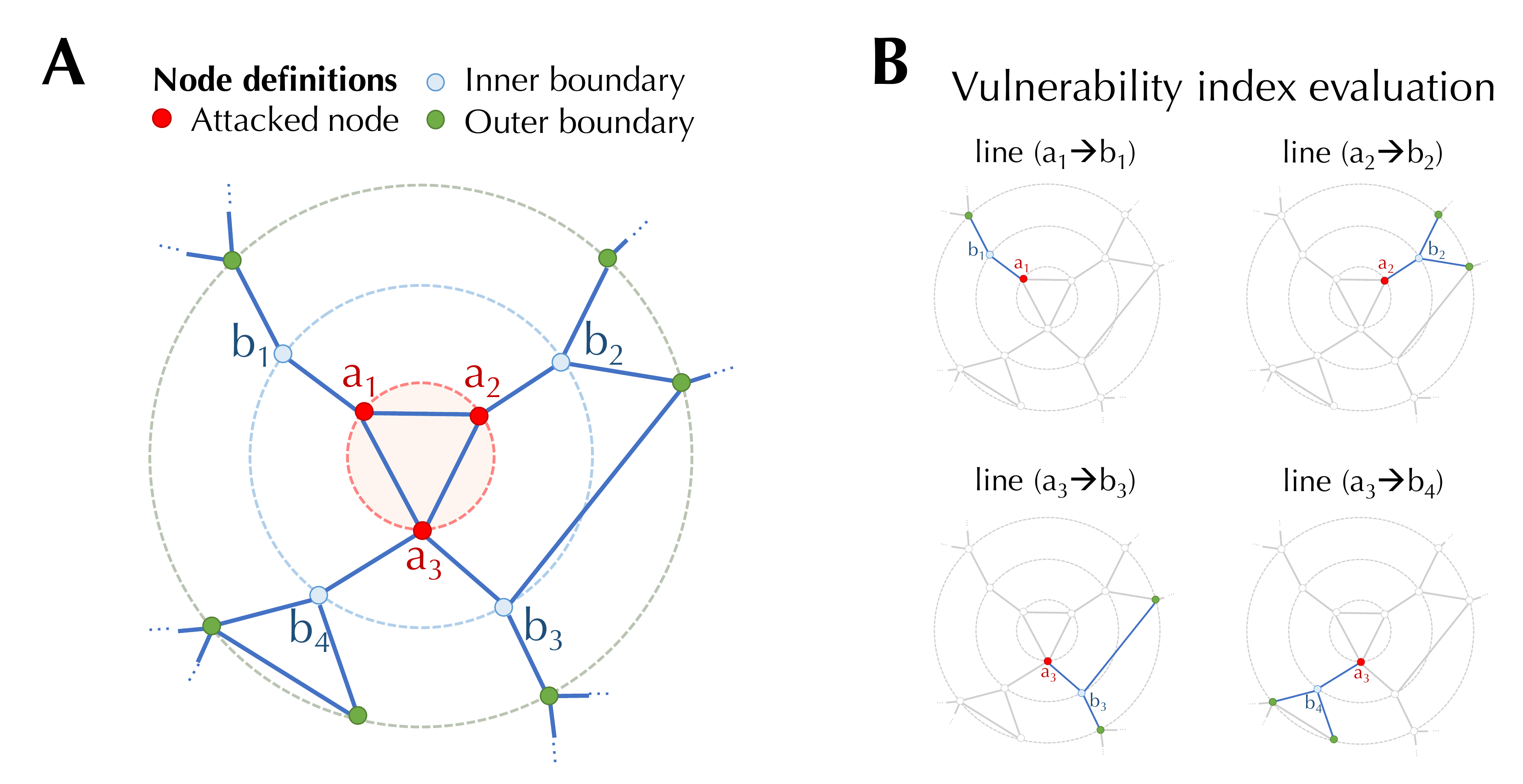}
\caption{\textbf{Illustration of the boundary defense mechanism.} \textbf{(A)} Schematic diagram showing the attacked nodes as well as inner and outer boundary nodes. \textbf{(B)} Vulnerability index evaluation. Only nodes and lines considered in the evaluation are highlighted for each line evaluation, with each line direction considered from the attacked node to the inner boundary node.}
\label{fig:bound_def}
\end{figure}

Our key idea is that instead of treating the power grid as a collection of buses and lines, we analyze each line individually. Specifically, we associate a ``vulnerability index'' (VI) to each line in one of the 2 directions. Note that VI is algorithmically dependent. In the case of unconstrained LP or QP, this metric for line $i\rightarrow j$ is given by the following minimax optimization:
\begin{equation}
    \aijlp=\underset{\|\bfxi\|_\infty\leq 1
	}{\text{max~~~}}\quad\underset{\hbf\in\Hcallp(\bfxi)}{\text{min~~~}} \quad\|\hbf\|_\infty
	\tag{VI}
    \label{equ:vi-lp}
\end{equation}
where $\Hcallp(\bfxi)=\Big\{\hbf\mid \Abf_{\Mbdijg,\Xbdij}^\top{\hbf}+\Abf_{\Mbdijb,\Xbdij}^\top{\bfxi}=\bfzero\Big\}$ is the set of admissible $\hbf$ for a given vector $\bfxi$ in the unit hypercube, the index sets $\Mbdijg$ and $\Mbdijb$ correspond to the defending and defective measurements on the boundary, and $\Xbdij$ denotes the set of variables associated with the boundary (supplementary material). We use the subscript notation $\Abf_{\Mbdijg,\Xbdij}$ to indicate the submatrix of $\Abf$ whose rows are indexed by $\Mbdijg$ and columns are indexed by $\Xbdij$. Figure \ref{fig:bound_def}(B) illustrates the nodes and lines relevant to the evaluation of four lines for a given attack scenario. The case of SOCP is defined similarly:
\begin{equation}
    \aijsocp(\xbf)=\underset{\|\bfxi\|_\infty\leq 1
	}{\text{max~~~}}\quad\underset{\hbf\in\Hcalsocp(\bfxi,\xbf)}{\text{min~~~}} \quad\|\hbf\|_\infty
	\tag{VI-SOC}
    \label{equ:vi-socp}
\end{equation}
where $\Hcalsocp(\bfxi,\xbf)$ is the set of admissible $\hbf$ defined in the supplementary material. Firstly, it can be seen that \eqref{equ:vi-socp} depends on the true state $\xbf$. However, this is not an issue because it can be shown that for every $\xbf$ that corresponds to a complex voltage state of the system, we have \[\aijsocp(\xbf)\leq \aijlp\]
In other words, the incorporation of second-order cone constraints \emph{always} improves robustness.

Our main result is stated in the following theorem (formal statement can be found in the supplementary):
\begin{theorem}[Boundary defense mechanism]
Consider a partition of the network into the attacked, boundary, and safe regions, where the bad data are contained within the attacked region. If the vulnerability index (LP/QP or SOCP) in the direction that points outwards from the attacked region is less than 1 for all lines on the boundary, then the solution obtained from the two-step pipeline has the following properties: \textbf{(i)} all the detected data are bad data, so there are no false positives in Step 1; and \textbf{(ii)} after removing the subgraph of the attacked region from the main graph, direct recovery in Step 2 recovers the underlying state of the system for the un-attacked region.
\end{theorem}
From an attacker's point of view, by attacking data in a local region, the adversary hopes that error would propagate throughout the system due to miscalculation. Indeed, this normally occurs for Newton's method. In contrast, the above theorem guarantees that it will never happen using the proposed algorithm as long as the boundary defense condition is satisfied. This explains the intriguing phenomenon observed in the beginning---for the case of topological errors (line or substation outage) or cyberattacks, the boundary defense mechanism is ``triggered'' to contain the error within the local neighborhood.

\section*{Geographic mapping of vulnerabilities}
\begin{figure}[b!]
\centering
\includegraphics[width=\linewidth]{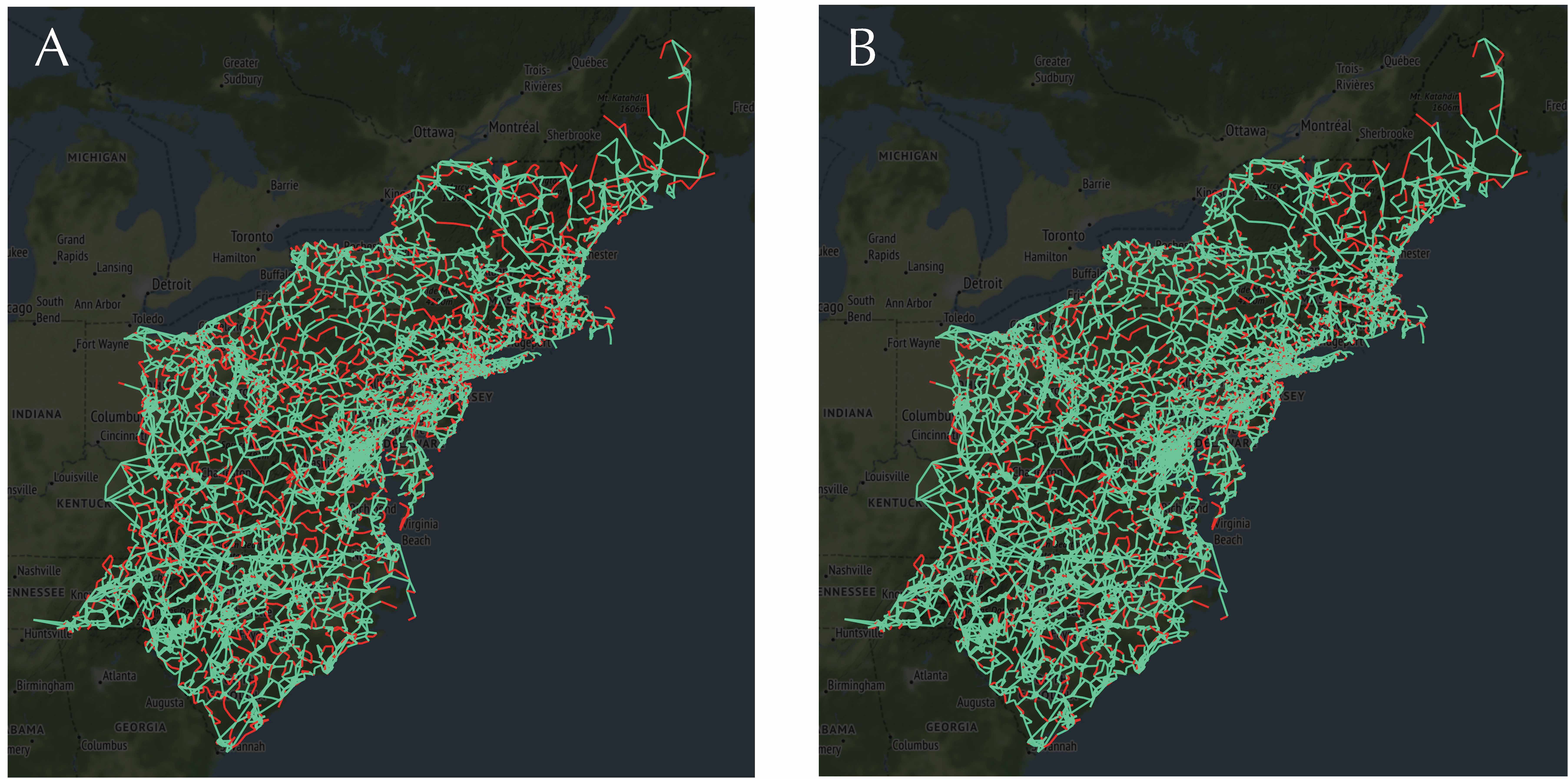}
\caption{\textbf{Comparison of vulnerability maps under different optimization strategies.} Vulnerability maps when using the proposed \textbf{(A)} LP/QP and \textbf{(B)} SOCP are shown, where a robust line is marked green and a vulnerable line is colored red. }
\label{fig:inco_US_lp_socp}
\end{figure}

\begin{figure}[t]
\centering
\includegraphics[width=\linewidth]{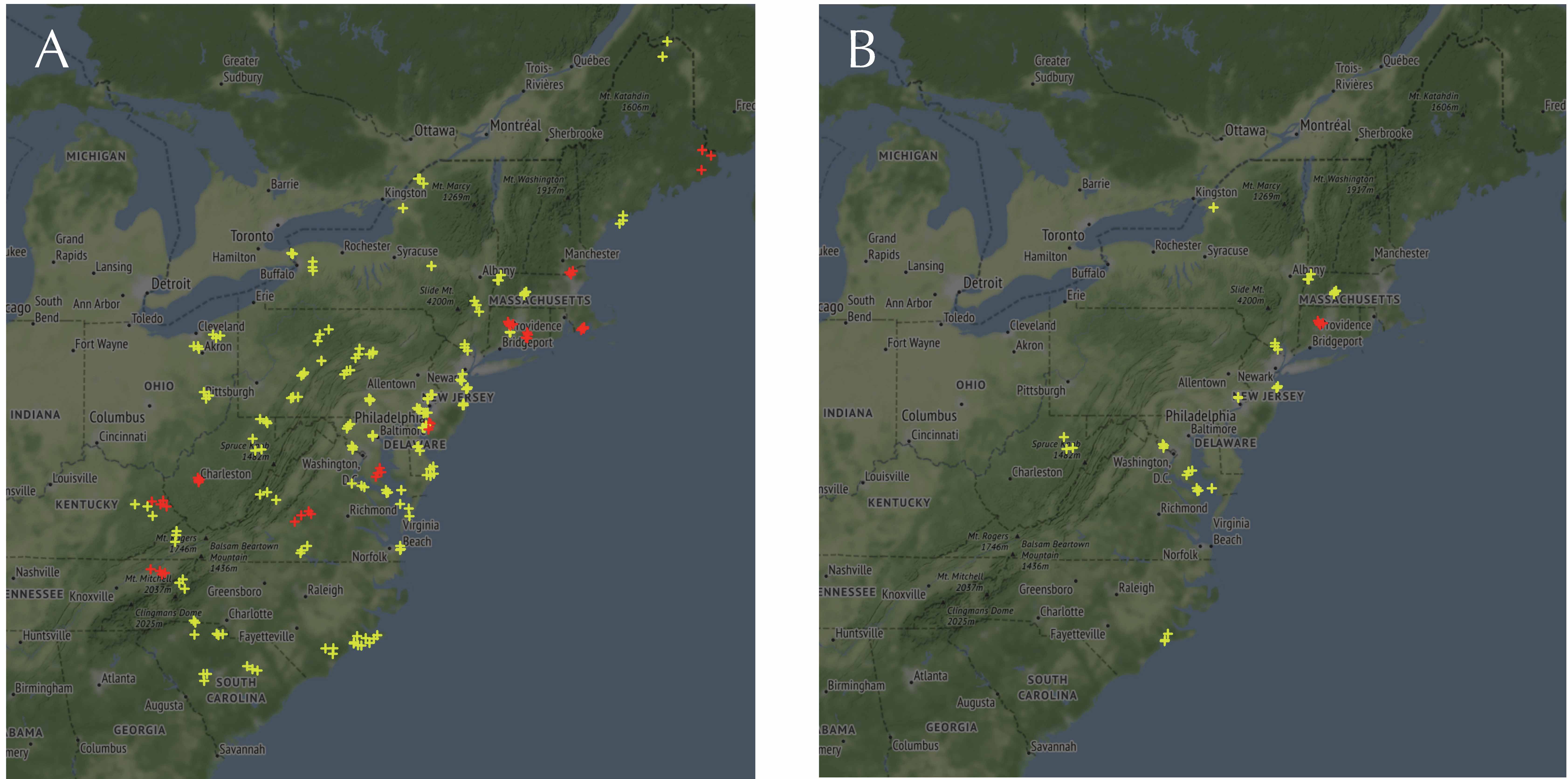}
\caption{\textbf{Comparison of bus critical index maps under different optimization strategies.} Since the bus critical indices are no larger than 3 within the map, we only show the locations with values 2 (yellow) and 3 (red) for the proposed \textbf{(A)} LP/QP and \textbf{(B)} SOCP state estimation strategies. }
\label{fig:ci_US_lp_socp}
\end{figure}

Based on the mathematical tools developed in the previous section, we assess the robustness of the synthetic U.S. grid. First, we visualize the vulnerability index on the map for both \eqref{equ:vi-lp} and \eqref{equ:vi-socp} in Figure \ref{fig:inco_US_lp_socp}. Due to its dependence on the underlying state, \eqref{equ:vi-socp} is shown for a profile described by the dataset, which represents a snapshot of the operating status. A line is considered ``robust'' if the VIs in both directions are less than 1; otherwise, it is ``vulnerable (V-line).'' The plot shows a geographic distribution of robust/vulnerable lines for the east coast of the U.S. grid. It can be seen that the density of vulnerable lines is relatively high for populated areas like Boston and New York, where we also observe a high density of robust lines. On average, 59\% lines are robust across the states, which is further split into each of the independent synchronous regions, as shown in Table \ref{tab:summary_lp_socp}. In addition, it can be validated in the map that \eqref{equ:vi-socp} always improves \eqref{equ:vi-lp}, which implies that the incorporation of SOCP constraints can help rectify state estimation and BDD.

The vulnerability map can be used in various ways. For instance, it can be used to investigate whether topological errors for a line or a substation can be contained locally. This corresponds to the case when there is a model mismatch for a transmission line or substation, such that the associated measurements are largely biased. While this is a challenging problem, it could be addressed using the vulnerability map systematically. Specifically, if the erroneous line/substation is surrounded by robust lines, then it is guaranteed that the error will be contained locally via the boundary defense mechanism. Otherwise, there is a possibility that error will ``escape'' through a vulnerable line to affect the outside region, which is referred to as a ``critical line (C-line)'' or a ``critical bus (C-bus).'' In particular, for topological errors such as line mis-specification, it can be regarded as a pair of gross injection errors at the two ends of the line; hence, we can identify it as long as the line is not a C-line. A summary of statistics is shown in Table \ref{tab:summary_lp_socp}.

\begin{table}[t]
\centering
\begin{tabular}{@{}lcccccccccc@{}}
 & \multicolumn{2}{c}{Basic properties} & \multicolumn{4}{c}{Properties of LP/QP} & \multicolumn{4}{c}{Properties of SOCP} \\ \cmidrule(l){2-11} 
 & Buses & Lines  & V-lines & C-lines & C-bus & Bus CI  & V-lines & C-lines & C-bus & Bus CI  \\ \midrule
Texas & 2,000 & 3,206  & .3762 & .4251 & .4775 & .20  & .2979 & .3674 & .4225 & .06  \\ \midrule
Western & 10,000 & 12,706  & .4715 & .5231 & .5313 & .15  & .3979 & .4636 & .4860 & .06  \\ \midrule
Eastern & 70,000 & 88,207 & .4932 & .5415 & .5327 & .14  & .4104 & .4780 & .4810 & .05 
\end{tabular}%
\caption{\textbf{Summary statistics of network properties and vulnerability characteristics.} We show the percentage of V-lines and C-lines among all network lines, and the percentage of C-bus among all network buses for LP/QP and SOCP. We also show the average bus critical index, which measures the influence of a single-bus attack on the rest of the network.}
\label{tab:summary_lp_socp}
\end{table}

Furthermore, we can extend the case study by defining a criticality index (CI) for each substation. CI gauges how many nodes in its neighborhoods will be affected if this substation is down. The higher the value, the more crucial the situation if the substation is compromised. This is analogous to the cascading failures for generators, but the difference is clear---our focus is on the robustness of data analytics rather than physical dynamics. For each substation, CI can be calculated as the size of the connected component rooted at the node, where an edge between two nodes is present if and only if the physical line that connects them is vulnerable. We visualize the distribution of CI on the map as shown in Figure \ref{fig:ci_US_lp_socp}. It can be seen that they are concentrated in populated areas.

\section*{Relating vulnerability to network and optimization properties}
\begin{figure}[h]
\centering
\includegraphics[width=.9\linewidth]{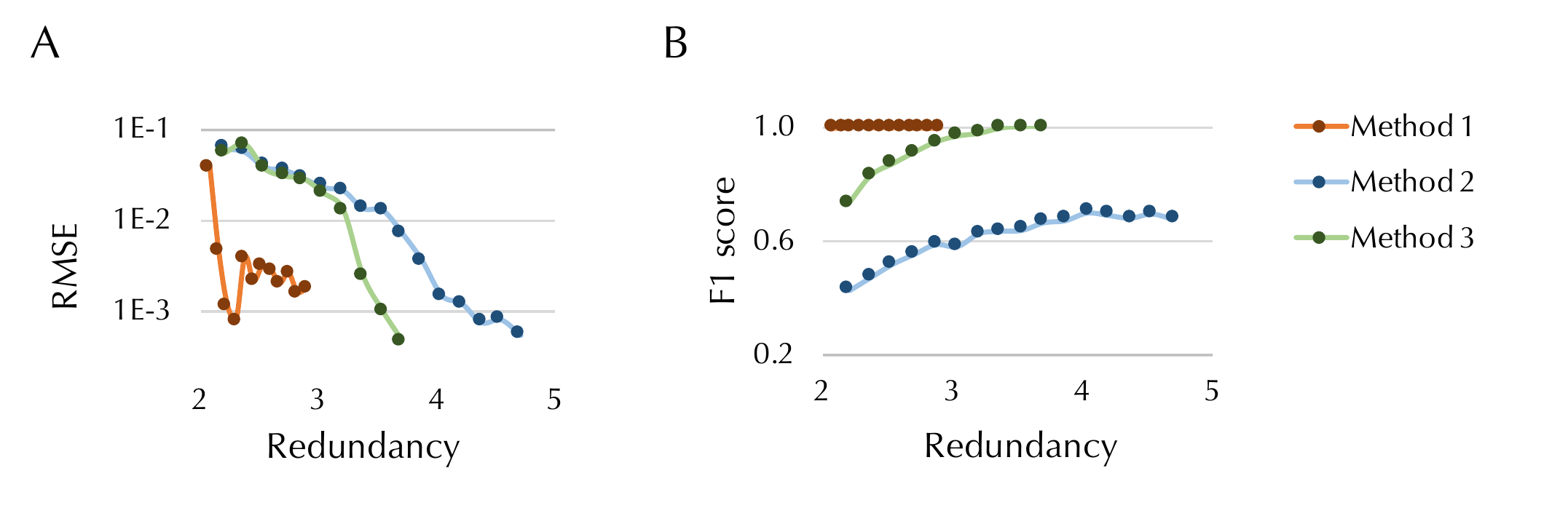}
\caption{\textbf{Comparison of different measurement profiles and redundancy.} We consider three different methods for sensor augmentation, as detailed in the main text. The redundancy value is calculated as the number of sensors divided by $2\times n_b (\text{number of buses})-1$, which is the degree of freedom in the traditional power flow problem. Each point for \textbf{(A)} RMSE and \textbf{(B)} F1 score is obtained by averaging over 20 independent simulations. }
\label{fig:acc_f1_redundancy_socp}
\end{figure}
To investigate factors that affect line vulnerability, we shift our focus to the underlying network and optimization properties. So far, our study has been conducted with respect to a specific measurement profile, which corresponds to the set of full nodal and branch measurements. An important question is: \emph{How do the number and locations of measurement sensors affect line vulnerability?} In particular, does decreasing the number of sensors make the network significantly more vulnerable, and what type of sensor measurements can bolster boundary defense? 

For this purpose, we examine three methods for ``measurement augmentation.'' The first method (Method 1) starts from a spanning tree of the network and incrementally adds a set of lines to the tree to obtain a subgraph that will be used for taking measurements. In this method, each bus is equipped with only voltage magnitude measurements, and each line has 3 out of 4 branch flow measurements. The second method (Method 2) starts with the full network, where each node has voltage magnitude measurements and each line has one real and one reactive power measurements, and it grows the set of sensors by randomly adding branch measurements. The third method (Method 3) differs from Method 2 only in that it grows the set of sensors by randomly adding branch measurements as well as nodal power injections. To evaluate these three methods, we devise a ``scattered attack'' strategy, where we randomly select 25 lines of the 2000-bus Texas Interconnection and corrupt all of its branch measurements, which amounts to 100 bad data. We then employ our proposed method to first detect bad data, and then rerun SE on the sanitized measurement set. The observation is that, in general, both the root mean squared error (RMSE) and the F1 score for bad data detection are enhanced as more sensors are added to the network, as shown in Figure \ref{fig:acc_f1_redundancy_socp}. Specifically, an F1 score close to 1 indicates that the algorithm detects all bad data (high recall rate) and does not falsely blame the good data (high precision rate).

\begin{figure}[t]
\centering
\includegraphics[width=.9\linewidth]{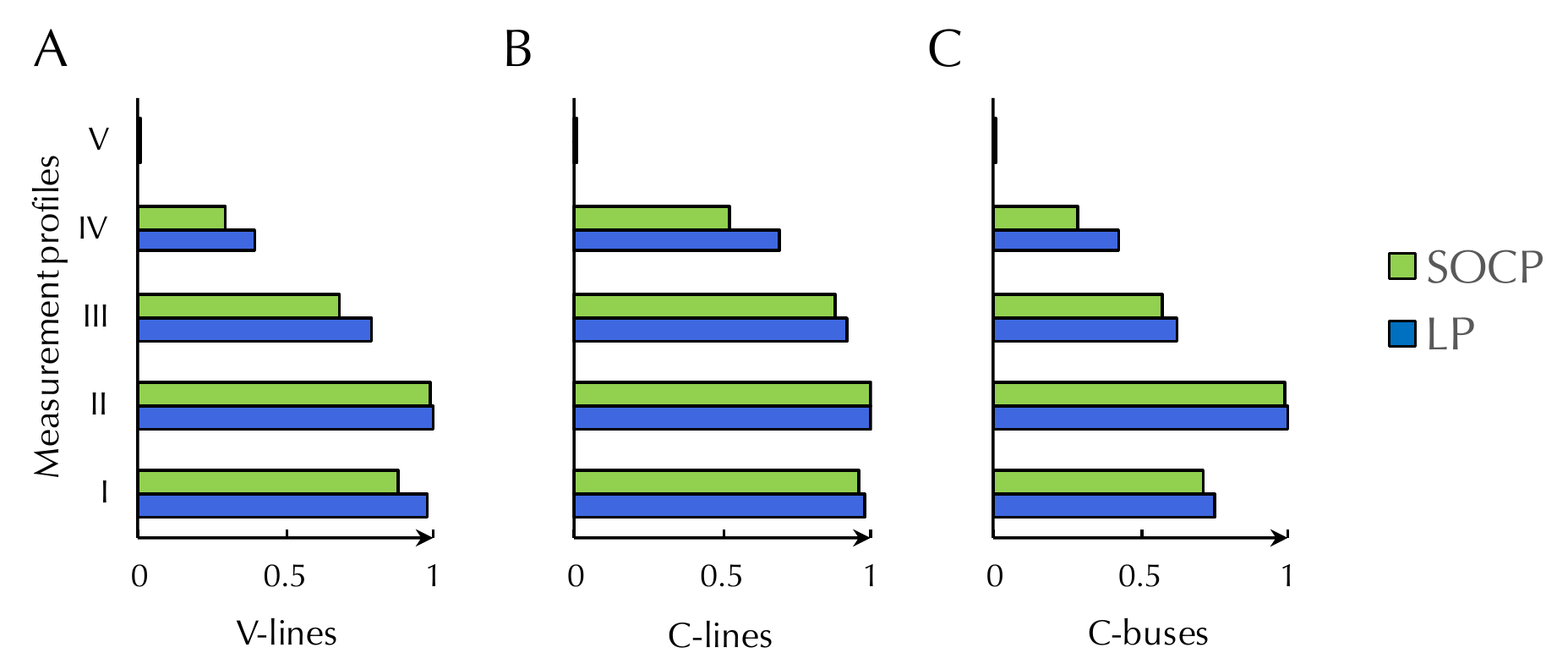}
\caption{\textbf{Characterization of vulnerability based on measurement profiles.} The five measurement profiles are full nodal measurements and 2/3/4 branch flows per line (I/III/IV); real and reactive power injections per bus and 3 branch flows per line (II); and voltage magnitude per bus and 3 branch flows per line (V). For each state estimation method (LP/QP or SOCP), we show the percentage of \textbf{(A)} V-lines, \textbf{(B)} C-lines and \textbf{(C)} C-buses within the Texas Interconnection.}
\label{fig:meas_profile_mi}
\end{figure}
There is also a major discrepancy among different methods for the same level of measurement redundancy. For instance, Method 1 significantly outperforms the other two methods at a low redundancy rate, whereas Method 2 steadily outmatches Method 3 with more sensors.  To explain this phenomenon, we need to examine the types of available measurements. Thus, we select five typical measurement profiles as snapshots of Figure \ref{fig:acc_f1_redundancy_socp} and calculate the percentage of V-lines, C-lines, and the average CI in each case (Figure \ref{fig:meas_profile_mi}). It turns out that the inclusion of voltage magnitude or branch flow measurements can enhance the robustness, whereas the addition of nodal power injections is a major factor that weakens the defense. For example, with only voltage magnitude and branch flow measurements, the network is almost ``everywhere defendable,'' namely the locations of scattered attack can be accurately detected with high probability. On the contrary, with the inclusion of nodal injections, even with a high rate of branch flow measurements, the network is still vulnerable. Intuitively, this is because nodal power injections are highly coupled measurements, which depend on state variables at all lines connected to the node. When one or few of the branches are under attack, this can lead to miscalculations at all incident lines. In contrast, voltage magnitude and branch flows are more localized measurements, whose corruptions have less effects on adjacent buses/lines.

In addition to the measurement set, network vulnerability also depends on topological properties. In particular, our findings show that the connectivity degree for each node is positively correlated with line vulnerability (Figure \ref{fig:degree_mi}(A)). For a boundary defense node, it is increasingly likely to defend against attacks as the degree grows. However, this trend is less obvious when the node is under attack. The reason is that high-degree nodes have more unattacked measurements to leverage in order to rectify the corrupted lines. On the other hand, it is more likely for a line to be critical if it is connected to a high-degree bus, as is shown in Figure \ref{fig:degree_mi}(B). This is because by the definition of critical line, as long as one of the remaining lines incident to that bus is vulnerable, then the error will propagate out through the vulnerable line. Similarly, a high-degree node is more likely to be a critical bus. In addition to the degree of connections, which is a local property, we have observed an interesting relation to the tree decomposition of the network, which provides a generalization of the discussed method. However, due to the technicality, we leave it to the supplementary materials.

\begin{figure}[t]
\centering
\includegraphics[width=\linewidth]{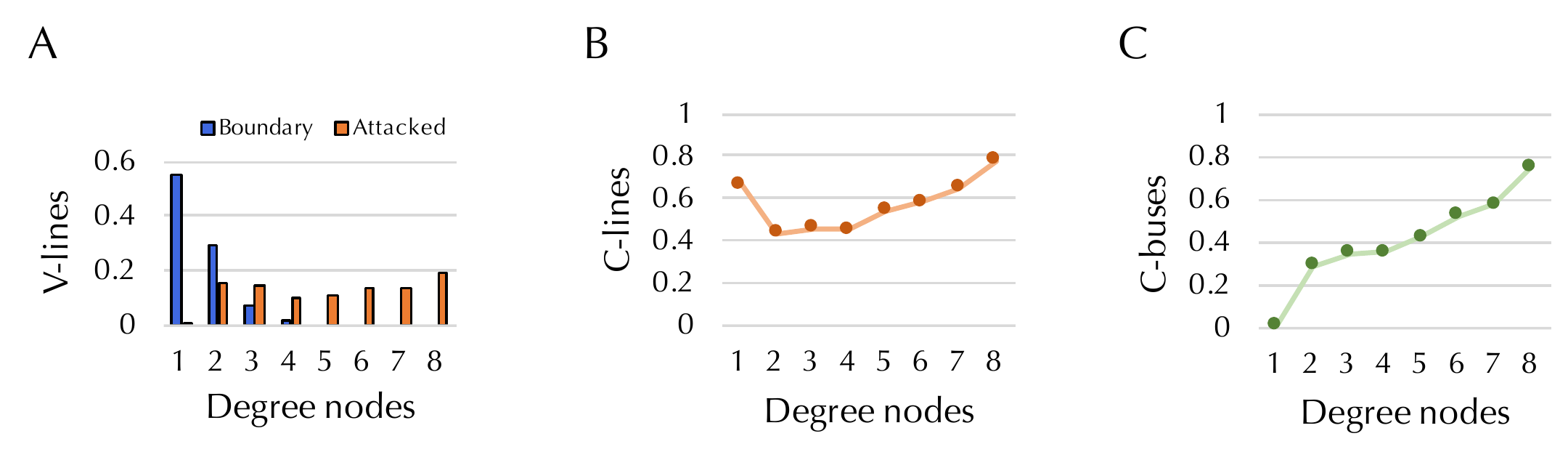}
\caption{\textbf{Characterization of vulnerability through nodal degrees.} \textbf{(A)} Percentage of V-lines when the nodes are at the boundary or in the attacked region. In this case, we distinguish the two directions of a line.  Percentage of \textbf{(B)} C-lines and \textbf{(C)} C-buses averaged over nodes with the same degree. Since the distribution of nodal degrees is light tailed, we group nodes with degree 8 or higher to the same bin.}
\label{fig:degree_mi}
\end{figure}

As for the optimization property, our theoretical analysis indicates that the incorporation of SOCs always improves line robustness (Proposition 2S). This can be visually verified in Figure \ref{fig:inco_US_lp_socp}. This can be also observed in Figure \ref{fig:meas_profile_mi} for different measurement profiles. 

\section*{Conclusion}

Our vulnerability analysis of power system state estimation is distinguished from previous works by its scalability but also by (i) robust two-step convex formulation of the nonconvex nonlinear problem; (ii) strong formal guarantees of boundary defense against cyber attacks; and (iii) localized vulnerability assessment that accounts for network and optimization properties. This study provides a set of notions and tools---the definition of vulnerability index, the boundary defense mechanism, and the analysis of topological and optimization relations to vulnerability---that are applicable to graph-structured data.

Our analysis is based on the assumption that the amount of data is not too low---an assumption far from being restrictive, as we show that with the right set of measurements, one can identify the true state of the system with only one more sensor per bus on average compared to the classical setting of power flow that is known to have multiple spurious local minima. More importantly, the emerging scenario of ``abundant but untrusted data'' considered in this study is more practically realistic and algorithmically challenging than the traditional scenario of ``redundant and reliable data.'' We proposed a robust two-step algorithm to simultaneously perform bad data detection and state estimation. We showed how the number and locations of sensors affect the robustness of state estimation to bad data. A well-chosen set of measurements is able to significantly improve bad data robustness and estimation accuracy without increasing the sensing budget.

We also proposed a boundary defense mechanism to defend against cyber attacks. When a subregion of the network is under attack, it becomes unrealistic to reliably recover the state within the region. By attacking locally, the adversary hopes that due to miscalculation, the error will propagate throughout the grid. However, under some mild conditions, our result shows that this will never occur using the proposed mathematical technique---we can detect the boundary of the attack region and remove the compromised data. Furthermore, this formal condition can be quantified and visualized on the map, leading to a system-wide vulnerability map to facilitate security assessment. 

Based on the proposed mathematical framework, our analysis revealed several key factors that can affect the robustness of the network. A highly connected node is able to defend against attacks if it happens to lie on the boundary, but it is also more prone to attacks with higher collateral damage. 
For a given topological structure, the inclusion of nodal power injection data can weaken the defense; by contrast, the inclusion of voltage magnitude or branch power flow measurements can enhance the robustness against bad data, which gives rise to a higher bad data detection accuracy. From an algorithmic perspective, the incorporation of second-order cone constraints is theoretically shown to be beneficial for network robustness, which is also validated through extensive experiments. Our analysis offers a scientific foundation for vulnerability-based resource allocation, which in the case of a power grid would be based on prioritizing upgrades of sensing infrastructure for critical locations.

\section*{Method summary}

The power grid is modeled as a network of buses connected by transmission lines, where each bus is associated with a complex voltage phasor as the state. Given the topology and measurement profile, some linear basis variables can be constructed for each bus and branch adaptively---if there are no branch measurements and nodal power injections on the connected buses, then the corresponding branch variables can be ignored. This ensures sparsity of the basis. From the measurements, we first estimate the linear basis using a quadratic programming or second-order cone programming. Bad data detection is performed by thresholding the estimated bad data vector. Then, we rerun the estimation on the sanitized dataset, whose results are fed into the second step in the pipeline to produce a state estimation.

We considered two types of attacks. The first attack is ``scattered attack'' (Figures \ref{fig:BDD} and \ref{fig:acc_f1_redundancy_socp}), where a random subset of lines are chosen whose measurements are all corrupted. In this case, the bad data are scattered throughout the network, and the goal is to correctly recover the overall system state. The second attack is ``zonal attack'' (Figure \ref{fig:attack_sim_texas}), where all measurements within a zone---usually governed by a single utility---are corrupted. In this case, the goal is to identify the boundary of the attack and correctly recover the state outside the attacked zone. For stealthy attack, there is a problem of symmetry, namely, without additional information, it is impossible to decide which zone is under attack, since the only inconsistency is observed at the boundary. To avoid this case, we arbitrarily break the symmetry by introducing some sensors within the attacked zone that are more secure than others, such that their values cannot be modified. We can also perform posterior inference based on our prior knowledge of which zones are more likely to be secured than others.

The vulnerability analysis is based on the partition of measurements and variables into attacked and boundary categories (Figure \ref{fig:bound_def}). The vulnerability index is defined by a min-max problem, which is NP-hard in general. For small-scale problems, we developed an efficient enumeration strategy that scales exponentially by the number of bad measurements. For large-scale instance, we proposed two reformulations of the problem, namely linear complimentarity problem and mixed-integer programming, which can be employed to solve the problem efficiently. The critical index for buses (Figure \ref{fig:ci_US_lp_socp}) is obtained by counting the size of the subgraph rooted at the substation and linked by a directional edge that is vulnerable. A critical line is identified when any one of the adjacent lines pointing outwards is vulnerable.

The formal result of boundary defense mechanism is established through a series of propositions and lemmas. The key steps include (1) a ``glueable property,'' which shows that local property of vulnerability implies global property (Lemma 2S and 5S), (2) a result that establishes that boundary defense can stop error from propagation (Lemma 1S and 3S), and (3) a statistical analysis of the first step algorithm based on concentration bounds and a primal-dual witness argument. Further details on the linear representation, two-step pipeline algorithm, theoretical analysis, and experimental setup are given in the supplementary materials.

\section*{Author contributions statement}

M.J. and J.L. developed the idea. M.J. developed the theoretical formalism. M.J., J.L., R.B. and S.S. designed the experiments and interpreted the results. M.J. performed the experiments, analyzed the data and prepared the manuscript. S.S., J.L. and R.B. revised the manuscript.

\bibliographystyle{abbrv}
\bibliography{myref}

\newpage

\appendix

\section*{Supplementary Material}

This supplementary material includes formal theory and additional experimental details for the paper ``Boundary Defense against Cyber Threat for Power System Operation.'' The manuscript is organized as follows. We first discuss the preliminaries in Section \ref{sec:prelim-sec}, including notations, power system modeling, the proposed linear basis of representation, and the measurement model considered in the study. We introduce the two-step pipeline of state estimation in Section \ref{sec:estimator}, where we discuss the algorithms with and without the second-order cone constraints and their connection to robust statistics. Section \ref{sec:boundary} introduces the boundary defense mechanism, including the main results for boundary defense (Lemmas \ref{lemma:bound_defense} and \ref{lem:bound_defend_socp}), implications of local property for global property (Lemmas \ref{lemma:local_sufficient_cond} and \ref{lemma:local_sufficient_cond_socp}), and performance guarantees for estimation accuracy and bad data detection (Theorems \ref{thm:l1_est_cyber}, \ref{thm:lasso_est_cyber}, \ref{thm:l1_est_cyber_socp} and \ref{thm:lasso_est_cyber_socp}). The proofs of the main theorems are delegated to Section \ref{sec:proofs}. Experimental details and additional figures are shown in Section \ref{sec:add_exp}.

\section{Preliminaries}
\label{sec:prelim-sec}

\subsection{Notations}
\label{sec:prelim}
Vectors are shown by bold letters, and matrices are shown by bold and capital letters. Let $x_i$ denote the $i$-th element of  vector $\xbf$. We use $\Rbb$ and $\Cbb$ as the sets of real and complex numbers, and $\mathbb{S}^n$ and $\mathbb{H}^n$ to represent the spaces of $n\times n$ real symmetric matrices and $n\times n$ complex Hermitian matrices, respectively. A set of indices $\{1,2,...,m\}$ is denoted by $[m]$. The cardinality $|\mathcal{J}|$ of a set $\Jcal$ is the number of elements in a set. The support $\supp({\xbf})$ of a vector $\xbf$ is the set of indices of the nonzero entries of $\xbf$. For a set $\mathcal{J}\subset [m]$, we use $\mathcal{J}^c=[m]\setminus \mathcal{J}$ to denote its complement. The symbols $(\cdot)^\top$ and $(\cdot)^*$ represent the transpose and conjugate transpose operators. We use $\Re(\cdot)$, $\Im(\cdot)$ and $\tr{\cdot}$ to denote the real part, imaginary part and trace of a scalar/matrix. The imaginary unit is denoted as $\irm$. The notations $\angle x$ and $|x|$  indicate the angle and magnitude of a complex scalar. For a convex function $g(\xbf)$, we use $\nabla g(\xbf)$ and $\partial g(\xbf)$ to denote its gradient and subgradient at $\xbf$, respectively. We use $\lambda_{\text{min}}(\Abf)$ to denote the smallest eigenvalue of $\Abf$, and $\Abf\succeq 0$ to indicate that $\Abf$ is a positive semidefinite matrix. Let $\Ibf^{(n)}$ denote the identity matrix of dimension $n$, but sometimes for simplicity, we omit the superscript whenever the dimension is clear from the context. The notations $\|\xbf\|_0$, $\|\xbf\|_1$, $\|\xbf\|_2$ and $\|\xbf\|_{\infty}$ show the cardinality, 1-norm, 2-form and $\infty$-norm of $\xbf$. We use $\|\cdot\|_\infty$ to denote the matrix infinity norm (i.e., the maximum absolute column sum of the matrix). Note that the notations $p$ and $q$ are used for active power and reactive power, respectively.

\subsection{Power system modeling}
\label{sec:modeling}
We model the electric grid as a graph $\mathcal{G}\coloneqq\{\mathcal{N},\mathcal{L}\}$, where $\mathcal{N}\coloneqq[n_b]$ and $\mathcal{L}\coloneqq[n_l]$ represent its sets of buses and branches. Each branch $\ell\in\mathcal{L}$ that connects bus $f$ and bus $t$ is characterized  by the branch admittance $y_{\ell}=g_{\ell}+\irm b_\ell$ and the shunt admittance  $y^{\text{sh}}_{\ell}=g^{\text{sh}}_{\ell}+\irm b^{\text{sh}}_\ell$, where $g_{\ell}$ (resp., $g^{\text{sh}}_{\ell}$) and $b_\ell$ (resp., $b^{\text{sh}}_{\ell}$) denote the (shunt) conductance and susceptance, respectively. Typically, $g^{\text{sh}}_{\ell}\ll b^{\text{sh}}_{\ell}$, so it is set to zero in the subsequent description. In addition, to avoid duplicate definition, each line $\ell=(i,j)$ is defined with a direction from bus $i$ (i.e., \textit{from} end, given by $f(\ell)=i$) to bus $j$ (i.e., \textit{to} end, given by $t(\ell)=j$). We also use $\{i,j\}_\ell$ or simply $\{i,j\}$ to denote a line $\ell$ that connects nodes $i$ and $j$.

The power system state is described by the complex voltage at each bus $\vbf=\begin{bmatrix}
v_1,...,v_{n_b}
\end{bmatrix}^\top  \in\mathbb{C}^{n_b}$, where $v_k\in\mathbb{C}$ is the complex voltage at bus $k\in\mathcal{N}$ with magnitude $|v_k|$ and phase $\theta_k\coloneqq\angle v_k$.  Given the complex voltages, by Ohm's law, the complex current injected into line $\{k,j\}_\ell$ at bus $k$ is given by:
\begin{equation*}
    i_{kj} = y_\ell(v_k-v_j)+\frac{\irm}{2}b^{\text{sh}}_\ell v_k.
\end{equation*}
By defining $\theta_{kj}\coloneqq \theta_k-\theta_j$, one can write the power flow from bus $k$ to bus $j$ as
\begin{align*}
    p^{(\ell)}_{kj}&=|v_k|^2g_\ell-|v_k||v_j|(g_\ell\cos\theta_{kj}-b_\ell \sin\theta_{kj}),\\
    q^{(\ell)}_{kj}&=-|v_k|^2(b_\ell+\tfrac{1}{2}b^{\text{sh}}_\ell)+|v_k||v_j|(b_\ell\cos\theta_{kj}-g_\ell \sin\theta_{kj}),
\end{align*}
and active (reactive) power injections at bust $f$ as
\begin{equation}
    p_k=\sum_{\{k,j\}_\ell}p^{(\ell)}_{kj},\qquad q_k=\sum_{\{k,j\}_\ell}q^{(\ell)}_{kj}.
    \label{equ:nodal_inject_equ}
\end{equation}
 The above formulas are based on polar coordinates of complex voltages, where measurements are nonlinear functions of voltage magnitudes and phases. Another popular representation is based on rectangular coordinates of complex numbers, where measurements are expressed as quadratic functions of the real and imaginary parts of voltages (see \cite[Chap. 1]{bienstock2015electrical} for more details). We use ``PV bus'' and ``PQ bus'' to denote buses with real power injection and voltage magnitudes, and buses with real and reactive power injection measurements, respectively.

\subsection{Linear basis of representation}

We introduce a new basis of representation, where  measurements can be expressed as \emph{linear combinations} of the quantities derived form bus voltages. Specifically, for a given system $\mathcal{G}$, we introduce two groups of variables:
\begin{enumerate}
    \item voltage  magnitude square, $x_k^{\text{mg}}\coloneqq|v_k|^2$, for each bus $k\in\mathcal{N}$, and
    \item real and imaginary parts of complex products, denoted as $x^{\text{re}}_\ell\coloneqq\Re({v_iv_j^*})$ and $x^{\text{im}}_\ell\coloneqq\Im({v_iv_j^*})$, respectively, for each line $\ell=(i,j)$. Note that there is only one set of variables $\{x^{\text{re}}_\ell,x^{\text{im}}_\ell\}$ for each line.
\end{enumerate}

Using this representation, we can  derive  various types of power and voltage measurements as follows:
 \begin{itemize} 
\item 
\textit{Voltage magnitude square.} The voltage square magnitude {square} at  bus $k\in\mathcal{N}$ is simply $x_k^{\text{mg}}$ by definition;
\item
\textit{Branch power flows.} For each line $\ell=(i,j)$, the real and reactive power flows from bus $i$ to bus $j$ and in the reverse direction are given by:
\begin{align*}
p_{ij}^{(\ell)} &=g_{{\ell}}x_i^{\text{mg}}-g_{\ell} x^{\text{re}}_\ell-b_{\ell} x^{\text{im}}_\ell\\
q_{ij}^{(\ell)} &=-(b_{{\ell}}+\tfrac{1}{2}b^{\text{sh}}_\ell)x_i^{\text{mg}}+b_{\ell} x^{\text{re}}_\ell-g_{\ell} x^{\text{im}}_\ell\\
p_{ji}^{(\ell)} &=g_{{\ell}}x_j^{\text{mg}}-g_{\ell} x^{\text{re}}_\ell+b_{\ell} x^{\text{im}}_\ell\\
q_{ji}^{(\ell)} &=-(b_{{\ell}}+\tfrac{1}{2}b^{\text{sh}}_\ell)x_j^{\text{mg}}+b_{\ell} x^{\text{re}}_\ell+g_{\ell} x^{\text{im}}_\ell
\end{align*}
\item
\textit{Nodal power injection.} The power injection at bus node $k$ consists of real and reactive powers, i.e. $p_k+\irm q_k$, where:
\begin{align*}
p_k &= \sum_{k\in\ell}g_{{\ell}}x_k^{\text{mg}}-\sum_{k\in\ell}g_{{\ell}} x^{\text{re}}_\ell-(\sum_{f(\ell)=k}b_{\ell}-\sum_{t(\ell)=k}b_{\ell}) x^{\text{im}}_\ell\\
q_k &= -(\sum_{k\in\ell}b_{{\ell}}+\tfrac{1}{2}b^{\text{sh}}_{{\ell}})x_k^{\text{mg}}+\sum_{k\in\ell}b_{{\ell}} x^{\text{re}}_\ell-(\sum_{f(\ell)=k}g_{\ell}-\sum_{t(\ell)=k}g_{\ell}) x^{\text{im}}_\ell,
\end{align*}
where $\sum_{k\in\ell}$ is the sum over all lines $\ell\in\mathcal{L}$ that are connected to $k$, $\sum_{f(\ell)=k}$ is the sum over all lines $\ell$ where $f(\ell)=k$, and similarly, $\sum_{t(\ell)=k}$ is the sum over all lines $\ell$ where $t(\ell)=k$. Equivalently, we can use \eqref{equ:nodal_inject_equ} to combine the branch power flows defined above.
 \end{itemize} 
Thus, each customary  measurement in power systems that belongs to one of the above \emph{measurement types} can be represented by a linear function\footnote{It is straightforward to include linear PMU measurements in our analysis as well using the relation $\tan\theta_{ij}={x^{\text{im}}_{\ell}}/{x^{\text{re}}_{\ell}}$ for each line $\ell=(i,j)$. Thus, as long as we have two adjacent PMU measurements, we can use the phase difference to construct a linear measurement equation ${x^{\text{im}}_{\ell}}-\tan\theta_{ij}{x^{\text{re}}_{\ell}}=0$.}:
\begin{equation}
m_i(\xbf)=\abf_i^\top \xbf_\natural,
\label{equ:measurement}
\end{equation}
where $\abf_i\in\Rbb^{n_x}$ is the vector for the $i$-th noiseless measurement and $\xbf_\natural=(\{x^{\text{mg}}_{k}\}_{k\in\mathcal{N}},\{x^{\text{im}}_{\ell},x^{\text{im}}_{\ell}\}_{\ell\in\mathcal{L}})$ is the regression vector. By collecting all the sensor measurements in a vector $\mbf\in\Rbb^{n_m}$, we have
\begin{equation}
    \mbf=\Abf\xbf_\natural,
    \label{equ:sensing_equ}
\end{equation}
where $\Abf\in\Rbb^{n_m\times n_x}$ is the sensing matrix with rows $\abf_i^\top$ for $i\in [n_m]$. 



\subsection{Measurement model}

To perform SE, the supervisory control and data acquisition (SCADA) system collects measurements about power flows and complex voltages at key locations instrumented with sensors. This process is subject to both ubiquitous sensor noise and randomly occurring sensor faults. We consider the measurement model as follows:
\begin{equation}
\ybf=\Abf\xbf_\natural+\mathbf{w}_\natural+\bbf_\natural,
\label{equ:mv}
\end{equation}
where $\Abf\in\Rbb^{n_m\times n_x}$ and $\xbf_\natural\in\Rbb^{n_x}$ are the sensing matrix and the true regression vector in \eqref{equ:sensing_equ}, $\mathbf{w}_\natural\in\Rbb^{n_m}$ denotes random noise, and  $\mathbf{b}_\natural\in\Rbb^{m}$ is the bad data error that accounts for sensor failures or adversarial noise \cite{jin2019acfdia}. Note that $\xbf$ serves as an intermediate parameter and the end goal is to find $\vbf$.

Because the sensor data are of different types and their corresponding measurements could be of different scales, we introduce the following condition.
\begin{definition}[Measurement normalization convention]
\label{def:norm_measure}
Each row of $\Abf$ corresponding to a voltage magnitude measurement is normalized by the degree of connection of the node $k$, $\|\abf_i\|_2^2= \mathrm{deg}(k)$, and 1 otherwise $\|\abf_i\|_2^2= 1$, where $\abf_i$ is the $i$-th row of $\Abf$. The only exception is when the line vulnerability (c.f., Def. \ref{def:line_vul_metric}) is calculated, when all the measurements are normalized by 1.
\end{definition}

This condition is straightforward to implement in practice, since the sensing matrix $\Abf$ is fixed for a given set of measurements. This is also known as preconditioning, which assists with the statistical performance of regression.

\section{Two-step pipeline of state estimation}
\label{sec:estimator}
This section describes the proposed two-step state estimation method. For the first step, we develop algorithms in two categories, which differ by whether or not the second-order cone constraints are incorporated. Within each category, we also propose two slight variations, which differ by whether the term of squared loss is included. For the second step, we propose two approaches based on quadratic programming.

\subsection{Step 1: Estimation of $\xbf_\natural$}
\label{sec:stage1}

In the first step, the goal is to estimate $\xbf_\natural$ from a set of noisy and corrupted measurements $\ybf$. We consider two cases separately. In the first case, the dense noise is negligible, i.e., $\wbf_\natural=\bfzero$, and we only need to consider the sparse measurement corruption $\bbf$. 

\subsection*{Case 1: Sparse corruption but no dense noise (i.e., $\wbf=\bfzero$)}

In this case, the measurements are given by $\ybf=\Abf\xbf_\natural+\bbf_\natural$. To estimate $\xbf_\natural$, we solve the following program:
\begin{equation}
    \min_{\bbf\in\Rbb^{n_m},\xbf\in\Rbb^{n_x}}\|\bbf\|_1,\quad \st\quad \Abf\xbf+\bbf=\ybf.
    \tag{S$^{(1)}$: $\ell_1$}
    \label{equ:primal_l1}
\end{equation}
Briefly, under some mild conditions on observability and robusteness to be specified in Section \ref{sec:boundary}, we can faithfully recover $\bbf_\natural$ from the above program. As a consequence, $\xbf_\natural$ can be obtained by performing regression using the remaining good data.

For this case, we can also incorporate second-order cone (SOC) constraints:
\begin{equation}
    \min_{\bbf\in\Rbb^{n_m},\xbf\in\Rbb^{n_x}}\|\bbf\|_1,\quad \st\quad \Abf\xbf+\bbf=\ybf, \quad\xbf\in\Kcal,
    \tag{S$^{(1)}$: $\ell_1$-$\Kcal$}
    \label{equ:primal_socp}
\end{equation}
where 
\begin{equation}
    \Kcal=\left\lbrace\xbf\in\Rbb^{n_x}\Big\vert\begin{bmatrix}x^{\text{mg}}_i&x^{\text{re}}_\ell+jx^{\text{im}}_\ell\\
x^{\text{re}}_\ell-jx^{\text{im}}_\ell&x^{\text{mg}}_j\end{bmatrix}\succeq 0,\quad\forall\ell\coloneqq(i,j)\in\Lcal\right\rbrace.
\label{equ:Kdef}
\end{equation}
Let $\sigma(x)$ denote the index of the variable $x$ (e.g., $x^{\text{mg}}_i,x^{\text{re}}_\ell,x^{\text{im}}_\ell$) in the vector $\xbf$. For instance, $\sigma(x^{\text{mg}}_i)$ denotes the index of $x^{\text{mg}}_i$ in $\xbf$. The SOC constraint can be equivalently written as:
\begin{equation}
    \cbf_\ell^\top\xbf\geq\left\|\Dbf_\ell\xbf\right\|_2 \quad\Leftrightarrow\quad \begin{bmatrix}\Dbf_\ell\\\cbf_\ell^\top\end{bmatrix}\xbf\in\Ccal_5,
    \label{equ:c_l_D_l}
\end{equation}
where $\cbf_\ell\in\Rbb^{n_x}$ has its $\sigma(x^{\text{mg}}_i)$ and $\sigma(x^{\text{mg}}_j)$ entries to be $\tfrac{1}{\sqrt{2}}$ and 0 elsewhere, and $\Dbf_\ell\in\Rbb^{4\times n_x}$ has its $(1,\sigma(x^{\text{mg}}_i))$ and $(2,\sigma(x^{\text{mg}}_j))$ entries to be $\tfrac{1}{\sqrt{2}}$ and its $(3,\sigma(x^{\text{re}}_\ell))$ and $(4,\sigma(x^{\text{im}}_\ell))$ entries to be 1, and 0 elsewhere, and $\Ccal_5$ denotes the second-order cone of dimension 5.

The problem \eqref{equ:primal_socp} can be reformulated as:
\begin{equation}
    \min_{\bbf\in\Rbb^{n_m},\xbf\in\Rbb^{n_x}}\|\bbf\|_1,\quad \st\quad \Abf\xbf+\bbf=\ybf, \begin{bmatrix}\Dbf_\ell\\\cbf_\ell^\top\end{bmatrix}\xbf\in\Ccal_5,\forall \ell\in\Lcal
    \label{equ:primal_socp2}
\end{equation}
using standard SOCP notations. The Lagrangian is given by:
\begin{align*}
    L\left(\xbf,\bbf,\{\nu_\ell,\bmu_\ell\}_{\ell\in\Lcal},\hbf\right)&=\|\bbf\|_1+\hbf^\top\left(\ybf-\Abf\xbf-\bbf\right)-\sum_{\ell\in\Lcal}\left(\nu_\ell\cbf_\ell^\top\xbf+\bmu_\ell\Dbf_\ell\xbf\right)
\end{align*}
The Karush-Kuhn-Tucker (KKT) conditions are given by:
\begin{align}
&\text{(primal feasibility) }\;\qquad\qquad \Abf\xbf+\bbf=\ybf,\quad\cbf_\ell^\top\xbf\geq\left\|\Dbf_\ell\xbf\right\|_2,\quad\forall\ell\in\Lcal\\
    &\text{(dual feasibility) } \;\quad\qquad\qquad\nu_\ell\geq\|\bmu_\ell\|_2,\quad\forall\ell\in\Lcal\\
    &\text{(stationarity) } \;\;\;\quad\qquad\qquad\quad-\sum_{\ell\in\Lcal}(\nu_\ell\cbf_\ell+\Dbf_\ell^\top\bmu_\ell)=\Abf^\top\hbf,\quad\hbf\in\partial\|\bbf\|_1\\
    &\text{(complementary slackness) } \quad\nu_\ell\cbf_\ell^\top\xbf+\bmu_\ell^\top\Dbf_\ell\xbf=\bfzero,\quad\forall\ell\in\Lcal.
\end{align}
Therefore, the dual program of \eqref{equ:primal_socp} is given by:
\begin{subequations}
    \label{equ:dual_socp_l1}
	\begin{align}
    & \underset{\hbf\in\Rbb^{n_m},\{\nu_\ell,\bmu_\ell\}_{\ell\in\Lcal}
	}{\text{max~~~}}\hspace{1.5cm}\hbf^\top\ybf
	&&\\
	& \qquad\text{subject to~~~}
	&&\hspace{-4.5cm}-\sum_{\ell\in\Lcal}(\nu_\ell\cbf_\ell+\Dbf_\ell^\top\bmu_\ell)=\Abf^\top\hbf\\
	&&&\hspace{-4.5cm} \|\hbf\|_\infty\leq 1\\
	&&&\hspace{-4.5cm} \nu_\ell\geq\|\bmu_\ell\|_2,\quad\forall \ell\in\Lcal
	\end{align}
\end{subequations}

\subsection*{Case 2: Sparse corruption and dense noise}

In this case, the dense noise cannot be ignored, and the measurements are given by \eqref{equ:measurement}. We perform the estimation by solving the following mixed-objective optimization:
\begin{equation}
    \min_{\bbf\in\Rbb^{n_m},\xbf\in\Rbb^{n_x}}\tfrac{1}{2n_m}\|\ybf-\Abf\xbf-\bbf\|_2^2+\lambda\|\bbf\|_1,
    \tag{S$^{(1)}$: $\ell_2\ell_1$}
    \label{equ:primal_noise}
\end{equation}
where $\lambda>0$ is the regularization coefficient. Due to the existence of dense noise, it is no longer possible to exactly recover the true $\xbf_\natural$; however, if the magnitude of each dense noise is small, then we can still have strong statistical bounds on the estimation error.

We can also incorporate second-order cone constraints:
\begin{equation}
    \min_{\bbf\in\Rbb^{n_m},\xbf\in\Rbb^{n_x}}\tfrac{1}{2n_m}\|\ybf-\Abf\xbf-\bbf\|_2^2+\lambda\|\bbf\|_1,\quad \st\quad \xbf\in\Kcal,
    \tag{S$^{(1)}$: $\ell_2\ell_1$-$\Kcal$}
    \label{equ:primal_lasso_socp}
\end{equation}
where $\Kcal$ is defined in \eqref{equ:Kdef}. The Lagrangian of \eqref{equ:primal_lasso_socp} is given by:
\begin{align*}
    L\left(\xbf,\bbf,\{\bmu_\ell\}_{\ell\in\Lcal},\{\nu_\ell\}_{\ell\in\Lcal},\hbf\right)&=\tfrac{1}{2n_m}\|\ybf-\Abf\xbf-\bbf\|_2^2+\lambda\|\bbf\|_1-\sum_{\ell\in\Lcal}\left(\nu_\ell\cbf_\ell^\top\xbf+\bmu_\ell\Dbf_\ell\xbf\right)
\end{align*}
The KKT conditions are given by:
\begin{align}
&\text{(primal feasibility) }\;\qquad\qquad \cbf_\ell^\top\xbf\geq\left\|\Dbf_\ell\xbf\right\|_2,\quad\forall\ell\in\Lcal\\
    &\text{(dual feasibility) } \;\quad\qquad\qquad\nu_\ell\geq\|\bmu_\ell\|_2,\quad\forall\ell\in\Lcal\\
    &\text{(stationarity) } \;\;\;\quad\qquad\qquad\quad\frac{1}{n_m}\Abf^\top(\ybf-\Abf\xbf-\bbf)+\sum_{\ell\in\Lcal}(\nu_\ell\cbf_\ell+\Dbf_\ell^\top\bmu_\ell)=\bfzero\\
    &\qquad\qquad\;\;\;\;\;\;\;\;\quad\qquad\qquad\quad\frac{1}{n_m}(\ybf-\Abf\xbf-\bbf)=\lambda\hbf,\quad\hbf\in\partial\|\bbf\|_1\\
    &\text{(complementary slackness) } \quad\nu_\ell\cbf_\ell^\top\xbf+\bmu_\ell^\top\Dbf_\ell\xbf=\bfzero,\quad\forall\ell\in\Lcal.
\end{align}

The KKT conditions are important for the analysis in Section \ref{sec:boundary}.

\subsection{Connection with robust statistics for bad data detection}

The so-called bad data rejection and state estimation form an important part of power
systems supervisory control and data acquisition.  There are traditional statistical
approaches to bad data rejection that involve iteratively eliminating the measurements
with the largest residual that are obtained from a least squares estimation (see \cite[Section 9.6]{wood2013power}). Such a smooth quadratic  objective can, however, mask bad data by ``spreading" the error  around the system.  An alternative approach developed in \cite{bagchi1994comparison} is to use an $\ell_1$ objective, which can identify 
multiple bad data directly.  However, the resulting estimate does not average out the effect of dense, independent measurement errors. 

The so-called Huber loss that is
quadratic for small measurement residuals but constant or linear for large
measurement residuals has been explored in \cite{mili1996robust,baldick1997implementing,zhao2018statistical}.  The quadratic-linear loss function is convex, continuous and differentiable at the transition between the quadratic and linear part, and is given by \cite{huber2011robust}:
\begin{equation}
    f_{\mathrm{Huber}}(r;\psi)=\begin{cases}
    \frac{1}{2}r^2&|r|\leq\psi\\
    \psi(|r|-\frac{1}{2}\psi)&|r|>\psi\\
    \end{cases},
\end{equation}
where $\psi$ is the hyper-parameter controlling the transition point between the $\ell_2$ and $\ell_1$ loss functions.

There is an interesting connection between \eqref{equ:primal_noise} and the Huber loss. To see this, we can view the optimization over $\bbf$ and $\xbf$ in \eqref{equ:primal_noise} as an inner optimization with $\bbf$ for a given $\xbf$, and an outer optimization with $\xbf$. The inner optimization is composed of a series of smaller optimization problems
\begin{equation}
    \min_{b_i}\tfrac{1}{2n_m}(y_i-\abf_i^\top\xbf-b_i)^2+\psi|b_i|,
\end{equation}
for $i\in[n_m]$, which has the optimal solution 
\begin{equation}
    b_i^*=\sign(y_i-\abf_i^\top\xbf)\max\left(0,\left|y_i-\abf_i^\top\xbf\right|-\psi\right),
\end{equation}
where $\sign(y)$ is the sign of $y$. Now, by defining $r_i\coloneqq y_i-\abf_i^\top\xbf$, we substitute the solution into the outer optimization to obtain
\begin{equation}
    \tfrac{1}{n_m}\sum_{i\in[n_m]}\tfrac{1}{2}(r_i-\sign(r_i)\max\left(0,|r_i|-\psi\right))^2+\psi|\max\left(0,|r_i|-\psi\right)|.
\end{equation}
Hence, it can be seen that the above expression is equal to the Huber loss:
\begin{equation}
    \frac{1}{n_m}\sum_{i\in [n_m]} f_{\mathrm{Huber}}(y_i-\abf_i^\top\xbf;\psi).
\end{equation}
Despite the wide usage of Huber loss in power system estimation, the existing studies in the literature are mostly empirical. The approach proposed here allows for strong mathematical results that go well beyond the  promising empirical results.

\subsection{Step 2: Recovery of $\vbf$}

The goal of the second step is to recover the underlying system voltage $\vbf$ from the  estimation $\hat\xbf$ obtained in Step 1. First, we transform $\hat\xbf$ into estimations of voltage magnitudes and phase differences:
\begin{itemize}
    \item The voltage magnitude at each bus $k\in\mathcal{N}$ can be obtained by $|\hat{v}_k|=\sqrt{\hat{x}^{\text{mg}}_k}$;
    \item The phase difference along each line ${\ell}=(i,j)$ is given by $\hat\theta_{ij}=\arctan{\hat{x}^{\text{im}}_{\ell}}/{\hat{x}^{\text{re}}_{\ell}}$.
\end{itemize}
To obtain the estimations of phases at each bus, we propose two methds. The first method is to solve the least-squares problem
\begin{equation}
    \hat{\bftheta}=\arg\min\limits_{\bftheta\in\Rbb^{n_b}}
    \sum_{\ell=(i,j)}(\theta_i-\theta_j-\hat\theta_{ij})^2,
    \tag{S$^{(2)}$: $\ell_2$}
    \label{equ:s2-phase}
\end{equation}
which has a closed-form solution: let $\bftheta_\Delta$ be a collection of $\hat\theta_{ij}$, and $\Lbf\in\Rbb^{n_l\times n_b}$ be a sparse matrix with $L(\ell,i)\coloneqq1$ and $L(\ell,j)\coloneqq-1$ for each line $\ell=(i,j)$ and zero elsewhere. Then, the solution for \eqref{equ:s2-phase} is given by:
\begin{equation}\label{equ:closed-form}
    \hat{\bftheta}=(\Lbf^\top\Lbf)^{-1}\Lbf^\top\bftheta_\Delta.
\end{equation}

The second approach is to solve a mixed-objective problem, similar to the first step:
\begin{equation}
    \hat{\bftheta}=\arg\min\limits_{\bftheta\in\Rbb^{n_b}}
    \tfrac{1}{n_l}\sum_{\ell=(i,j)}(\theta_i-\theta_j-\hat\theta_{ij})^2+\lambda_2\sum_{\ell=(i,j)}|\theta_i-\theta_j-\hat\theta_{ij}|.
    \tag{S$^{(2)}$: $\ell_2\ell_1$}
    \label{equ:s2-phase-lasso}
\end{equation}
In this case, there is no longer a closed-form solution available, but the advantage is that it is robust to large errors in the phase difference estimation, in case the first step method does not fully detect the bad data in the measurements.

Finally, we can reconstruct $\hat\vbf$ via the formula:
\begin{equation}
    \hat{v}_k=|\hat{v}_k|e^{\irm\hat\theta_k},\qquad k\in\mathcal{N}.
\end{equation}
If the regression vector from Step 1 is exact, i.e., $\hat\xbf=\xbf_\natural$, then we can use \eqref{equ:s2-phase} to accurately recover the system state $\hat\vbf=\vbf$. Even if the $\hat\xbf$ is not exact, the second stage estimator \eqref{equ:s2-phase-lasso} has nice properties to control the estimation error, and therefore any potential error in $\hat\theta_{ij}$ does not propagate along the branches.

\section{Boundary defense mechanism}
\label{sec:boundary}
In this section, we give a detailed discussion of the new notion of defense on networks, called ``boundary defense mechanism.'' For a given attack scenario, we define a natural partition of the network into the attacked, inner and outer boundaries, and safe regions. We describe a fairly general framework, which incorporates a wide range of adversarial scenarios that are localized, including line outage, substation down, and zonal attacks.  For the rest of the analysis, we denote $\xbf_{\sharp}$ and $\bbf_\sharp$ as the ground truth for state $\xbf$ and bad data $\bbf$, as defined in \eqref{equ:mv}.

\begin{definition}[Attacked, boundary, and safe regions]
\label{def:attack_bound_safe}
Let $\NBat$ be the set of nodes under attack and the ``{attacked region}'' $\Bat\coloneqq\{\NBat,\LBat\}$ be the induced subgraph. Let the ``{inner boundary}'' be the set of nodes adjacent to the attacked region $\NBbi\coloneqq\left\lbrace i\in\Ncal\setminus\NBat\;\middle|\; \exists j\in \NBat,\;\;\mathrm{ s.t. }\;\;\{i,j\}\in\Lcal\right\rbrace$ and the induced graph be denoted as $\Bbi$, and the ``{outer boundary}'' be the set of nodes adjacent to the inner boundary region $\NBbo\coloneqq\left\lbrace i\in\Ncal\setminus(\Ncal_\Bcal\cup\NBbi)\;\middle|\; \exists j\in \NBbi,\;\;\mathrm{ s.t. }\;\;\{i,j\}\in\Lcal\right\rbrace$ and the induced graph be denoted as $\Bbo$. Let $\NBbd\coloneqq \NBbi\cup\NBbo$ be  nodes in the ``{boundary region}'' and $\Bbd\coloneqq\{\NBbd,\LBbd\}$ be the induced subgraph. We also denote the set of lines that bridge nodes between $\Bat$ and $\Bbi$ as $\Latbi$, and the set of lines that brige nodes between $\Bbi$ and $\Bbo$ as $\Lbibo$. Lastly, let $\Ncal_\Bsf\coloneqq\Ncal\setminus(\NBat\cup\NBbd)$ be the rest of the nodes and the ``{safe region}'' $\Bsf\coloneqq\{\NBsf,\LBsf\}$ be the induced subgraph.
\end{definition}

When there is an attack on a local region, a subset of the local measurements are compromised. We use $\Bcal=\Bat\cup\Bbi$ to delineate the smallest subgraph to cover this region. For the simplicity of the analysis, we assume that there are no lines connecting two inner boundary nodes in $\Bbi$, and that no two nodes in $\Bat$ are connected to the same node in $\Bbi$ (one can always enlarge the region $\Bcal$ to satisfy these conditions).  Furthermore, we make the assumption that no measurements on the nodes (e.g., voltage magnitudes and nodal injections) or on the lines (e.g., power branch flows) within the boundary region $\Bbd$ are attacked. The partition set notations in Def. \ref{def:attack_bound_safe} are illustrated in Fig. \ref{fig:region_illustrate}. With the set partition notions ready, we introduce a partition of the measurements and variables.

\begin{figure}[h]
  \centering
  \begin{center}
    \includegraphics[width=.6\columnwidth,trim=0mm 0mm 0mm 0mm,clip]{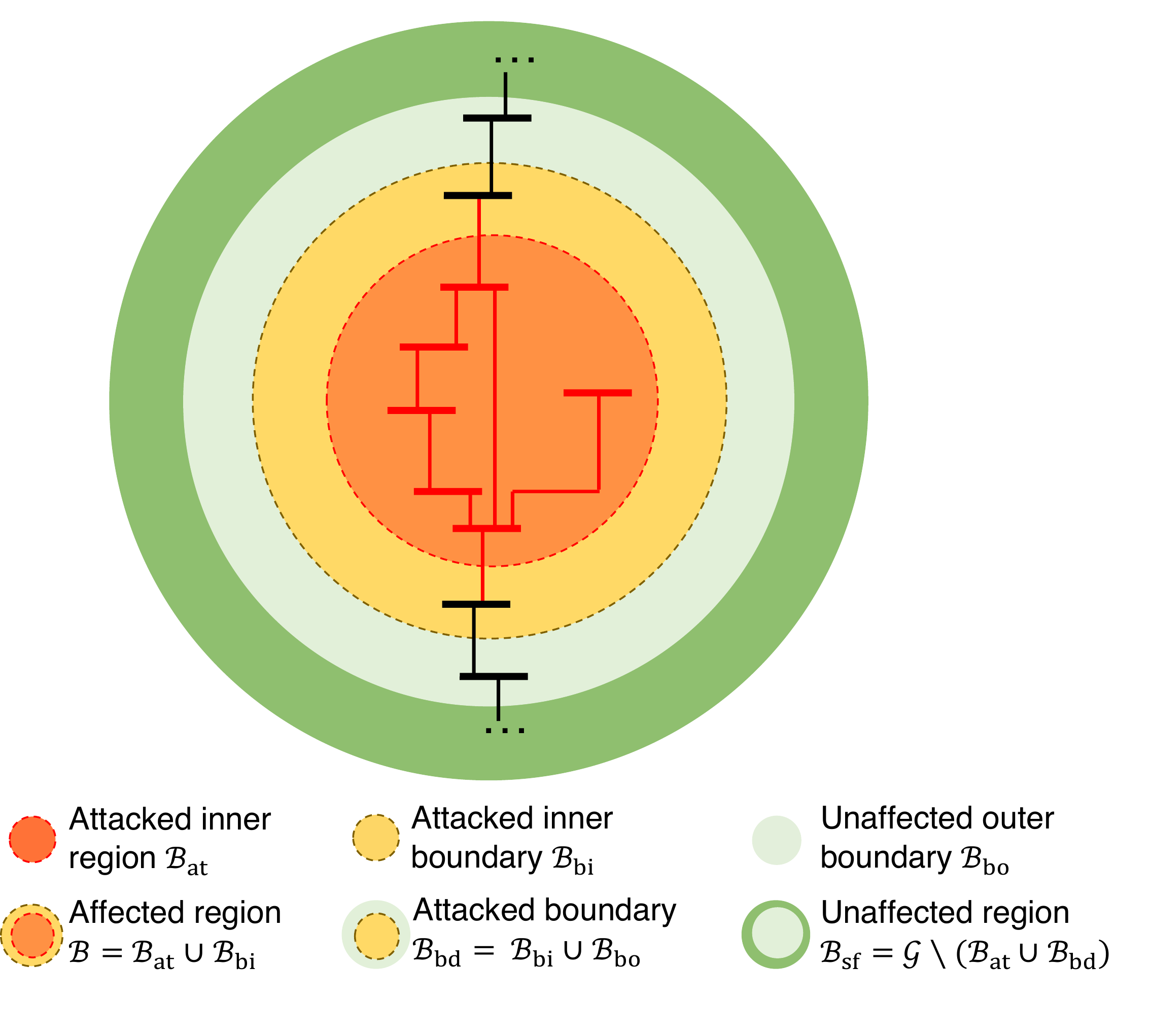}
  \end{center}
  \caption{The illustrations of the partition set concepts introduced for the case of zonal attacks. Lines or buses whose measurements are under attack are shown in red.}
  \label{fig:region_illustrate}
\end{figure}

\begin{definition}[Attacked, boundary and safe variables and measurements]
The set of ``attacked variables'' $\Xat$ includes variables on nodes in $\Bat$ and lines in $\LBat\cup\Latbi$. The set of ``boundary variables'' $\Xbd$ includes variables on nodes in $\Bbd$ and lines in $\LBbd$. The set of ``safe variables'' $\Xsf$ includes all other variables. The set of ``attacked measurements'' $\Mat$ includes measurements on nodes in $\Bat$ and lines in $\LBat$. The set of ``inner boundary measurements'' $\Mbi$ includes nodal power injections in $\Bbi$ and line measurements in $\Latbi$, and the set of ``outer boundary measurements'' $\Mbo$ includes voltage magnitude and line measurements in $\Bbd$. Together, they form the ``boundary measurements'' $\Mbd\coloneqq\Mbi\cup\Mbo$. The rest of the measurements $\Msf$ are ``safe measurements.''
\label{def:bound_var_mea_all}
\end{definition}

By definition, the sets $\Msf$, $\Mbo$, $\Mbi$, $\Mat$ form a partition of $[n_m]$, and the sets ${\Xsf}$, ${\Xbd}$, and $\Xat$ form a partition of $[n_x]$. Thus, we can rearrange and partition the matrix $\Abf$ as follows:
\begin{equation}
    \Abf=\begin{bmatrix}
\Abf_{\Msf,\Xsf}&\Abf_{\Msf,\Xbd}&\Abf_{\Msf,\Xat}\\
\Abf_{\Mbo,\Xsf}&\Abf_{\Mbo,\Xbd}&\Abf_{\Mbo,\Xat}\\
\Abf_{\Mbi,\Xsf}&\Abf_{\Mbi,\Xbd}&\Abf_{\Mbi,\Xat}\\
\Abf_{\Mat,\Xsf}&\Abf_{\Mat,\Xbd}&\Abf_{\Mat,\Xat}
\end{bmatrix}=\begin{bmatrix}
\Abf_{\Msf,\Xsf}&\Abf_{\Msf,\Xbd}&\bfzero\\
\bfzero&\Abf_{\Mbo,\Xbd}&\bfzero\\
\bfzero&\Abf_{\Mbi,\Xbd}&\Abf_{\Mbi,\Xat}\\
\bfzero&\bfzero&\Abf_{\Mat,\Xat}
\end{bmatrix}.
\label{equ:A_struct}
\end{equation}
There is no loss of generality in arranging $\Abf$ as above, which is simply for the purpose of presentation. Let $\Ibf_{\Mat}^{(n_m)}$, $\Ibf_{\Mbi}^{(n_m)}$, $\Ibf_{\Mbo}^{(n_m)}$, and $\Ibf_{\Msf}^{(n_m)}$ be matrices that consist of the ${\Mat}$, ${\Mbi}$, ${\Mbo}$, and ${\Msf}$ rows from the identity matrix of size $n_m$, respectively, and $\Ibf_{\Xat}^{(n_x)}$, $\Ibf_{\Xbd}^{(n_x)}$ and $\Ibf_{\Xsf}^{(n_x)}$ be the matrices that consist of the ${\Xat}$, ${\Xbd}$ and ${\Xsf}$ rows from the identity matrix of size $n_x$. Then, we can obtain each subblock that accounts for a set of measurements (e.g. $\Msf$) and variables (e.g. $\Xsf$) using the equation $\Abf_{\Msf,\Xsf}=\Ibf_{\Msf}^{(n_m)}\Abf\Ibf_{\Xsf}^{(n_x)\top}$ without having to specify a particular order sequence of measurements $\ybf$ or variables $\xbf$,

We introduce the following properties to characterize the sensing matrix $\Abf$. 
\begin{definition}[Lower eigenvalue]
Let $\Qbf_{\Mbd,\Xbd}\coloneqq\begin{bmatrix}\Abf_{\Mbd,\Xbd}&\Ibf_{\Mbi}^{(|\Mbd|)\top}\end{bmatrix}$, where $\Ibf_{\Mbi}^{(|\Mbd|)}$ consists of $\Mbi$ rows of the size--$|\Mbd|$ identity matrix. Then, the lower eigenvalue $\Cmin$ is the lower bound:
    \begin{equation}
 \min\left\lbrace\lambda_{\mathrm{min}}\left(\Qbf_{\Mbd,\Xbd}^\top \Qbf_{\Mbd,\Xbd}\right),\lambda_{\mathrm{min}}\left(\Abf_{\Mbo,\Xbd}^\top \Abf_{\Mbo,\Xbd}\right),\lambda_{\mathrm{min}}\left(\Abf_{\Msf,\Xsf}^\top \Abf_{\Msf,\Xsf}\right)\right\rbrace\geq\Cmin. \label{equ:lower_eig_Q}
    \end{equation}
\end{definition}
The value $C_{\text{min}}$ characterizes the influence of bad data on the identifiability of $\xbf_\natural$ outside the attacked region. If $C_{\text{min}}$ is strictly positive and one can accurately detect the support of bad data on the boundary, then it is possible to obtain a satisfactory estimation of $\xbf_\natural$ outside the attacked region. 

The next property turns out to be critical for bad data support recovery.

\begin{definition}[Global mutual incoherence]
Let $\Jcal$ denote the support of bad data, and let the pseudoinverse of $\Abf_{\Jcal^c}$ be $\Abf^{+}_{\Jcal^c}=(\Abf_{\Jcal^c}^\top \Abf_{\Jcal^c})^{-1}\Abf_{\Jcal^c}^\top$. Then, the mutual incoherence parameter $\rho(k_b)$ is given by:
\begin{equation}
\rho(\Jcal)=\|\Abf_{\Jcal^c}^{\top+}\Abf_\Jcal^\top\|_{\infty}.
\end{equation}
\label{def:mutual_incoherence}
\end{definition}
The name ``mutual incoherence'' originates from the compressed sensing literature \cite{fuchs2005recovery,tropp2006just,zhao2006model,wainwright2009sharp}. The proposed mutual incoherence definition is not the same as any of the existing mutual incoherence conditions. Intuitively, it measures the alignment of the sensing directions of the corrupted measurements (i.e., $\Abf_\Jcal$) with those of the clean data (i.e., $\Abf_{\Jcal^c}$). If these directions are misaligned (a.k.a., incoherent), then the value $\rho(\Jcal)$ is low, and it is likely to uncover the support of bad data. In general, the less bad data exist, the more likely that $\rho(\Jcal)$ will be small. However, the main drawback of this metric is that it depends on each instance of the bad data support $\Jcal$, and therefore it cannot be used as a robustness metric in a general sense. Moreover, it turns out that this metric is more conservative than the vulnerability index to be discussed next (see Proposition \ref{prop:relate_mi}).


\subsection{Vulnerability index and boundary defense for linear/quadratic programming}
Our goal is to find the attacked region by detecting a sufficiently large number of measurements within $\Mat$ while avoiding making false positive detection for measurements belonging to the unaffected region. In other words, if $\hat{\Jcal}\coloneqq\supp(\hat{\bbf})$ denotes the support of the estimated bad data, then it is desirable to have $\hat{\Jcal}\subseteq\Mat\cup\Mbi$ (here, we relax the condition that $\hat{\Jcal}\subseteq\Jcal$ and allow both false positives and false negatives within the attacked region). The following lemma establishes a key result for the estimation without SOCs.

\begin{lemma}[Boundary defense stops error propagation]
Suppose tthat here is no dense measurement noise (i.e., $\wbf=\bfzero$), and the bad data are confined within $\Mat$, i.e., $\supp(\bbf_{\natural})\subseteq \Mat$. Also, suppose that $\Abf_{\Msf\cup\Mbd,\Xsf\cup\Xbd}$ has full column rank. If for an arbitrary $\bbf_{\star\Mbd}$ with support limited to the inner boundary, i.e., $\supp(\bbf_{\star\Mbd})\subseteq\Mbi$, the solution $\hat{\xbf}_{\Ibd}\in\Xbd$ to the program
\begin{equation}
    \min_{{\xbf}_{\Ibd}}\|\zbf_{\Mbd}-\Abf_{\Mbd,\Xbd}\xbf_{\Ibd}\|_1       \label{equ:oracle_B_bd}
\end{equation}
is unique and satisfies the properties $\hat{\xbf}_{\Ibd}={\xbf}_{\natural \Ibd}$, where $\zbf_{\Mbd}=\Abf_{\Mbd,\Xbd}\xbf_{\natural \Ibd}+\bbf_{\star\Mbd}$, then the solution $\hat{\xbf}$ to \eqref{equ:primal_l1} satisfies the properties $\hat{\xbf}_{\Ibd}={\xbf}_{\natural \Ibd}$ and $\hat{\xbf}_{\Isf}={\xbf}_{\natural \Isf}$.
\label{lemma:bound_defense}
\end{lemma}
To sketch the proof, since by assumption the unique optimal solution for the measurement-sensing matrix pair $(\ybf_\Msf,\Abf_{\Msf,\Xsf\cup\Xbd})$ given ${\xbf}_{{\natural\Ibd}}$ recovers the ground truth ${\xbf}_{{\natural\Isf}}$, we aim at showing that the unique optimal solution of $(\ybf_{\Mbd\cup\Mat},\Abf_{\Mbd\cup\Mat,\Xbd\cup\Xat})$ corresponding to the boundary state coincides with ${\xbf}_{{\natural\Ibd}}$, which completes the proof because this set of measurements is independent of the states $\xbf_{\Isf}$. This achieves a de facto coupling of the ``weakly coupled'' system due to the overlapping regions corresponding to measurements $\ybf_{\Mbd}$.
\begin{proof}
There are two ways to prove the statement. The first one relies on logical reasoning that is intuitive, while the second approach is based on KKT conditions that can be easily generalized to measurements with dense noise. We start with the first approach, which partitions the loss function in \eqref{equ:primal_l1} into the sum of three terms:
\begin{align*}
    f_1({\xbf}_{\Isf},{\xbf}_{\Ibd})&=\|\ybf_{\Msf}-\Abf_{\Msf,\Xsf}\xbf_{\Isf}-\Abf_{\Msf,\Xbd}\xbf_{\Ibd}\|_1;\\
    f_2({\xbf}_{\Ibd},{\xbf}_{\Iat})&=\|\ybf_{\Mbd}-\Abf_{\Mbd,\Xbd}\xbf_{\Ibd}-\Abf_{\Mbd,\Xat}\xbf_{\Iat}\|_1;\\
    f_3({\xbf}_{\Iat})&=\|\ybf_{\Mat}-\Abf_{\Mat,\Xat}\xbf_{\Iat}\|_1.
\end{align*}
Let $\zbf_\Mbd=\ybf_{\Mbd}-\Abf_{\Mbd,\Xat}\xbf_{\Iat}=\Abf_{\Mbd,\Xbd}\xbf_{\natural\Ibd}-\Abf_{\Mbd,\Xat}(\xbf_{\natural\Iat}-\xbf_\Iat)$, and by the structure of $\Abf_{\Mbd,\Xat}$ shown in \eqref{equ:A_struct}, we have $\supp\left(\Abf_{\Mbd,\Xat}(\xbf_{\Iat}-\xbf_{\natural\Iat})\right)\subseteq \Mat$.  Hence, we have that the unique optimal of $f_2({\xbf}_{\Ibd},{\xbf}_{\Iat})$ satisfies $\hat{\xbf}_{\Ibd}={\xbf}_{\natural\Ibd}$ for any given ${\xbf}_{\Iat}$. Since there are no bad data for $\ybf_{\Msf}$ and $\ybf_{\Mbd}$ and moreover $\Abf_{\Msf\cup\Mbd,\Xsf\cup\Xbd}$ has full column rank, the unique minimum of $f_1({\xbf}_{\Isf},{\xbf}_{\Ibd})$ is $(\hat{\xbf}_{\Isf},\hat{\xbf}_{\Ibd})=({\xbf}_{\natural\Isf},{\xbf}_{\natural\Ibd})$. Therefore, for any given ${\xbf}_{\Iat}$, the unique optimal of $f_1({\xbf}_{\Isf},{\xbf}_{\Ibd})+f_2({\xbf}_{\Ibd},{\xbf}_{\Iat})$ is $(\hat{\xbf}_{\Isf},\hat{\xbf}_{\Ibd})=({\xbf}_{\natural\Isf},{\xbf}_{\natural\Ibd})$. Since $f_3({\xbf}_{\Iat})$ does not depend on $({\xbf}_{\Isf},{\xbf}_{\Ibd})$, the unique optimal solution of \eqref{equ:primal_l1} recovers the true solution.

The second approach is as follows. We can write the dual program of \eqref{equ:primal_l1} as:
\begin{equation}
    \max_{\hbf\in\Rbb^{n_m}}\hbf^\top\ybf,\quad \st\quad \Abf^\top\hbf=\bfzero,\quad\|\hbf\|_\infty\leq 1.
    \tag{S$^{(1)}$: $\ell_1$-dual}
    \label{equ:dual_l1}
\end{equation}
To show that $\left(\hat{\xbf}=\begin{bmatrix}\xbf_{\natural\Isf}^\top&\xbf_{\natural\Ibd}^\top&\hat{\xbf}_{\Iat}\end{bmatrix}^\top,\hat{\bbf}=\begin{bmatrix}\bfzero^\top&\hat{\bbf}_\Mbd^\top&\hat{\bbf}_\Mat^\top\end{bmatrix}^\top\right)$ is the optimal solution of \eqref{equ:primal_l1}, we simply need to find a dual certificate $\hbf_\star=\begin{bmatrix}\hbf_\Msf^\top&\hbf_\Mbd^\top&\hbf_\Mat^\top\end{bmatrix}^\top$ that satisfies the KKT conditions:
\begin{align}
    &\text{(dual feasibility) } \quad\Abf^\top\hbf_\star=\bfzero,\label{equ:dual_feas_l1}\\
    &\text{(stationarity) } \;\;\quad\quad\hbf_\star\in\partial\|\hat{\bbf}\|_1.
\end{align}
Since by the reasoning above, $\begin{bmatrix}\xbf_{\natural\Isf}^\top&\xbf_{\natural\Ibd}^\top\end{bmatrix}^\top$ is the unique optimal of the objective $f_1(\xbf_\Isf,\xbf_\Ibd)+f_2(\xbf_\Ibd,\xbf_\Iat)$, it corresponds to a dual certificate $\begin{bmatrix}\hbf_\Msf^\top&\hbf_\Mbd^\top\end{bmatrix}^\top$ such that 
\begin{align}
    &\Abf_{\Msf,\Xsf\cup\Xbd}^\top\hbf_\Msf+\Abf_{\Mbd,\Xsf\cup\Xbd}^\top\hbf_\Mbd=\bfzero,\\
    &\|\hbf_\Msf\|_\infty\leq 1,\|\hbf_\Mbd\|_\infty\leq 1.
\end{align}
Similarly, by the optimality of $\hat{\xbf}_\Iat$ for $f_3(\xbf_\Iat)$, we can find a dual certificate such that:
\begin{align}
    \Abf_{\Mat,\Xat}^\top\hbf_\Mat=\bfzero,\quad\hbf_\Mat\in\partial\|\hat{\bbf}_\Mat\|_1.
\end{align}
Thus, by the structure of $\Abf$, the construction $\hbf_\star=\begin{bmatrix}\hbf_\Msf^\top&\hbf_\Mbd^\top&\hbf_\Mat^\top\end{bmatrix}^\top$ yields a dual certificate. 

\end{proof}


A key condition in Lemma \ref{lemma:bound_defense} is the recovery of the boundary variables in the presence of arbitrary bad data that occur in the attacked region. This condition needs to be checked for every possible attack scenario, which is not useful to understand the system vulnerability in the general case. Instead, we propose a line-based vulnerability index notion in the main text, which provides a sufficient condition in this context. The technical definition is as follows.

\begin{definition}[Local boundary variables and measurements]
For each line $\ell$ that connects nodes $i$ and $j$, let us distinguish the directions $i\rightarrow j$ and $j\rightarrow i$. For the direction $i\rightarrow j$, let $i$ denote the node under attack and $j$ be the node within the inner defense boundary. Accordingly, let $\Bboij$ denote the set of buses (other than $i$) that are directly connected to $j$ as the outer boundary,  $\Bbiij=\{j\}$ be the one-bus inner boundary, and $\Batij=\{i\}$ be the one-bus attack set. Let $\Lbdij$ represent the union of line $\ell$ and the set of lines that bridge $\Bbiij$ and $\Bboij$. Define the ``{boundary variables}'' $\Xbdij$ as the collection of voltage magnitudes $\{\xmg_k\}_{k\in\Bbiij\cup\Bboij}$ and  variables $\{\xre_\eta,\xim_\eta\}$ for the set of lines $\eta\in\Lcal$ that connect the inner boundary $j$ to nodes in the outer boundary $\Bboij$. Define the ``{boundary measurements}'' $\Mbdij=\Mbdijg\cup\Mbdijb$ as the collection of measurements that depend only on the boundary variables $\Xbdij$, denoted by $\Mbdijg$, and measurements that depend on both $\Xbdij$ and variables $\{\xre_\ell,\xim_\ell\}$ of the attacked line $\ell$, denoted by $\Mbdijb$. The above terms can be similarly defined for the direction $j\rightarrow i$ by replacing $i\rightarrow j$ to $j\rightarrow i$ in the notations. Thus, for each line, we will have two sets of boundary variables and measurements.
\label{def:bound_var_mea}
\end{definition}

With the above notations, we can formally describe the line vulnerability index.

\begin{definition}[Line vulnerability index]
\label{def:line_vul_metric}
For each line $\{i,j\}_\ell\in\Lcal$, define the line vulnerability metric $\aij$ along the direction $i\rightarrow j$ as the optimal objective value of the following minimax program:
\begin{subequations}
    \label{equ:lin_vul_for}
	\begin{align}
    \aij=& \underset{\bfxi\in\{-1,+1\}^{\nijb}
	}{\text{max~~~}}\quad\underset{\alpha\in\Rbb,\hbf\in\Rbb^{\nijg}}{\text{min~~~}}
	&&\hspace{-3.5cm} \alpha\\
	& \qquad\text{subject to~~~}
	&&\hspace{-4.5cm}\Abf_{\Mbdijg,\Xbdij}^\top\hbf+\Abf_{\Mbdijb,\Xbdij}^\top\bfxi=\bfzero\\
	&&&\hspace{-4.5cm} \|\hbf\|_\infty\leq\alpha,
	\end{align}
\end{subequations}
where $\nijg=|\Mbdijg|$ and $\nijb=|\Mbdijb|$ are the number of measurements in $\Mbdijg$ and $\Mbdijb$, respectively, and $\Xbdij$, $\Mbdijg$ and $\Mbdijb$ are the boundary variables and measurement indices introduced in Def. \ref{def:bound_var_mea}. Similarly, we can define the backward line vulnerability metric $\aji$ by replacing $i\rightarrow j$ to $j\rightarrow i$ in \eqref{equ:lin_vul_for}. We adopt the measurement normalization convention in Def. \ref{def:norm_measure}.
\end{definition}

Note that for the simple case where there are no lines between any two nodes in $\Ncal_{\Bbi}$, we can extend the above definition to treat each node in $\Ncal_{\Bbi}$ separately. Due to the localized nature, this condition is much weaker than the global mutual incoherence condition in Def. \ref{def:mutual_incoherence}. This is intuitive, because if the network is attacked and the data for a subset of the network are manipulated, then this can be modeled by a cut that removes a subgraph. Then, even if data analytics cannot reason about the lines inside the subgraph, we can still identify the boundary of the subgraph and correctly recover the state for the rest of the network. In fact, we can show the following relationship with the mutual incoherence metric.

\begin{proposition}[Mutual incoherence is more conservative than vulnerability index]
\label{prop:relate_mi}
For each line $\ell$ and the corresponding partitions of measurements $\Mbdijg$, $\Mbdijb$ and variables $\Xbdij$, let $$\rho(\Mbdijb),=\|\Abf_{\Mbdijg,\Xbdij}^{\top+}\Abf_{\Mbdijb,\Xbdij}^\top\|_{\infty}$$ be the mutual incoherence metric defined in Def. \ref{def:mutual_incoherence}. Then, it holds that $\rho(\Mbdijb)\geq\aij$.
\end{proposition}
\begin{proof}
Notice that the line vulnerability index can be written as
\begin{subequations}
	\begin{align}
    \aij=& \underset{\bfxi\in\{-1,+1\}^{\nijb}
	}{\text{max~~~}}\quad\underset{\alpha\in\Rbb,\hbf\in\Rbb^{\nijg}}{\text{min~~~}}
	&&\hspace{-3.5cm} \|\hbf\|_\infty\\
	& \qquad\text{subject to~~~}
	&&\hspace{-4.5cm}\Abf_{\Mbdijg,\Xbdij}^\top\hbf+\Abf_{\Mbdijb,\Xbdij}^\top\bfxi=\bfzero.
	\end{align}
\end{subequations}
Since for any $\bfxi$, the vector $\hat{\hbf}(\bfxi)=-\Abf_{\Mbdijg,\Xbdij}^{\top +}\Abf_{\Mbdijb,\Xbdij}^\top\bfxi$ is a feasible point for the inner optimization, and 
\begin{equation}
    \max_{\bfxi\in\{-1,+1\}^{\nijb}
	}\|\hat{\hbf}(\bfxi)\|_\infty=\rho(\Mbdijb),
\end{equation}
the proof is immediately concluded.
\end{proof}

A key step in establishing the validity of the boundary defense mechanism is to ensure that local defense is sufficient to guard against attacks when solving the problem globally.

\begin{lemma}[Local property implies global property]
\label{lemma:local_sufficient_cond}
Given $\Bat, \Bbi, \Bbo,$ and $\Bsf$ and the associated set partitioning (c.f., Def. \ref{def:bound_var_mea_all}), let $\Abfo=\begin{bmatrix}\Abf_{\Msf,\Xsf}&\Abf_{\Msf,\Xbd}\\\bfzero&\Abf_{\Mbo,\Xbd}\\\bfzero&\Abf_{\Mbi,\Xbd}\end{bmatrix}$, and let $\Lbi\coloneqq\left\{\{i,j\}\in\Lcal\;\middle|\;i\in\Bat,j\in\Bbi\right\}$ be the set of lines that bridge between $\Bat$ and $\Bbi$. If $\aij\leq 1-\gamma$ and $\gamma>0$ for all $\{i,j\}\in\Lbi$ such that $i\in\Bat$ and $j\in\Bbi$, then for any $\hat{\hbf}_{\Mbi}\in[-1,1]^{|\Mbi|}$, there exists an $\hat{\hbf}_{\Msf\cup\Mbo}$ with the properties $\|\hat{\hbf}_{\Msf\cup\Mbo}\|_\infty\leq 1-\gamma$ and
\begin{equation}
    \Abfot_{\Msf\cup\Mbo}\hat{\hbf}_{\Msf\cup\Mbo}+\Abfot_{\Mbi}\hat{\hbf}_{\Mbi}=\bfzero.
    \label{equ:all_exist_h}
\end{equation}
\end{lemma}
\begin{proof}
First, we show that a sufficient condition for the existence of $\hat{\hbf}_{\Msf\cup\Mbo}=\begin{bmatrix}\hat{\hbf}_{\Msf}^\top&\hat{\hbf}_{\Mbo}^\top\end{bmatrix}^\top$ such that $\|\hat{\hbf}_{\Msf\cup\Mbo}\|_\infty\leq 1-\gamma$ and \eqref{equ:all_exist_h} is satisfied is that for any $\hat{\hbf}_{\Mbi}$, there exists an $\hat{\hbf}_{\Mbo}$ such that $\|\hat{\hbf}_{\Mbo}\|_\infty\leq 1-\gamma$ and
\begin{equation}
    \Abf_{\Mbo,\Xbd}^\top\hat{\hbf}_{\Mbo}+\Abf_{\Mbi,\Xbd}^\top\hat{\hbf}_{\Mbi}=\bfzero.
    \label{equ:small_exist_h}
\end{equation}
    This is immediate by simply choosing $\hat{\hbf}_{\Msf\cup\Mbo}=\begin{bmatrix}\bfzero^\top&\hat{\hbf}_{\Mbo}^\top\end{bmatrix}^\top$. In what follows, we prove \eqref{equ:small_exist_h} by induction. The induction rule is as follows: we start by arbitrarily choosing one line $\{i,j\}\in\Lbi$, where $i\in\Bat$ and $j\in\Bbi$, and initialize the measurement set $\Mbo^{(1)}\coloneqq\Mbdijg$, $\Mbi^{(1)}\coloneqq\Mbdijb$ and the variable set $\Xbd^{(1)}\coloneqq\Xbdij$. For each step $k$, we add a new line $\{f,t\}\in\Lbi$ and the associated measurements and variables to $\Mbo^{(k)}$, $\Mbi^{(k)}$ and $\Xbd^{(k)}$, respectively. After the inclusion of all the lines in $\Lbi$, we should obtain the set $\Mbo$, $\Mbi$ and $\Xbd$. In each step, we check whether there exists a vector $\hat{\hbf}_{\Mbo^{(k)}}$ such that $\|\hat{\hbf}_{\Mbo^{(k)}}\|_\infty\leq 1-\gamma$ and
\begin{equation}
    \Abf_{\Mbo^{(k)},\Xbd^{(k)}}^\top\hat{\hbf}_{\Mbo^{(k)}}+\Abf_{\Mbi^{(k)},\Xbd^{(k)}}^\top\hat{\hbf}_{\Mbi^{(k)}}=\bfzero.
    \label{equ:small_exist_h_iter}
\end{equation}
The base case for $k=1$ follows directly from the condition that $\aij\leq 1-\gamma$. For any $k\geq 1$, let $\{f,t\}\in\Lbi$ denote the line to be added, where $f\in\Mat$ and $t\in\Mbi$. There are two possible cases: \textbf{1)} the new line does not share any nodes with the lines that have been already added; or \textbf{2)} the new line shares the attack node $f$ with one (or more) of the lines already added (note that by definition, the new line cannot share the inner boundary node $t$ with one (or more) of the lines already added). For each case, there are also three events that may occur: \textbf{a)} one or more of the nodes in $\Bboij$ are connected to one or more of the nodes in the inner boundaries of lines that have already been added; and/or \textbf{b)} one or more of the nodes in the outer boundary of the lines that have already been added are connected to $t$; or \textbf{c)} none of the above (note that by definition, there are no lines within the inner boundary region). We need to consider all the combinations between the three cases and the three events to show that \eqref{equ:small_exist_h_iter} holds in all scenarios. Fortunately, all the combinations can be reduced to two typical scenarios, where the proofs can be directly applied. We consider these scenarios now. 

The first scenario applies to Cases 1c and 2c, where $\Mbo^{(k+1)}=\Mbo^{(k)}\cup\Mbdftg$, $\Mbi^{(k+1)}=\Mbi^{(k)}\cup\Mbdftb$, $\Xbd^{(k+1)}=\Xbd^{(k)}\cup\Xbdft$, $\Mbo^{(k)}\cap\Mbdftg=\emptyset$, $\Mbi^{(k)}\cap\Mbdftb=\emptyset$, and $\Xbd^{(k)}\cap\Xbdft=\emptyset$. Therefore, for any given $\hat{\hbf}_{\Mbi^{(k+1)}}=\begin{bmatrix}\hat{\hbf}_{\Mbi^{(k)}}^\top&\hat{\bfxi}^\top\end{bmatrix}^\top$ with $\|\hat{\bfxi}\|_\infty\leq 1$, we can always find $\hat{\hbf}_{\Mbo^{(k+1)}}=\begin{bmatrix}\hat{\hbf}_{\Mbo^{(k)}}^\top&\hat{\hbf}^\top\end{bmatrix}^\top$, where $\hat{\hbf}_{\Mbo^{(k)}}$ is given by \eqref{equ:small_exist_h_iter} and $\hat{\hbf}_{\Mbo^{(k)}}$ is given by \eqref{equ:lin_vul_for}, and $\|\hat{\hbf}_{\Mbo^{(k+1)}}\|_\infty\leq1-\gamma$ by definition.

The second scenario applies to Cases 1a, 1b, 2a and 2b. Let $\tNBbo$ be the set of nodes in the outer boundary shared by the new line $\Bboft$ and those of the lines that have been added. Then, we have $\Mbo^{(k+1)}=\Mbo^{(k)}\cup\Mbdftg$, $\Mbi^{(k+1)}=\Mbi^{(k)}\cup\Mbdftb$, $\Xbd^{(k+1)}=\Xbd^{(k)}\cup\Xbdft$, where $\Mbo^{(k)}\cap\Mbdftg$ is the set of voltage magnitude measurements of nodes in $\tNBbo$, $\Mbi^{(k)}\cap\Mbdftb=\emptyset$, and $\Xbd^{(k)}\cap\Xbdft$ is the set of voltage magnitude variables of nodes in $\tNBbo$. For any given $\hat{\hbf}_{\Mbi^{(k)}}$ and $\hat{\bfxi}^\top$, we can always find $\hat{\hbf}_{\Mbo^{(k)}}$ and $\hat{\hbf}^\top$, where $\hat{\hbf}_{\Mbo^{(k)}}$ is given by \eqref{equ:small_exist_h_iter} and $\hat{\hbf}_{\Mbo^{(k)}}$ is given by \eqref{equ:lin_vul_for}. Let $\hat{\hbf}_{\Mbo^{(k)}}$ be further divided into the parts corresponding to the voltage magnitude measurements (if available) of nodes in $\tNBbo$ (i.e. $\left[\hat{\hbf}_{\Mbo^{(k)}}\right]_\tNBbo$) and the rest (i.e. $\left[\hat{\hbf}_{\Mbo^{(k)}}\right]_\tNBboc$); similarly, let $\hat{\hbf}$ be further divided into $\left[\hat{\hbf}\right]_\tNBbo$ and the rest $\left[\hat{\hbf}\right]_\tNBboc$. Then, by setting $\hat{\hbf}_{\Mbo^{(k+1)}}=\begin{bmatrix}\left[\hat{\hbf}_{\Mbo^{(k)}}\right]_\tNBboc^\top&\frac{1}{\mathrm{deg}(\tNBbo)}\circ\left(\left[\hat{\hbf}_{\Mbo^{(k)}}\right]_\tNBbo+\left[\hat{\hbf}\right]_\tNBbo\right)^\top&\left[\hat{\hbf}\right]_\tNBboc^\top\end{bmatrix}^\top$, where $\mathrm{deg}(\tNBbo)$ is the connectivity degree for each node in $\tNBbo$, and $\circ$ indicates the Hadamard (element-wise) product, we can satisfy \eqref{equ:small_exist_h_iter} for any given $\hat{\hbf}_{\Mbi^{(k+1)}}=\begin{bmatrix}\hat{\hbf}_{\Mbi^{(k)}}^\top&\hat{\bfxi}^\top\end{bmatrix}^\top$ (note that the voltage magnitude measurement in the calculation of line vulnerability metric is normalized by 1, but it is weighted by the degree of each node in the actual estimation algorithm, c.f., Def. \ref{def:norm_measure}). Moreover, by construction, we have $\|\hat{\hbf}_{\Mbo^{(k+1)}}\|_\infty\leq 1-\gamma$ for all $k$. This completes the induction proof.
\end{proof}

Lemma \ref{lemma:local_sufficient_cond} implies that as long as all the line vulnerability indices are bounded away from 1, we have a desirable property in terms of defending against bad data on the boundary. This is formalized in the following theorem.

\begin{theorem}
Consider the measurements $\ybf=\Abf\xbf_\natural+\bbf_\natural$, where $\supp(\bbf_\natural)\subseteq\Mat$. Suppose that for the given  partitioning of the network as $\Bat, \Bbi, \Bbo,$ and $\Bsf$, the following conditions hold:
\begin{itemize}
    \item (Full column rank for the safe and boundary region) $\Abf_{\Msf\cup\Mbd,\Xsf\cup\Xbd}$ and $$\Qbf_{\Mbd,\Xbd}=\begin{bmatrix}\Abf_{\Mbd,\Xbd}&\Ibf_{\Mbi}^{(|\Mbd|)\top}\end{bmatrix}$$ have full column rank.
    \item (Localized mutual incoherence) for all lines $\{i,j\}\in\Latbi$ that bridge the attacked region and the inner boundary, where $i\in\Bat$, $j\in\Bbi$, we have $\aij\leq 1-\gamma$ for some $\gamma>0$.  
\end{itemize}
Then, the solution to \eqref{equ:primal_l1}, denoted as $(\hat{\xbf},\hat{\bbf})$, uniquely recovers the true state outside the attacked region (i.e., $\hat{\xbf}_\Isf=\xbf_{\natural\Isf}$ and $\hat{\xbf}_\Ibd=\xbf_{\natural\Ibd}$). Furthermore, the state estimation by \eqref{equ:s2-phase} recovers the true state for the unaffected region (i.e., $\hat{v}_k=v_k$ for $k\in\Bsf\cup\Bbd$).
\label{thm:l1_est_cyber}
\end{theorem}
\begin{proof}
To prove the claim, we simply need to show that for an arbitrary $\bbf_{\star}$ with its support limited to the inner boundary $\supp(\bbf_{\star})\subseteq\Mbi$, the solution $\hat{\xbf}_{\Ibd}\in\Xbd$ to the program
\begin{equation}
    \min_{{\xbf}_{\Ibd}}\|\zbf_{\Mbd}-\Abf_{\Mbd,\Xbd}\xbf_{\Ibd}\|_1
    \label{equ:aux_small_program}
\end{equation}
is unique and satisfies $\hat{\xbf}_{\Ibd}={\xbf}_{\natural \Ibd}$, where $\zbf_{\Mbd}=\Abf_{\Mbd,\Xbd}\xbf_{\natural \Ibd}+\bbf_{\star}$. To show this, we obtain the
dual program:
\begin{equation}
    \max_{\hbf_\Mbd}\hbf_\Mbd^\top\zbf_\Mbd,\quad \st\quad \Abf_{\Mbd,\Xbd}^\top\hbf_\Mbd=\bfzero,\|\hbf_\Mbd\|_\infty\leq 1.
    \label{equ:dual_l1_small}
\end{equation}
Our goal is to find a dual certificate $\hbf_{\star\Mbd}$ that satisfies the KKT conditions:
\begin{align}
    &\text{(dual feasibility) } \quad\Abf_{\Mbd,\Xbd}^\top\hbf_{\star\Mbd}=\bfzero,\label{equ:dual_feas_l1_small}\\
    &\text{(stationarity) } \;\;\quad\quad\hbf_{\star\Mbd}\in\partial\|\bbf_{\star}\|_1.
\end{align}
By the limited support assumption, we need to find a vector $\hbf_\star$ such that $\hbf_{\star\Mbi}=\sign(\bbf_{\star\Mbi})$ and $\|\hbf_{\star\Mbo}\|_\infty\leq 1$. By the mutual incoherence condition and Lemma \ref{lemma:local_sufficient_cond}, we can always find $\hbf_{\star\Mbo}$ that satisfies \eqref{equ:dual_feas_l1_small} for any given $\hbf_{\star\Mbi}$ and $\|\hbf_{\star\Mbo}\|\leq 1-\gamma<1$. Thus, this certifies the optimality of $({\xbf}_{\natural\Ibd},\bbf_\star)$ for \eqref{equ:dual_l1_small}. 

To show that $({\xbf}_{\natural\Ibd},\bbf_\star)$ is the unique optimal solution, let $(\tilde{\xbf},\tilde{\bbf})$ be an arbitrary feasible point of \eqref{equ:aux_small_program} that is different from $({\xbf}_{\natural\Ibd},\bbf_\star)$. Due to the lower eigenvalue condition, the matrix $\Qbf_{\Mbd,\Xbd}\coloneqq\begin{bmatrix}\Abf_{\Mbd,\Xbd}&\Ibf_{\Mbi}^{(|\Mbd|)\top}\end{bmatrix}$ has full column rank. By letting $\tilde{\Jcal}=\supp(\tilde{\bbf})$, the set $\tilde{\Jcal}$ can not be equal to or be a subset of $\Mbi$, because otherwise, from $\Qbf_{\Mbd,\Xbd}\begin{bmatrix}\xbf_{\natural\Ibd}\\
\bbf_\star
\end{bmatrix}=\Qbf_{\Mbd,\Xbd}\begin{bmatrix}\tilde{\xbf}\\
\tilde{\bbf}
\end{bmatrix}=\zbf_\Mbd$, we must have $\begin{bmatrix}\xbf_{\natural\Ibd}\\
\bbf_\star
\end{bmatrix}=\begin{bmatrix}\tilde{\xbf}\\
\tilde{\bbf}
\end{bmatrix}$, which is contradictory  to the assumption. Let $\tilde{\Jcal}_c=\tilde{\Jcal}\setminus\Mbi$; then,
\begin{align}
    \|\bbf_\star\|_1&=\hbf_{\star\Mbd}^\top\zbf_\Mbd\label{equ:der_strong_duality_l1_small}\\
    &=\hbf_{\star\Mbd}^\top(\Abf_{\Mbd,\Xbd}\tilde{\xbf}+\Ibf_{\tilde{\Jcal}_c}^\top\tilde{\bbf}_{\tilde{\Jcal}_c}+\Ibf_{{\Mbi}}^\top\tilde{\bbf}_{{\Mbi}})\label{equ:der_primal_feasible_l1_small}\\
    &=\hbf_{\star\tilde{\Jcal}_c}^\top\tilde{\bbf}_{\tilde{\Jcal}_c}+\hbf_{\star{\Mbi}}^\top\tilde{\bbf}_{{\Mbi}}\label{equ:der_dual_feasible_l1_small}\\
    &\leq \|\hbf_{\star\tilde{\Jcal}_c}\|_\infty\|\tilde{\bbf}_{\tilde{\Jcal}_c}\|_1+\|\hbf_{\star{\Mbi}}\|_\infty\|\tilde{\bbf}_{{\Mbi}}\|_1\label{equ:der_holder_small}\\
    &<\|\tilde{\bbf}_{\tilde{\Jcal}_c}\|_1+\|\tilde{\bbf}_{{\Mbi}}\|_1\label{equ_der_strict_feas_small}\\
    &=\|\tilde{\bbf}\|_1,
\end{align}
where \eqref{equ:der_strong_duality_l1_small} is due to the strong duality between \eqref{equ:aux_small_program} and \eqref{equ:dual_l1_small}, \eqref{equ:der_primal_feasible_l1_small} is due to the primal feasibility of $(\tilde{\xbf},\tilde{\bbf})$, \eqref{equ:der_dual_feasible_l1_small} is due to the dual feasibility condition \eqref{equ:dual_feas_l1_small}, \eqref{equ:der_holder_small} is due to the H\"older inequality, and \eqref{equ_der_strict_feas_small} is due to the strict feasibility of $\hbf_\star$. Thus, we have shown the uniqueness of the optimal solution $(\xbf_{\natural\Ibd},\bbf_\star)$. Together with Lemma \ref{lemma:local_sufficient_cond}, we have proved the theorem.

\end{proof}

This result can be used to certify robustness under different attack scenarios. For example, if there is a topological error caused by line mis-specification, say $\ell=(i,j)$, we can treat the two ends of the line as the attacked nodes, i.e., $\NBat=\{i,j\}$, treat the adjacent nodes to them as inner boundary $\NBbi$, and treat the adjacent nodes to inner boundary as outer boundary $\NBbo$. As long as the line vulnerability index for the lines surrounding the attacked nodes are less than 1, one can identify this gross injection error and thus the topological mistake. We can extend the analysis to the case where the measurements have both sparse bad data and dense noise. In this case, we need to solve a program that combines quadratic loss with absolute value loss. The guarantees now depend on the distribution of the dense noise.
\begin{theorem}[Robust SE with \eqref{equ:primal_noise}]
Consider the measurements  $\ybf=\Abf\xbf_\natural+\wbf_\natural+\bbf_\natural$, where $\wbf_\natural$ has independent entries with zero mean and subgaussian parameter $\sigma$ and $\supp(\bbf_\natural)\subseteq\Mat$. Suppose that the rows of $\Abf$ are normalized (c.f., Def. \ref{def:norm_measure}), and the regularization parameter $\lambda$ is chosen such that
\begin{equation}
    \lambda> \frac{2}{n_m\gamma} \sqrt{2\sigma^2\log n_m}.
\end{equation}
In addition, suppose that for the given  partitioning of the network, i.e. $\Bat, \Bbi, \Bbo,$ and $\Bsf$, the following conditions hold:
\begin{itemize}
    \item (Full column rank for the safe and boundary region) $\Abf_{\Msf\cup\Mbd,\Xsf\cup\Xbd}$ and $$\Qbf_{\Mbd,\Xbd}=\begin{bmatrix}\Abf_{\Mbd,\Xbd}&\Ibf_{\Mbi}^{(|\Mbd|)\top}\end{bmatrix}$$ have full column rank.
    \item (Localized mutual incoherence) for all lines $\{i,j\}\in\Latbi$ that bridge the attacked region and the inner boundary, where $i\in\Bat$, $j\in\Bbi$, we have $\aij\leq 1-\gamma$ for some $\gamma>0$.  
\end{itemize}

Then, the following properties hold for the solution to \eqref{equ:primal_noise}, denoted as $(\hat{\xbf},\hat{\bbf})$:
\begin{enumerate}
    \item (No false inclusion) The solution $(\hat{\xbf},\hat{\bbf})$ has no false bad data inclusion (i.e., $\supp(\hat{\bbf})\subset \supp(\bbf_\natural)$) with probability greater than $1-\frac{c_0}{n_m}$, for some constant $c_0>0$.
    \item (Large bad data detection)  Let ${\Abfo}\coloneqq\begin{bmatrix}\Abf_{\Msf,\Xsf}&\Abf_{\Msf,\Xbd}\\\bfzero&\Abf_{\Mbo,\Xbd}\\
        \bfzero&\Abf_{\Mbi,\Xbd}\end{bmatrix}$ and $\Qbfo_{\Mbi}=\begin{bmatrix}\Abfo&\Ibfot_{\Mbi}\end{bmatrix}$, and $$g(\lambda)=n_m\lambda\left(\frac{1}{2\sqrt{{C_{\text{min}}}}}+\|\Ibf_b(\Qbfot_{\Mbi}\Qbfo_{\Mbi})^{-1}\Ibf_b^\top\|_\infty\right)$$ be a threshold value, and let $\tbbf_{\Mbi}=\Abf_{\Mbi,\Xat}({\xbf}_{\natural\Iat}-\hat{\xbf}_{\Iat})$ be the error at the boundary. Then, all bad data with magnitude greater than $g(\lambda)$ will be detected (i.e., if $|\tilde{b}_{i}|>g(\lambda)$, then $|\hat{b}_{i}|>0$) with probability greater than $1-\frac{c_2}{m}$.
    \item (Bounded error) The estimator error is bounded by \begin{align*}
    \|\xbf_{\natural\Xsf\cup\Xbd}-\hat{\xbf}_{\Xsf\cup\Xbd}\|_2\leq t\frac{\sqrt{|\Xsf|+|\Xbd|+|\Mbi|}}{C_{\text{min}}}+n_m\lambda\|\Ibf_x(\Qbfot_{\Mbi}\Qbfo_{\Mbi})^{-1}\Ibf_b^\top\|_{\infty,2}
\end{align*}
with probability greater than $1-\exp\left(-\frac{c_1t^2}{\sigma^4}\right)$.
\end{enumerate}
\label{thm:lasso_est_cyber}
\end{theorem}

Despite the difference in measurement assumptions (i.e., existence of dense noise $\wbf$) and estimation algorithms (i.e., \eqref{equ:primal_l1} or \eqref{equ:primal_noise}), it is remarkable that the boundary defense conditions in Theorems \ref{thm:l1_est_cyber} and \ref{thm:lasso_est_cyber} are coincident. In the case of negligible dense noise, a deterministic boundary defense is achieved. With the presence of dense noise, it is no longer possible to have deterministic guarantees; however, Theorem~\ref{thm:lasso_est_cyber} indicates that with a proper selection of the penalty coefficient $\lambda$, one can avoid false detection of bad data in the unaffected region (part 1), detect bad data with magnitudes greater than a threshold in the attacked region (part 2), and achieve estimation within bounded error margin for states within the unaffected region. Furthermore, both the bad data threshold and the error bound decrease with stronger mutual incoherence  condition and lower-eigenvalue condition. The proof of the theorem is provided in Section \ref{sec:proof_thm_qp}.

\subsection{Vulnerability index and boundary defense for second-order cone programming}

In this section, we extend the analysis of boundary defense to the case where we perform state estimation with the additional second-order cone constraints.

\begin{lemma}[Boundary defense stops error propagation with SOCP]
Suppose that there is no dense measurement noise (i.e., $\wbf=\bfzero$), and the bad data are confined within $\Mat$, i.e., $\supp(\bbf_{\natural})\subseteq \Mat$. Let $\Kbd$ and $\Kat$ be the subsets of SOC constraints $\Kcal$ restricted to variables $\xbf_\Ibd$ and $\xbf_\Iat$, respectively, and let
\begin{align*}
    \tKat(\hat{\xbf}_\Ibd)=&\Bigg\lbrace\xbf_\Iat\Big\vert\begin{bmatrix}x^{\text{mg}}_i&x^{\text{re}}_\ell+jx^{\text{im}}_\ell\\
x^{\text{re}}_\ell-jx^{\text{im}}_\ell&x^{\text{mg}}_j\end{bmatrix}\succeq 0,\\
&\qquad\qquad\forall\ell\coloneqq(i,j)\in\LBat\cup\Latbi,\text{ where }x^{\text{mg}}_i=\hat{x}^{\text{mg}}_i\quad\forall i\in\Bbi\Bigg\rbrace,
\end{align*}
be the confined feasible set for $\xbf_\Iat$, which fixes the boundary variables $\hat{\xbf}_\Ibd$ in the SOCP constraints. Assume that for an arbitrary $\bbf_{\star\Mbd}$ with its support limited to the inner boundary, i.e. $\supp(\bbf_{\star\Mbd})\subseteq\Mbi$, the solution $\hat{\xbf}_{\Ibd}\in\Xbd$ to the program
\begin{equation}
    \min_{{\xbf}_{\Ibd}\in\Kbd}\|\zbf_{\Mbd}-\Abf_{\Mbd,\Xbd}\xbf_{\Ibd}\|_1, \label{equ:oracle_B_bd_socp}
\end{equation}
is unique and satisfies $\hat{\xbf}_{\Ibd}={\xbf}_{\natural \Ibd}$, where $\zbf_{\Mbd}=\Abf_{\Mbd,\Xbd}\xbf_{\natural \Ibd}+\bbf_{\star\Mbd}$. Assume that the optimal solution $\hat{\xbf}_\Iat$ to
\begin{equation}
    \min_{{\xbf}_{\Iat}\in\Kat}\|\ybf_{\Mat}-\Abf_{\Mat,\Xat}\xbf_{\Iat}\|_1, \label{equ:oracle_B_at_socp}
\end{equation}
also satisfies that $\hat{\xbf}_\Iat\in\tKat({\xbf}_{\natural\Ibd})$.
Then, the solution $\hat{\xbf}$ to \eqref{equ:primal_socp} satisfies $\hat{\xbf}_{\Ibd}={\xbf}_{\natural \Ibd}$ and $\hat{\xbf}_{\Isf}={\xbf}_{\natural \Isf}$.
\label{lem:bound_defend_socp}
\end{lemma}
\begin{proof}
To show that $\left(\hat{\xbf}=\begin{bmatrix}\xbf_{\natural\Isf}^\top&\xbf_{\natural\Ibd}^\top&\hat{\xbf}_{\Iat}^\top\end{bmatrix}^\top,\hat{\bbf}=\begin{bmatrix}\bfzero^\top&\hat{\bbf}_\Mbd^\top&\hat{\bbf}_\Mat^\top\end{bmatrix}^\top\right)$ is the optimal solution of \eqref{equ:primal_socp}, we simply need to find a dual certificate $\left(\hbf_\star=\begin{bmatrix}\hbf_\Msf^\top&\hbf_\Mbd^\top&\hbf_\Mat^\top\end{bmatrix}^\top,\{\nu_\ell,\ubf_\ell\}_{\ell\in\Lcal}\right)$ that satisfies the KKT conditions:
\begin{align}
    &\text{(stationarity) } &\hbf_\star\in\partial\|\hat{\bbf}\|_1,\\
    &\text{(dual feasibility) } &\Abf^\top\hbf_\star+\sum_{\ell\in\Lcal}\left(\nu_\ell\cbf_\ell+\Dbf_\ell^\top\ubf_\ell\right)=\bfzero;\qquad\nu_\ell\geq\|\ubf_\ell\|_2,\quad\forall\ell\in\Lcal,\label{equ:dual_feas_socp}\\
    &\text{(complementary slackness) } &\nu_\ell\cbf_\ell^\top\hat{\xbf}+\ubf_\ell^\top\Dbf_\ell\hat{\xbf}=\bfzero,\quad\forall \ell\in\Lcal,
\end{align}
For a given $\xbf_\Ibd=\xbf_{\natural\Ibd}$, let $\hat{\xbf}_\Isf$ be the optimal solution to
\begin{equation*}
    \min_{{\xbf}_{\Isf}\in\Ksf}\|\ybf_{\Msf}-\Abf_{\Msf,\Xsf}\xbf_{\Isf}-\Abf_{\Msf,\Xbd}\xbf_{\natural\Ibd}\|_1, 
\end{equation*}
where $\Ksf$ is set of all SOCP constraints that involve at least one variable in $\Xsf$. By the lower eigenvalue condition, $\hat{\xbf}_\Isf=\xbf_{\natural\Isf}$ is the unique optimal solution. Since for a given $\hat{\xbf}_\Iat\in\tKat(\xbf_{\natural\Ibd})$, $\hat{\xbf}_\Ibd=\xbf_{\natural\Ibd}$ is the unique optimal of \eqref{equ:oracle_B_bd_socp}, we can conclude that $\begin{bmatrix}\xbf_{\natural\Isf}^\top&\xbf_{\natural\Ibd}^\top\end{bmatrix}^\top$ is the unique optimal of 
\begin{equation*}
    \min_{{\xbf}_{\Isf}\in\Ksf,\xbf_\Ibd\in\Kbd}\|\ybf_{\Msf}-\Abf_{\Msf,\Xsf}\xbf_{\Isf}-\Abf_{\Msf,\Xbd}\xbf_{\Ibd}\|_1+\|\zbf_{\Mbd}-\Abf_{\Mbd,\Xbd}\xbf_{\Ibd}\|_1,
\end{equation*}
which corresponds to a dual certificate $\left(\begin{bmatrix}\hbf_\Msf^\top&\hbf_\Mbd^\top\end{bmatrix}^\top,\{\nu_\ell,\ubf_\ell\}_{\ell\in\LBsf\cup\LBbd}\right)$ such that 
\begin{subequations}
\label{equ:kkt_socp}
    \begin{align}
    &\Abf_{\Msf,\Xsf\cup\Xbd}^\top\hbf_\Msf+\Abf_{\Mbd,\Xsf\cup\Xbd}^\top\hbf_\Mbd+\sum_{\ell\in\LBsf\cup\LBbd}\left(\nu_\ell\cbf_\ell+\Dbf_\ell^\top\ubf_\ell\right)=\bfzero,\\
    &\nu_\ell\geq\|\ubf_\ell\|_2,\quad\forall\ell\in\LBsf\cup\LBbd,\\
    &\nu_\ell\cbf_\ell^\top\hat{\xbf}+\ubf_\ell^\top\Dbf_\ell\hat{\xbf}=\bfzero,\quad\forall \ell\in\LBsf\cup\LBbd,\\
    &\|\hbf_\Msf\|_\infty\leq 1,\|\hbf_\Mbd\|_\infty\leq 1.
\end{align}
\end{subequations}

Similarly, by the optimality of $\hat{\xbf}_\Iat$ for \eqref{equ:oracle_B_at_socp}, we can find a dual certificate such that:
\begin{align*}
    \Abf_{\Mat,\Xat}^\top\hbf_\Mat+\sum_{\ell\in\LBat}\left(\nu_\ell\cbf_\ell+\Dbf_\ell^\top\ubf_\ell\right)=\bfzero,\quad\hbf_\Mat\in\partial\|\hat{\bbf}_\Mat\|_1.
\end{align*}
Thus, by setting $\left(\left\lbrace\nu_\ell=0,\ubf_\ell=\bfzero\right\rbrace_{\ell\in\Latbi}\right)$, and note that $\Lcal=\LBsf\cup\LBbd\cup\Latbi\cup\LBat$, the construction $\left(\left\lbrace\nu_\ell,\ubf_\ell\right\rbrace_{\ell\in\Lcal}\right)$ and $\hbf_\star=\begin{bmatrix}\hbf_\Msf^\top&\hbf_\Mbd^\top&\hbf_\Mat^\top\end{bmatrix}^\top$ yield a dual certificate. 
\end{proof}

Now, we formally define the vulnerability index.
\begin{definition}[Line vulnerability for SOCP]
\label{def:line_vul_metric_socp}
For each line $\{i,j\}_\ell\in\Lcal$ and a given $\xbf\in\Kcal$ that satisfies primal feasibility, define the line vulnerability metric $\aijsocp$ along the direction $i\rightarrow j$ as the optimal value of the following minimax program:
\begin{subequations}
    \label{equ:lin_vul_for_socp}
	\begin{align}
    \aijsocp(\xbf)=& \underset{\bfxi\in\{-1,+1\}^{\nijb}
	}{\text{max~~~}}\quad\underset{\alpha\in\Rbb,\bomega\in\Rbb^{\nijl},\hbf\in\Rbb^{\nijg}}{\text{min~~~}}
	&&\hspace{-3.5cm} \alpha\\
	& \;\;\;\text{subject to~~~}
	&&\hspace{-4.5cm}\Abf_{\Mbdijg,\Xbdij}^\top\hbf+\Abf_{\Mbdijb,\Xbdij}^\top\bfxi+\sum_{\ell\in\Lbdij}\omega_\ell\Tbf_\ell\xbf=\bfzero\\
	&&&\hspace{-4.5cm} \omega_\ell\geq 0,\qquad\qquad\forall \ell\in\Lbdij\\
	&&&\hspace{-4.5cm} \|\hbf\|_\infty\leq\alpha,
	\end{align}
\end{subequations}
where $\nijg=|\Mbdijg|,\nijb=|\Mbdijb|,\nijl=|\Lbdij|$ are the number of measurements/lines  in $\Mbdijg$, $\Mbdijb$ and $\Lbdij$, respectively, and $\Xbdij$, $\Mbdijg$, $\Mbdijb$ and $\Lbdij$ are defined in Def. \ref{def:bound_var_mea}. Also, we define $\Tbf_\ell=\cbf_\ell\cbf_\ell^\top-\Dbf_\ell^\top\Dbf_\ell$, where $\cbf_\ell$ and $\Dbf_\ell$ are given in \eqref{equ:c_l_D_l}. Similarly, we define the backward line vulnerability metric $\aji$ by replacing $i\rightarrow j$ to $j\rightarrow i$ in \eqref{equ:lin_vul_for_socp}. We adopt the measurement normalization convention in Def. \ref{def:norm_measure}.
\end{definition}

\begin{lemma}
\label{lem:equi_vul_socp}
The line vulnerability metric $\aijsocp(\xbf)$ for a given $\xbf\in\Kcal$ that satisfies the primal feasibility coincides with the optimal objective value of the following minimax program:
\begin{subequations}
    \label{equ:lin_vul_for2_socp}
	\begin{align}
    \atijsocp(\xbf)=&\underset{\tilde{\bfxi}\in [-1,+1]^{\nijb}
	}{\text{max~~~}}\quad\underset{\tilde{\alpha}\in\Rbb,\bnu\in\Rbb^{\nijl},\tilde{\hbf}\in\Rbb^{\nijg}}{\text{min~~~}}
	&& \hspace{0.1cm}\tilde{\alpha}\\
	& \;\;\text{subject to~~~}
	&&\hspace{-4.5cm}\Abf_{\Mbdijg,\Xbdij}^\top\tilde{\hbf}+\Abf_{\Mbdijb,\Xbdij}^\top\tilde{\bfxi}+\sum_{\ell\in\Lbdij}\nu_\ell\cbf_\ell+\Dbf_\ell^\top\ubf_\ell=\bfzero\\
	&&&\hspace{-4.5cm}\nu_\ell\geq\|\ubf_\ell\|_2,\qquad\qquad\qquad\quad\forall\ell\in\Lbdij\label{equ:dual_cons_vul_socp}\\
	&&&\hspace{-4.5cm}\nu_\ell\cbf_\ell^\top\xbf+\ubf_\ell^\top\Dbf_\ell\xbf=\bfzero,\qquad\;\;\forall\ell\in\Lbdij\label{equ:comp_slack_vul_socp}\\
	&&&\hspace{-4.5cm} \|\tilde{\hbf}\|_\infty\leq\tilde{\alpha},
	\end{align}
\end{subequations}
with the same notations as in Def. \ref{def:line_vul_metric}, where $\cbf_\ell$ and $\Dbf_\ell$ are define in \eqref{equ:c_l_D_l}.
\end{lemma}
\begin{proof}
The equivalence between optimizing over $[-1,+1]^{\nijb}$ and $\{-1,+1\}^{\nijb}$ for the outer minimization can be reasoned as in \eqref{lem:equiv_vul} due to the convexity of the feasibility region given $\xbf\in\Kcal$ and $\tilde{\bfxi}$. Since $\xbf$ satisfies the primal feasibility, which can be expressed as in \eqref{equ:c_l_D_l}, a standard result (c.f., \cite[Lemma 15]{alizadeh2003second}) in analogy to linear programming indicates that \eqref{equ:comp_slack_vul_socp} is equivalent to:
\begin{equation*}
    \nu_\ell\Dbf_\ell\xbf+\cbf_\ell^\top\xbf\ubf_\ell=\bfzero,\qquad\;\;\forall\ell\in\Lbdij,
\end{equation*}
which indicates that $\nu_\ell=\omega_\ell\cbf_\ell^\top\xbf$ and $\ubf_\ell=-\omega_\ell \Dbf_\ell\xbf$ for $\omega_\ell\geq 0$ and $\ell\in\Lbdij$. It can be verified that this also satisfies the SOCP constraints \eqref{equ:dual_cons_vul_socp}. By the definition of $\Tbf_\ell=\cbf_\ell\cbf_\ell^\top-\Dbf_\ell^\top\Dbf_\ell$, the equivalence to \eqref{equ:lin_vul_for_socp} is established.

\end{proof}

\begin{lemma}[Local property implies global property for SOCP]
\label{lemma:local_sufficient_cond_socp}
Given $\Bat, \Bbi, \Bbo,$ and $\Bsf$ and the associated set partitioning (c.f., Def. \ref{def:bound_var_mea_all}), let  $\Abfo=\begin{bmatrix}\Abf_{\Msf,\Xsf}&\Abf_{\Msf,\Xbd}\\\bfzero&\Abf_{\Mbo,\Xbd}\\\bfzero&\Abf_{\Mbi,\Xbd}\end{bmatrix}$, and $\cbf_\ell^\circ$ and $\Dbf_\ell^\circ$ to be the   subvector and submatrix of $\cbf_\ell$ and $\Dbf_\ell$ indexed by $\Xsf\cup\Xbd$. If $\aijsocp\leq 1-\gamma$ and $\gamma>0$ for all $\{i,j\}\in\Lbi$ such that $i\in\Bat$ and $j\in\Bbi$, then for any $\hat{\hbf}_{\Mbi}\in[-1,1]^{|\Mbi|}$, there exist  $\hat{\hbf}_{\Msf\cup\Mbo}$ and $\{\hat{\nu}_\ell,\hat{\ubf}_\ell\}_{\ell\in\Latbi\cup\LBbd\cup\LBsf}$ with the properties that $\|\hat{\hbf}_{\Msf\cup\Mbo}\|_\infty\leq 1-\gamma$ and
\begin{equation}
    \Abfot_{\Msf\cup\Mbo}\hat{\hbf}_{\Msf\cup\Mbo}+\Abfot_{\Mbi}\hat{\hbf}_{\Mbi}+\sum_{\ell\in\Latbi\cup\LBbd\cup\LBsf}\hat{\nu}_\ell\cbf_\ell^\circ+\Dbf_\ell^{\circ\top}\hat{\ubf}_\ell=\bfzero.
    \label{equ:all_exist_h_socp}
\end{equation}
\end{lemma}
\begin{proof}
The proof is similar to the one for Lemma \ref{lemma:local_sufficient_cond}. First, we show that a sufficient condition for Lemma \ref{lemma:local_sufficient_cond_socp} is that for any $\hat{\hbf}_{\Mbi}$, there exists $\hat{\hbf}_{\Mbo}$ and $\{\hat{\nu}_\ell,\hat{\ubf}_\ell\}_{\ell\in\Latbi\cup\LBbd}$ such that $\|\hat{\hbf}_{\Mbo}\|_\infty\leq 1-\gamma$ and
\begin{equation}
    \Abf_{\Mbo,\Xbd}^\top\hat{\hbf}_{\Mbo}+\Abf_{\Mbi,\Xbd}^\top\hat{\hbf}_{\Mbi}+\sum_{\ell\in\Latbi\cup\LBbd}\left[\hat{\nu}_\ell\cbf_\ell+\Dbf_\ell^\top\hat{\ubf}_\ell\right]_{\Xbd}=\bfzero.
    \label{equ:small_exist_h_socp}
\end{equation}
    This is immediate by simply choosing $\hat{\hbf}_{\Msf\cup\Mbo}=\begin{bmatrix}\bfzero^\top&\hat{\hbf}_{\Mbo}^\top\end{bmatrix}^\top$ and $\hat{\nu}_\ell=0$ and $\hat{\ubf}_\ell=\bfzero$ for $\ell\in\LBsf$. In what follows, we prove \eqref{equ:small_exist_h_socp} by induction. The induction rule is as follows: we start by arbitrarily choosing one line $\{i,j\}\in\Lbi$, where $i\in\Bat$ and $j\in\Bbi$, and initialize the line set $\LBbdnum{1}=\Lbdij$, the measurement set $\Mbo^{(1)}\coloneqq\Mbdijg$, $\Mbi^{(1)}\coloneqq\Mbdijb$ and the variable set $\Xbd^{(1)}\coloneqq\Xbdij$. For each step $k$, we add a new line $\{f,t\}\in\Lbi$ such that $\LBbdnum{k}=\LBbdnum{k-1}\cup\Lbdft$, and the associated measurements and variables to $\Mbo^{(k)}$, $\Mbi^{(k)}$ and $\Xbd^{(k)}$, respectively. After the inclusion of all the lines in $\Lbi$, we should obtain the set $\Mbo$, $\Mbi$ and $\Xbd$. In each step, we check whether there exist $\{\hat{\nu}_\ell,\hat{\ubf}_\ell\}_{\ell\in\LBbdnum{k}}$ and $\hat{\hbf}_{\Mbo^{(k)}}$ such that $\|\hat{\hbf}_{\Mbo^{(k)}}\|_\infty\leq 1-\gamma$ and
\begin{equation}
    \Abf_{\Mbo^{(k)},\Xbd^{(k)}}^\top\hat{\hbf}_{\Mbo^{(k)}}+\Abf_{\Mbi^{(k)},\Xbd^{(k)}}^\top\hat{\hbf}_{\Mbi^{(k)}}+\sum_{\ell\in\LBbdnum{k}}\left[\hat{\nu}_\ell\cbf_\ell+\Dbf_\ell^\top\hat{\ubf}_\ell\right]_{\Xbd^{(k)}}=\bfzero.
    \label{equ:small_exist_h_iter_socp}
\end{equation}
The base case for $k=1$ follows directly from the condition that $\aijsocp\leq 1-\gamma$. For any $k\geq 1$, let $\{f,t\}\in\Lbi$ denote the line to be added, where $f\in\Mat$ and $t\in\Mbi$. There are two possible cases: \textbf{1)} the new line does not share any nodes with lines that have been already added; or \textbf{2)} the new line shares the attack node $f$ with one (or more) of the lines already added (note that by definition, the new line cannot share the inner boundary node $t$ with one (or more) of the lines already added). For each case, there are also three events that might occur: \textbf{a)} one or more of the nodes in $\Bboft$ are connected to one or more of the nodes in the inner boundaries of lines that have already been added; and/or \textbf{b)} one or more of the nodes in outer boundary of the lines that have already been added are connected to $t$; or \textbf{c)} none of the above (note that by definition, there are no lines within the inner boundary region). We need to consider all the combinations between the three cases and the three events to show that \eqref{equ:small_exist_h_iter} holds in all scenarios. Fortunately, all the combinations can be reduced to two typical scenarios, where the proofs can be directly applied. We consider these scenarios now. 

The first scenario applies to Cases 1c and 2c, where $\Mbo^{(k+1)}=\Mbo^{(k)}\cup\Mbdftg$, $\Mbi^{(k+1)}=\Mbi^{(k)}\cup\Mbdftb$, $\Xbd^{(k+1)}=\Xbd^{(k)}\cup\Xbdft$, $\Mbo^{(k)}\cap\Mbdftg=\emptyset$, $\Mbi^{(k)}\cap\Mbdftb=\emptyset$, and $\Xbd^{(k)}\cap\Xbdft=\emptyset$. Therefore, for any given $\hat{\hbf}_{\Mbi^{(k+1)}}=\begin{bmatrix}\hat{\hbf}_{\Mbi^{(k)}}^\top&\hat{\bfxi}^\top\end{bmatrix}^\top$, we can always find $\hat{\hbf}_{\Mbo^{(k+1)}}=\begin{bmatrix}\hat{\hbf}_{\Mbo^{(k)}}^\top&\hat{\hbf}^\top\end{bmatrix}^\top$ and $\{\hat{\nu}_\ell,\hat{\ubf}_\ell\}_{\ell\in\LBbdnum{k+1}}=\{\hat{\nu}_\ell,\hat{\ubf}_\ell\}_{\LBbdnum{k}\cup\Lbdft}$, where $\hat{\hbf}_{\Mbo^{(k)}}$ and $\{\hat{\nu}_\ell,\hat{\ubf}_\ell\}_{\LBbdnum{k}}$ are given by \eqref{equ:small_exist_h_iter_socp}, $\hat{\hbf}$ and $\{\hat{\nu}_\ell,\hat{\ubf}_\ell\}_{\Lbdft}$ are given by \eqref{equ:lin_vul_for_socp}, and $\|\hat{\hbf}_{\Mbo^{(k+1)}}\|_\infty\leq1-\gamma$ by definition.

The second scenario applies to Cases 1a, 1b, 2a and 2b. Let $\tNBbo$ be the set of nodes in the outer boundary shared by the new line $\Bboft$ and those of the lines that have been added. Then, we have $\Mbo^{(k+1)}=\Mbo^{(k)}\cup\Mbdftg$, $\Mbi^{(k+1)}=\Mbi^{(k)}\cup\Mbdftb$, $\Xbd^{(k+1)}=\Xbd^{(k)}\cup\Xbdft$, where $\Mbo^{(k)}\cap\Mbdftg$ is the set of voltage magnitude measurements of nodes in $\tNBbo$, $\Mbi^{(k)}\cap\Mbdftb=\emptyset$, and $\Xbd^{(k)}\cap\Xbdft$ is the set of voltage magnitude variables of nodes in $\tNBbo$. For any given $\hat{\hbf}_{\Mbi^{(k)}}$ and $\hat{\bfxi}^\top$ with $\|\hat{\bfxi}\|_\infty\leq 1$, we can always find $\hat{\hbf}_{\Mbo^{(k)}}$, $\hat{\hbf}^\top$, $\{\hat{\nu}_\ell,\hat{\ubf}_\ell\}_{\LBbdnum{k}}$ and $\{\hat{\nu}_\ell,\hat{\ubf}_\ell\}_{\Lbdft}$, where $\hat{\hbf}_{\Mbo^{(k)}}$ and $\{\hat{\nu}_\ell,\hat{\ubf}_\ell\}_{\LBbdnum{k}}$ are given by \eqref{equ:small_exist_h_iter_socp}, and $\hat{\hbf}_{\Mbo^{(k)}}$ and $\{\hat{\nu}_\ell,\hat{\ubf}_\ell\}_{\Lbdft}$ are given by \eqref{equ:lin_vul_for_socp}. Let $\hat{\hbf}_{\Mbo^{(k)}}$ be further divided into the parts corresponding to the voltage magnitude measurements (if available) of nodes in $\tNBbo$, namely $\left[\hat{\hbf}_{\Mbo^{(k)}}\right]_\tNBbo$ and the rest, namely $\left[\hat{\hbf}_{\Mbo^{(k)}}\right]_\tNBboc$; similarly, $\hat{\hbf}$ be further divided into $\left[\hat{\hbf}\right]_\tNBbo$ and the rest, namely $\left[\hat{\hbf}\right]_\tNBboc$. Then, we set $\hat{\hbf}_{\Mbo^{(k+1)}}=\begin{bmatrix}\left[\hat{\hbf}_{\Mbo^{(k)}}\right]_\tNBboc^\top&\frac{1}{\mathrm{deg}(\tNBbo)}\circ\left(\left[\hat{\hbf}_{\Mbo^{(k)}}\right]_\tNBbo+\left[\hat{\hbf}\right]_\tNBbo\right)^\top&\left[\hat{\hbf}\right]_\tNBboc^\top\end{bmatrix}^\top$. Similarly, we can perform the transformation for $\{\hat{\nu}_\ell,\hat{\ubf}_\ell\}_{\LBbdnum{k+1}}$. Hence, we can satisfy \eqref{equ:small_exist_h_iter_socp} for any given $\hat{\hbf}_{\Mbi^{(k+1)}}=\begin{bmatrix}\hat{\hbf}_{\Mbi^{(k)}}^\top&\hat{\bfxi}^\top\end{bmatrix}^\top$ with $\|\hat{\bfxi}\|_\infty\leq 1$ (note that the voltage magnitude measurement in the calculation of line vulnerability metric is normalized by 1, but it is weighted by the degree of connections of each node the actual estimation algorithm, c.f., Def. \ref{def:norm_measure}). Moreover, by construction, we have $\|\hat{\hbf}_{\Mbo^{(k+1)}}\|_\infty\leq 1-\gamma$ for all $k$. This completes the induction proof.
\end{proof}
\begin{proposition}[SOC constraint can improve line vulnerability]
\label{prop:soc_better}
For any $\xbf\in\Kcal$, it holds that $$\aijsocp(\xbf)\leq\aij$$
\end{proposition}
\begin{proof}
For any given $\bfxi$, let $\hat{\hbf}$ be the optimal solution of the inner minimizer in \eqref{equ:lin_vul_for} with $\|\hat{\hbf}\|_\infty\leq\aij$. Then, the tuple $(\aijsocp=\aij,\bomega=\bfzero,\hbf=\hat{\hbf})$ is a feasible solution for \eqref{equ:lin_vul_for_socp}, which proves that we always have $\aijsocp(\xbf)\leq\aij$.
\end{proof}

The above proposition implies a key advantage of incorporating SOCP constraints---to improve robustness. This has also been empirically validated in our study as shown in the main text.

\begin{theorem}
Consider the measurements $\ybf=\Abf\xbf_\natural+\bbf_\natural$, where $\supp(\bbf_\natural)\subseteq\Mat$, and also a  partitioning of the network as $\Bat, \Bbi, \Bbo,$ and $\Bsf$. Let $\Kbd$ and $\Kat$ be the subsets of SOCP constraints $\Kcal$ restricted to variables $\xbf_\Ibd$ and $\xbf_\Iat$, respectively, and let
\begin{align*}
    \tKat(\hat{\xbf}_\Ibd)=&\Bigg\lbrace\xbf_\Iat\Big\vert\begin{bmatrix}x^{\text{mg}}_i&x^{\text{re}}_\ell+jx^{\text{im}}_\ell\\
x^{\text{re}}_\ell-jx^{\text{im}}_\ell&x^{\text{mg}}_j\end{bmatrix}\succeq 0,\\
&\qquad\qquad\qquad\forall\ell\coloneqq(i,j)\in\LBat\cup\Latbi,\text{ where }x^{\text{mg}}_i=\hat{x}^{\text{mg}}_i\quad\forall i\in\Bbi\Bigg\rbrace,
\end{align*}
be the confined feasible set for $\xbf_\Iat$, which fixes the boundary variables $\hat{\xbf}_\Ibd$ in the SOCP constraints. Suppose that the following conditions hold:
\begin{itemize}
    \item (Full column rank for the safe and boundary region) $\Abf_{\Msf\cup\Mbd,\Xsf\cup\Xbd}$ and $$\Qbf_{\Mbd,\Xbd}=\begin{bmatrix}\Abf_{\Mbd,\Xbd}&\Ibf_{\Mbi}^{(|\Mbd|)\top}\end{bmatrix}$$ have full column rank.
    \item (Localized mutual incoherence) for all lines $\{i,j\}\in\Latbi$ that bridge the attacked region and the inner boundary, where $i\in\Bat$, $j\in\Bbi$, we have $\aijsocp\leq 1-\gamma$ for some $\gamma>0$.  
    \item (Nonbinding SOCP constraints in the boundary) the solution for the attacked states satisfies $\hat{\xbf}_\Iat\in\tKat({\xbf}_{\natural\Ibd})$.
\end{itemize}
Then, the solution to \eqref{equ:primal_socp}, denoted as $(\hat{\xbf},\hat{\bbf})$, uniquely recovers the true state outside the attacked region (i.e., $\hat{\xbf}_\Isf=\xbf_{\natural\Isf}$ and $\hat{\xbf}_\Ibd=\xbf_{\natural\Ibd}$). Furthermore, the state estimation by \eqref{equ:s2-phase} recovers the true state for the unaffected region (i.e., $\hat{v}_k=v_k$ for $k\in\Bsf\cup\Bbd$).
\label{thm:l1_est_cyber_socp}
\end{theorem}
\begin{proof}
To prove the claim, we simply need to show that for an arbitrary $\bbf_{\star}$ with its support limited to the inner boundary $\supp(\bbf_{\star})\subseteq\Mbi$, the solution $\hat{\xbf}_{\Ibd}\in\Xbd$ to the program
\begin{equation}
    \min_{{\xbf}_{\Ibd}\in \Kbd,\bbf}\|\bbf\|_1,\quad\st\quad \Abf_{\Mbd,\Xbd}\xbf_{\Ibd}+\bbf=\zbf_{\Mbd}
    \label{equ:aux_small_program_socp}
\end{equation}
is unique and satisfies $\hat{\xbf}_{\Ibd}={\xbf}_{\natural \Ibd}$, where $\zbf_{\Mbd}=\Abf_{\Mbd,\Xbd}\xbf_{\natural \Ibd}+\bbf_{\star}$. To show this, we obtain the dual program:
\begin{subequations}
    \label{equ:dual_socp_l1_kbd}
	\begin{align}
    &\underset{\hbf_\Mbd,\{\nu_\ell,\bmu_\ell\}_{\ell\in\LBbd}}{\text{min~~~}}
	&& \hspace{-1.0cm}\hbf_\Mbd^\top\zbf_\Mbd\\
	& \;\;\quad\text{subject to~~~}
	&&\hspace{-1.0cm}\Abf_{\Mbd,\Xbd}^\top\hbf_\Mbd+\sum_{\ell\in\Latbi\cup\LBbd}(\nu_\ell\cbf_\ell+\Dbf_\ell^\top\bmu_\ell)=\bfzero\\
	&&&\hspace{-1.0cm}\|\hbf_\Mbd\|_\infty\leq 1\\
	&&&\hspace{-1.0cm}\nu_\ell\geq\|\bmu_\ell\|_2,\forall \ell\in\LBbd,
	\end{align}
\end{subequations}
Our goal is to find a dual certificate $\hbf_{\star\Mbd}$ and $\{\lambda_{\star\ell},\bmu_{\star\ell}\}_{\LBbd}$ that satisfies the KKT conditions:
\begin{align}
    &\text{(dual feasibility) } \;\quad\qquad\qquad\lambda_{\star\ell}\geq\|\bmu_{\star\ell}\|_2,\quad\forall\ell\in\LBbd\\
    &\text{(stationarity) } \;\;\;\quad\qquad\qquad\quad\Abf_{\Mbd,\Xbd}^\top\hbf_{\star\Mbd}+\sum_{\ell\in\Latbi\cup\LBbd}(\lambda_{\star\ell}\cbf_\ell+\Dbf_\ell^\top\bmu_{\star\ell})=\bfzero,\\
    &\qquad\qquad\qquad\qquad\qquad\qquad\;\hbf_{\star\Mbd}\in\partial\|\bbf_\star\|_1\\
    &\text{(complementary slackness) } \quad\lambda_{\star\ell}\cbf_\ell^\top\xbf_\star+\bmu_{\star\ell}^\top\Dbf_\ell\xbf_\star=\bfzero,\quad\forall\ell\in\LBbd.
\end{align}
where $\xbf_\star=\begin{bmatrix}\xbf_{\natural\Isf}^\top&\xbf_{\natural\Ibd}^\top&\hat{\xbf}_{\Iat}^\top\end{bmatrix}^\top$. By the limited support assumption, we need to find a vector $\hbf_\star$ such that $\hbf_{\star\Mbi}=\sign(\bbf_{\star\Mbi})$ and $\|\hbf_{\star\Mbo}\|_\infty\leq 1$. By the vulnerability index condition and Lemma \ref{lemma:local_sufficient_cond_socp}, we can always find $\hbf_{\star\Mbo}$ and $\{\lambda_{\star\ell},\bmu_{\star\ell}\}_{\LBbd}$ that satisfy the KKT conditions for a given $\hbf_{\star\Mbi}$, such that $\|\hbf_{\star\Mbo}\|\leq 1-\gamma<1$. Thus, this certifies the optimality of $({\xbf}_{\natural\Ibd},\bbf_\star)$ for \eqref{equ:dual_l1_small}. Clearly, under the nonbinding SOCP constraints assumption, $({\xbf}_{\natural\Ibd},\bbf_\star)$ is feasible. Following the uniqueness argument of Theorem \ref{thm:l1_est_cyber}, we conclude the proof.
\end{proof}

We can extend the analysis to the case where the measurements have both sparse bad data and dense noise. In this case, we need to solve a second-order cone program that combines quadratic loss with absolute value loss, in addition to the SOCP constraints. 
\begin{theorem}[Robust SE with \eqref{equ:primal_lasso_socp}]
Given the measurements  $\ybf=\Abf\xbf_\natural+\wbf_\natural+\bbf_\natural$, where $\wbf_\natural$ has independent entries with zero mean and subgaussian parameter $\sigma$ and $\supp(\bbf_\natural)\subseteq\Mat$, consider a  partitioning of the network as $\Bat, \Bbi, \Bbo,$ and $\Bsf$. Let $\Kbd$ and $\Kat$ be the subsets of SOCP constraints $\Kcal$ restricted to the variables $\xbf_\Ibd$ and $\xbf_\Iat$, respectively, and let
\begin{align*}
    \tKat(\hat{\xbf}_\Ibd)=&\Bigg\lbrace\xbf_\Iat\Big\vert\begin{bmatrix}x^{\text{mg}}_i&x^{\text{re}}_\ell+jx^{\text{im}}_\ell\\
x^{\text{re}}_\ell-jx^{\text{im}}_\ell&x^{\text{mg}}_j\end{bmatrix}\succeq 0,\\
&\qquad\qquad\qquad\qquad\qquad\forall\ell\coloneqq(i,j)\in\LBat\cup\Latbi,\text{ where }x^{\text{mg}}_i=\hat{x}^{\text{mg}}_i\quad\forall i\in\Bbi\Bigg\rbrace,
\end{align*}
be the confined feasible set for $\xbf_\Iat$, which fixes the boundary variables $\hat{\xbf}_\Ibd$ in the SOCP constraints. Suppose that the rows of $\Abf$ are normalized (c.f., Def. \ref{def:norm_measure}), and the regularization parameter $\lambda$ is chosen such that
\begin{equation}
    \lambda> \frac{2}{n_m\gamma} \sqrt{2\sigma^2\log n_m}.
\end{equation}
In addition, suppose the following conditions hold:
\begin{itemize}
    \item (Full column rank for the safe and boundary region) both $\Abf_{\Msf\cup\Mbd,\Xsf\cup\Xbd}$ and $$\Qbf_{\Mbd,\Xbd}=\begin{bmatrix}\Abf_{\Mbd,\Xbd}&\Ibf_{\Mbi}^{(|\Mbd|)\top}\end{bmatrix}$$ have full column rank.
    \item (Localized mutual incoherence) for all lines $\{i,j\}\in\Latbi$ that bridge the attacked region and the inner boundary, where $i\in\Bat$, $j\in\Bbi$, we have $\aijsocp\leq 1-\gamma$ for some $\gamma>0$.  
    \item (Nonbinding SOCP constraints in the boundary) the solution for the attacked states satisfies $\hat{\xbf}_\Iat\in\tKat({\xbf}_{\natural\Ibd})$.
\end{itemize}

Then, the following properties hold for the solution to \eqref{equ:primal_noise}, denoted as $(\hat{\xbf},\hat{\bbf})$:
\begin{enumerate}
    \item (No false inclusion) The solution $(\hat{\xbf},\hat{\bbf})$ has no false bad data inclusion (i.e., $\supp(\hat{\bbf})\subset \supp(\bbf_\natural)$) with probability greater than $1-\frac{c_0}{n_m}$, for some constant $c_0>0$.
    \item (Large bad data detection) Let ${\Abfo}\coloneqq\begin{bmatrix}\Abf_{\Msf,\Xsf}&\Abf_{\Msf,\Xbd}\\\bfzero&\Abf_{\Mbo,\Xbd}\\
        \bfzero&\Abf_{\Mbi,\Xbd}\end{bmatrix}$ and $\Qbfo_{\Mbi}=\begin{bmatrix}\Abfo&\Ibfot_{\Mbi}\end{bmatrix}$, and $$g(\lambda)=n_m\lambda\left(\frac{1}{2\sqrt{{C_{\text{min}}}}}+\|\Ibf_b(\Qbfot_{\Mbi}\Qbfo_{\Mbi})^{-1}\Qbfot_{\Mbi}\|_\infty\right)$$ be a threshold value, and let $\tbbf_{\Mbi}=\Abf_{\Mbi,\Xat}({\xbf}_{\natural\Iat}-\hat{\xbf}_{\Iat})$ be the error at the boundary. Then, all bad data with magnitude greater than $g(\lambda)$ will be detected (i.e., if $|\tilde{b}_{i}|>g(\lambda)$, then $|\hat{b}_{i}|>0$) with probability greater than $1-\frac{c_2}{m}$.
    \item (Bounded error) The estimator error is bounded by \begin{align*}
    \|\xbf_{\natural\Xsf\cup\Xbd}-\hat{\xbf}_{\Xsf\cup\Xbd}\|_2\leq t\frac{\sqrt{|\Xsf|+|\Xbd|+|\Mbi|}}{C_{\text{min}}}+n_m\lambda\|\Ibf_x(\Qbfot_{\Mbi}\Qbfo_{\Mbi})^{-1}\Qbfot_{\Mbi}\|_{\infty,2}
\end{align*}
with probability greater than $1-\exp\left(-\frac{c_1t^2}{\sigma^4}\right)$.
\end{enumerate}
\label{thm:lasso_est_cyber_socp}
\end{theorem}

\subsection{Scalable methods to calculate the vulnerability index}

The minimax program \eqref{equ:lin_vul_for} consists of a linear programming in the inner minimization and a discrete optimization in the outer maximization. For small-scale systems, the number of feasible points in the outer maximization is not too large. This is the case when we consider the vulnerability on a line-by-line basis. But for large-scale problems when we consider a group of attacked lines, it is essential to develop more scalable numerical algorithms. We first show the following result.

\begin{lemma}
\label{lem:equiv_vul}
The line vulnerability index $\aij$ coincides with the optimal value of the following minimax program:
\begin{subequations}
    \label{equ:lin_vul_for2}
	\begin{align}
    \atij=&\underset{\tilde{\bfxi}\in [-1,+1]^{\nijb}
	}{\text{max~~~}}\quad\underset{\tilde{\alpha}\in\Rbb,\tilde{\hbf}\in\Rbb^{\nijg}}{\text{min~~~}}
	&&\hspace{-3.5cm} \tilde{\alpha}\\
	& \qquad\text{subject to~~~}
	&&\hspace{-4.5cm}\Abf_{\Mbdijg,\Xbdij}^\top\tilde{\hbf}+\Abf_{\Mbdijb,\Xbdij}^\top\tilde{\bfxi}=\bfzero\\
	&&&\hspace{-4.5cm} \|\tilde{\hbf}\|_\infty\leq\tilde{\alpha},
	\end{align}
\end{subequations}
with the same notations as in Def. \ref{def:line_vul_metric}. Note that the difference in \eqref{equ:lin_vul_for2} is that the minimizer is over the hypercube $[-1,+1]^{\nijb}$ rather than the simplex $\{-1,+1\}^{\nijb}$.
\end{lemma}
\begin{proof}
Since the feasible region of the outside maximizer in \eqref{equ:lin_vul_for2} is a superset of that in \eqref{equ:lin_vul_for}, we always have $\atij\geq\aij$. To show the other direction, we simply need to show that for any $\tilde{\bfxi}\in [-1,+1]^{\nijb}$, we can always find a feasible solution for the minimizer $\tilde{\hbf}$ such that $\|\tilde{\hbf}\|_\infty\leq\aij$. Since $\tilde{\bfxi}$ belongs to a hypercube, which is convex, there always exists a set of non-negative coefficients $\beta_k$ such that $\beta_k\geq 0$, $\sum_k\beta_k=1$ and $\tilde{\bfxi}=\sum_k\beta_k\bfxi_k$, where $\bfxi_k\in\{-1,+1\}^{\nijb}$. Since for each $\bfxi_k$, there exists $\hbf_k$ such that it is feasible in \eqref{equ:lin_vul_for} and $\|\hbf_k\|_\infty\leq\aij$, by choosing $\tilde{\hbf}=\sum_k\beta_k\hbf_k$, we have:
\begin{equation*}
    \|\tilde{\hbf}\|_\infty\leq\sum_k\beta_k\|\hbf_k\|_\infty\leq\sum_k\beta_k\aij=\aij,
\end{equation*}
which completes the proof.
\end{proof}

We can thereby reformulate the problem as a linear complimentarity problem as follows. The KKT conditions for the inner minimization of \eqref{equ:lin_vul_for2} are:
\begin{itemize}
    \item (Primal feasibility) $\Abf_{\Mbdijg,\Xbdij}^\top{\hbf}+\Abf_{\Mbdijb,\Xbdij}^\top{\bfxi}=\bfzero$, $\qbf_+=\alpha\onebf-\hbf$, $\qbf_-=\alpha\onebf+\hbf$, $\qbf_+\geq 0,\qbf_-\geq 0$;
    \item (Dual feasibility) $\bmu_+\geq 0,\bmu_-\geq 0$;
    \item (Stationarity) The Lagrangian function $$\Lcal(\alpha,\hbf,\bmu_+,\bmu_-,\blambda)=\alpha+\blambda^\top(\Abf_{\Mbdijg,\Xbdij}^\top{\hbf}+\Abf_{\Mbdijb,\Xbdij}^\top{\bfxi})+\bmu_+^\top(\hbf-\alpha\onebf)+\bmu_-^\top(-\hbf-\alpha\onebf)$$ and stationarity conditions:
    \begin{align*}
        \frac{\partial\Lcal}{\partial\alpha} &= 1-\bmu_+^\top\onebf-\bmu_-^\top\onebf=0\\
        \frac{\partial\Lcal}{\partial\hbf} &= \Abf_{\Mbdijg,\Xbdij}\blambda+\bmu_+-\bmu_-=\bfzero
    \end{align*}
    \item (complementary slackness) $\bmu_+\circ\qbf_+=\bfzero$, $\bmu_-\circ\qbf_-=\bfzero$
    \end{itemize}

Thus, we can write \eqref{equ:lin_vul_for2} as a linear complementarity problem:
\begin{subequations}
    \label{equ:lin_vul_for3}
	\begin{align}
    \atij=&\underset{{\bfxi}\in\Rbb^{\nijb},{\alpha}\in\Rbb,{\hbf}\in\Rbb^{\nijg}
	}{\text{max~~~}}
	&&\hspace{-3.5cm} {\alpha}\\
	& \text{subject to~~~}
	&& \hspace{-4.5cm}-\onebf\leq{\bfxi}\leq\onebf\\
	&&&\hspace{-4.5cm}\Abf_{\Mbdijg,\Xbdij}^\top{\hbf}+\Abf_{\Mbdijb,\Xbdij}^\top{\bfxi}=\bfzero\\
	&&&\hspace{-4.5cm}\qbf_+=\alpha\onebf-\hbf\\
	&&&\hspace{-4.5cm}\qbf_-=\alpha\onebf+\hbf\\
	&&&\hspace{-4.5cm} 1-\bmu_+^\top\onebf-\bmu_-^\top\onebf=0\\
	&&&\hspace{-4.5cm} \Abf_{\Mbdijg,\Xbdij}\blambda+\bmu_+-\bmu_-=\bfzero\\
	&&&\hspace{-4.5cm} \bmu_+\circ\qbf_+=\bmu_-\circ\qbf_-=\bfzero\\,
	&&&\hspace{-4.5cm} \qbf_+,\qbf_-,\bmu_+ ,\bmu_-\geq 0
	\end{align}
\end{subequations}

This problem can be solved readily using off-the-shelf solvers such as PATH Solver \cite{ferris2000complementarity} or YALMIP \cite{lofberg2004yalmip}. We can also use the big-M method to replace the complimentarity condition using a mixed-integer formulation, and solve the problem using standard packages such as Gurobi. In our experiments, we only focus on each line, so the improvement of computation is not significant. The advantage becomes more obvious when we scale the computation to multiple lines.

\subsection{Extension to tree decomposition}

So far, we have been focusing on evaluating line vulnerabilities. In this section, we introduce a powerful extension of the vulnerability index to tree decomposition of a graph. This allows us to study the effect of sparsity on network robustness. We  use $\Ncal(\Gcal)$ to represent the vertices of graph $\Gcal$, and $\Lcal(i;\Gcal)=\{j\in\Ncal\mid \{i,j\}\in\Lcal\}$ to represent the set of nodes in $\Gcal$ that are connected to node $i$. First, we introduce the standard definition of tree decomposition and treewidth.

\begin{definition}[Tree decomposition and treewidth]
A tree decomposition of a graph $\Gcal\coloneqq\{\Ncal,\Lcal\}$ is $(\Tcal,\Wcal)$, where $\Tcal$ is a tree and $\Wcal\coloneqq\{W_t\mid t\in\Ncal(\Tcal)\}$ is the set of ``bags'' $W_t$ which satisfies the following properties
\begin{enumerate}
    \item (Node coverage) $\cup_{t\in\Ncal(\Tcal)}\Wcal_t=\Ncal(\Gcal)$, i.e., the union of the vertices of $\Tcal$, referred to as ``bags,'' is the set of nodes of $\Gcal$;
    \item (Edge coverage) For any $(i,j)\in\Lcal$, there exists $t\in\Ncal(\Tcal)$ such that $i,j\in\Wcal_t$, i.e., each edge of $\Gcal$ is in at least one of the ``bags'' of $\Tcal$;
    \item (Running intersection property) The subtree of $\Tcal$ consisting of all ``bags'' containing $u\in\Ncal$ is connected.
\end{enumerate} 
Furthermore, the width of a tree decomposition is $\max(|\Wcal_t|-1:t\in\Ncal(\Tcal))$. The treewidth of $\Gcal$ is the minimum width of a tree decomposition of $\Gcal$.
\end{definition}

Clearly, a graph may have several different tree decompositions. The analysis below does not require any particular tree decompositions. However, the easiest tree decomposition is to lump all vertices into one bag, which does not reveal any robustness properties of the graph. In general, the smaller the width of the decomposition, the easier it is to certify robustness. 

\begin{definition}[Infected bags, link bags, safe bags]
For a given set of attacked nodes $\NBat$ and a tree decomposition $(\Tcal,\Wcal)$, any bag that contains attacked nodes is referred to as an infected bag $\Wif_t\in\Wif=\{\Wcal_t\mid \Wcal_t\cap\NBat\neq\emptyset\}$. Furthermore, the set of lines induced by the union of infected bags is denoted as $\Lif$. The bags that are immediately connected to an infected bag are called link bags $\Wlk_t\in\Wlk=\{\Wcal_t\mid \Wcal_t\cap\NBat=\emptyset,\exists\;\Wif_{t'},\Wcal_t\in\Lcal(\Wif_{t'};\Tcal)\}$, and the set of lines induced by the union of link bags is shown as $\Llk$. The rest of the bags are safe bags $\Wsf_t$, and the set of lines induced by the union of safe bags is represented by $\Lsf$. Nodes shared between a link bag $\Wlk_t$ and an infected bag $\Wif_t$ are called adhesion nodes $\Nad(\Wlk_t,\Wif_t)=\Wlk_t\cap\Wif_t\subseteq\Nad$, and the rest of the nodes in $\Wlk_t$ are outer link nodes $\Nol(\Wlk_t,\Wif_t)=\Wlk_t\setminus\Wif_t\subseteq\Nol$. We denote the edges that connect adhesion nodes in $\Wlk_t$ with infected nodes by $\Lad(\Wlk_t)\subseteq\Lad$. 
\end{definition}

\begin{definition}[Attacked, boundary and safe variables and measurements for tree decomposition]
The set of ``infected variables'' $\Xif$ includes all variables on lines induced by $\Wif$ and on nodes in $\Wif$ except for adhesion nodes $\Nad$. The set of ``link variables'' $\Xlk$ includes variables on nodes in $\Wlk$ and the induced lines. The set of ``safe variables'' $\Xsf$ includes all other variables. The set of ``infected measurements'' $\Mif$ includes measurements on lines induced by nodes in $\Wif$ and on nodes in $\Wif$ except for voltage magnitude measurements on $\Nad$. The set of ``adhesion measurements'' $\Mad$ includes nodal power injections on nodes in $\Nad$ and line measurements on $\Lad$, and the set of ``outer link measurements'' $\Mol$ includes voltage magnitude on nodes in $\Wlk$ and line measurements induced by nodes in $\Wlk$. Together, they form the ``boundary measurements'' $\Mbd\coloneqq\Mad\cup\Mol$. The rest of the measurements $\Msf$ are ``safe measurements.''
\label{def:bound_var_mea_all_tree}
\end{definition}


Next, we introduce some useful properties associated with the above definitions. If $\Tcal'$ is a subtree of $\Tcal$, we use $\Gcal_{\Tcal'}$ to denote the subgraph of $\Gcal$ induced by the nodes in all the bags associated with $\Tcal'$, namely $\cup_{t\in\Tcal'}\Wcal_t$.

\begin{lemma}
The following properties are satisfied:
\begin{enumerate}
    \item[(i)] There are no shared nodes between the safe bags and the infected bags.
    \item[(ii)] There are no shared nodes between the set of outer link nodes and the infected bags.
    \item[(iii)] Suppose that the infected bags form a subtree of $\Tcal$. Then, there are no shared outer link nodes between any link bags.
    \item[(iv)] Suppose that the infected bags form a subtree of $\Tcal$. Consider any link bag $\Wlk_t$ that is adjacent to only one infected bag $\Wif_{t'}$ connected by an edge $\Lcal(\Wlk_t,\Wif_{t'})$. If we delete the edge, the tree falls apart into two connected components, $\Tcal_1$ and $\Tcal_2$. Deleting the adhesion nodes $\Wlk_t\cap\Wif_{t'}$ from $\Ncal$ disconnects $\Gcal$ into the two subgraphs $\Gcal_{\Tcal_1} -(\Wlk_t\cap\Wif_{t'})$ and $\Gcal_{\Tcal_2} -(\Wlk_t\cap\Wif_{t'})$. Furthermore, all the infected nodes are contained in only one of the subgraph, and there is no edge across the two subgraphs.
\end{enumerate}
\label{lem:tree_property}
\end{lemma}
\begin{proof}
(i): For any safe bag $\Wsf_t$ and affected bag $\Wif_t$, if there exists a node $i$ that is shared between them, it contradicts the definition of a safe bag.

(ii): For any link bag $\Wlk_t$ and affected bag $\Wif_t$, if there exists a node $i$ that is shared between $\Wif_t$ and the outer link nodes in $\Wlk_t$, then by the running intersection property, it must also appear in the infected bag connected to $\Wlk_t$. This is contradictory, because it makes $i$ an adhesion node.

(iii): For any two link bags $\Wlk_t$ and $\Wlk_{t'}$, suppose that they share an outer link node $i$. By the running intersection property, there must exist a path of bags between $\Wlk_t$ and $\Wlk_{t'}$. Since this path cannot go through the infected bags, it must be outside the infected region. Since the infected bags form a subtree, this would create a loop within $\Tcal$, which is impossible. Therefore, there cannot be any shared outer link nodes between any two link bags.

(iv) Assume that there is a node $i$ that belongs to both $\Gcal_{\Tcal_1}-(\Wlk_t\cap\Wif_{t'})$ and $\Gcal_{\Tcal_2}-(\Wlk_t\cap\Wif_{t'})$. Therefore, by the node coverage property, there must exist $\Wcal_x$ with $x\in\Tcal_1$ and $\Wcal_y$ with $y\in\Tcal_2$ such that $i\in\Wcal_x$ and $i\in\Wcal_y$. Since $\Wlk_t$ and $\Wif_{t'}$ lie on a $x-y$ path in $\Tcal$, by the running intersection property, $i\in\Wlk_t\cap\Wif_{t'}$. Hence, $i$ belongs to neither $\Gcal_{\Tcal_1}-(\Wlk_t\cap\Wif_{t'})$ nor $\Gcal_{\Tcal_2}-(\Wlk_t\cap\Wif_{t'})$.

Now, assume that there is an edge $(i,j)$ in $\Gcal$ such that $i\in \Gcal_{\Tcal_1}-(\Wlk_t\cap\Wif_{t'})$ and $j\in\Gcal_{\Tcal_2}-(\Wlk_t\cap\Wif_{t'})$. Then, by the edge coverage property, there must be a bag $\Wcal_x$ containing both $i$ and $j$. However, $x$ cannot be in both $\Tcal_1$ and $\Tcal_2$, otherwise, $i$ and $j$ will belong to $\Wlk_t\cap\Wif_{t'}$. Assume that $x\not\in\Tcal_2$. Since $j$ is in $\Gcal_{\Tcal_2}-(\Wlk_t\cap\Wif_{t'})$, it must be in a bag $y\in\Tcal_2$ different than $x$. Since $j$ belongs to both $\Wcal_x$ and $\Wcal_y$, it lies on a $x-y$ path in $\Tcal$. By the running intersection property, we have $j\in \Wlk_t\cap\Wif_{t'}$, which is a contradiction.
\end{proof}

If the infected bags form a subtree and we can find a link bag that is adjacent to only one infected bag, then by property~(iv) in Lemma \ref{lem:tree_property}, if we remove the adhesion nodes, we can separate the infected region with the rest of the safe region.

Now, we can define a generalized version of vulnerability index using tree decomposition.
\begin{definition}[Bag vulnerability index]
\label{def:bag_vul_metric}
For each adhesion link $\Lcal(\Wlk_t,\Wif_t)$, define the measurement and variable partitions according to Def. \ref{def:bound_var_mea_all_tree}. The bag vulnerability index $\alpha_{\Wif_t\rightarrow\Wlk_t}$ is given by the optimal value of the following minimax program:
\begin{subequations}
    \label{equ:lin_vul_for_bag}
	\begin{align}
    \alpha_{\Wif_t\rightarrow\Wlk_t}=& \underset{\bfxi\in\{-1,+1\}^{|\Mad|}
	}{\text{max~~~}}\quad\underset{\alpha\in\Rbb,\hbf\in\Rbb^{|\Mol|}}{\text{min~~~}}
	&&\hspace{-3.5cm} \alpha\\
	& \qquad\text{subject to~~~}
	&&\hspace{-4.5cm}\Abf_{\Mol,\Xlk}^\top\hbf+\Abf_{\Mad,\Xlk}^\top\bfxi=\bfzero\\
	&&&\hspace{-4.5cm} \|\hbf\|_\infty\leq\alpha,
	\end{align}
\end{subequations}
where $\Mol,\Mad$ and $\Xlk$ are the boundary measurement and variable indices introduced in Def. \ref{def:bound_var_mea_all_tree}.
\end{definition}

Similarly, we can extend the definition to incorporate SOCs.

\begin{definition}[Bag vulnerability index for SOCP]
\label{def:bag_vul_metric_socp}
For each adhesion link $\Lcal(\Wlk_t,\Wif_t)$, define the measurement and variable partitions according to Def. \ref{def:bound_var_mea_all_tree}. The bag vulnerability index $\alpha^{\mathrm{SOCP}}_{\Wif_t\rightarrow\Wlk_t}$ for a given $\xbf\in\Kcal$ that satisfies primal feasibility is given by the optimal value of the following minimax program:
\begin{subequations}
    \label{equ:lin_vul_for_bag_socp}
	\begin{align}
    \alpha^{\mathrm{SOCP}}_{\Wif_t\rightarrow\Wlk_t}=& \underset{\bfxi\in\{-1,+1\}^{|\Mad|}
	}{\text{max~~~}}\quad\underset{\alpha\in\Rbb,\hbf\in\Rbb^{|\Mol|}}{\text{min~~~}}
	&&\hspace{-3.5cm} \alpha\\
	& \qquad\text{subject to~~~}
	&&\hspace{-4.5cm}\Abf_{\Mol,\Xlk}^\top\hbf+\Abf_{\Mad,\Xlk}^\top\bfxi+\sum_{\ell\in{\Lcal(\Wlk_t)}}\omega_\ell\Tbf_\ell\xbf=\bfzero\\
	&&&\hspace{-4.5cm} \omega_\ell\geq 0,\qquad\qquad\forall \ell\in{\Lcal(\Wlk_t)}\\
	&&&\hspace{-4.5cm} \|\hbf\|_\infty\leq\alpha,
	\end{align}
\end{subequations}
where $\Mol,\Mad$ and $\Xlk$ are the boundary measurement and variable indices introduced in Def. \ref{def:bound_var_mea_all_tree}, $\Lcal(\Wlk_t)$ is the set of lines induced by nodes in $\Wlk_t$. Also, we define $\Tbf_\ell=\cbf_\ell\cbf_\ell^\top-\Dbf_\ell^\top\Dbf_\ell$, where $\cbf_\ell$ and $\Dbf_\ell$ are defined in \eqref{equ:c_l_D_l}. 
\end{definition}

With the above definition of  bag vulnerability index, we can show the following key results for SE robustness.

\begin{lemma}[Local property implies global property in tree decomposition]
\label{lemma:local_sufficient_cond_tree}
Consider a tree decomposition $\Tcal$ and the associated set partitioning (c.f., Def. \ref{def:bound_var_mea_all_tree}). Suppose that the infected bags form a subtree of $\Tcal$, and that there exists a link bag $\Wlk_t$ that is adjacent to only one infected bag $\Wif_t$. For simplicity of presentation, we also treat the rest of the bags in the subtree as infected. Let $\Abfo=\begin{bmatrix}\Abf_{\Msf,\Xsf}&\Abf_{\Msf,\Xlk}\\\bfzero&\Abf_{\Mol,\Xlk}\\\bfzero&\Abf_{\Mad,\Xlk}\end{bmatrix}$ be a submatrix of the sensing matrix. If $\alpha_{\Wif_t\rightarrow\Wlk_t}\leq 1-\gamma$ for some $\gamma>0$, then for any $\hat{\hbf}_{\Mad}\in[-1,1]^{|\Mad|}$, there exists an $\hat{\hbf}_{\Msf\cup\Mol}$ such that $\|\hat{\hbf}_{\Msf\cup\Mol}\|_\infty\leq 1-\gamma$ and
\begin{equation}
    \Abfot_{\Msf\cup\Mol}\hat{\hbf}_{\Msf\cup\Mol}+\Abfot_{\Mad}\hat{\hbf}_{\Mad}=\bfzero.
    \label{equ:all_exist_h_tree}
\end{equation}
\end{lemma}
\begin{proof}
The proof is similar to Lemma \ref{lemma:local_sufficient_cond}. First, we show that a sufficient condition for the existence of $\hat{\hbf}_{\Msf\cup\Mol}=\begin{bmatrix}\hat{\hbf}_{\Msf}^\top&\hat{\hbf}_{\Mol}^\top\end{bmatrix}^\top$ such that $\|\hat{\hbf}_{\Msf\cup\Mol}\|_\infty\leq 1-\gamma$ and \eqref{equ:all_exist_h_tree} is satisfied is that for any $\hat{\hbf}_{\Mad}$, there exists a vector $\hat{\hbf}_{\Mol}$ such that $\|\hat{\hbf}_{\Mol}\|_\infty\leq 1-\gamma$ and
\begin{equation}
    \Abf_{\Mol,\Xlk}^\top\hat{\hbf}_{\Mol}+\Abf_{\Mad,\Xlk}^\top\hat{\hbf}_{\Mad}=\bfzero.
    \label{equ:small_exist_h_tree}
\end{equation}
    This is immediate by simply choosing $\hat{\hbf}_{\Msf\cup\Mol}=\begin{bmatrix}\bfzero^\top&\hat{\hbf}_{\Mol}^\top\end{bmatrix}^\top$. Since it is guaranteed that there exists a vector $\hat{\hbf}_{\Mol}$ to satisfy \eqref{equ:small_exist_h_tree} under the condition that $\alpha_{\Wif_t\rightarrow\Wlk_t}\leq 1-\gamma$, the claim is proved.
\end{proof}

\begin{lemma}[Local property implies global property for SOCP with tree decomposition]
\label{lemma:local_sufficient_cond_socp_tree}
Consider a tree decomposition $\Tcal$ and the associated set partitioning (c.f., Def. \ref{def:bound_var_mea_all_tree}). Suppose that the infected bags form a subtree of $\Tcal$, and that there exists a link bag $\Wlk_t$ that is adjacent to only one infected bag $\Wif_t$. For simplicity of presentation, we also treat the rest of the bags in the subtree as infected. Let $\Abfo=\begin{bmatrix}\Abf_{\Msf,\Xsf}&\Abf_{\Msf,\Xlk}\\\bfzero&\Abf_{\Mol,\Xlk}\\\bfzero&\Abf_{\Mad,\Xlk}\end{bmatrix}$ be a submatrix of the sensing matrix, and $\cbf_\ell^\circ$ and $\Dbf_\ell^\circ$ be the   subvector and submatrix of $\cbf_\ell$ and $\Dbf_\ell$ indexed by $\Xsf\cup\Xlk$. If $\alpha^{\mathrm{SOCP}}_{\Wif_t\rightarrow\Wlk_t}\leq 1-\gamma$ for some $\gamma>0$, then for any $\hat{\hbf}_{\Mad}\in[-1,1]^{|\Mad|}$, there exist  $\hat{\hbf}_{\Msf\cup\Mol}$ and $\{\hat{\nu}_\ell,\hat{\ubf}_\ell\}_{\ell\in\Lad(\Wlk_t)\cup\Lcal(\Wlk_t)\cup\LBsf}$ such that $\|\hat{\hbf}_{\Msf\cup\Mol}\|_\infty\leq 1-\gamma$ and
\begin{equation}
    \Abfot_{\Msf\cup\Mol}\hat{\hbf}_{\Msf\cup\Mol}+\Abfot_{\Mad}\hat{\hbf}_{\Mad}+\sum_{\ell\in\Lad(\Wlk_t)\cup\Lcal(\Wlk_t)\cup\LBsf}\hat{\nu}_\ell\cbf_\ell^\circ+\Dbf_\ell^{\circ\top}\hat{\ubf}_\ell=\bfzero.
    \label{equ:all_exist_h_socp_tree}
\end{equation}
\end{lemma}
\begin{proof}
The proof is similar to the one for Lemma \ref{lemma:local_sufficient_cond_socp}. First, we show that a sufficient condition for Lemma \ref{lemma:local_sufficient_cond_socp_tree} is that for any $\hat{\hbf}_{\Mad}$, there exists a vector $\hat{\hbf}_{\Mol}$ and $\{\hat{\nu}_\ell,\hat{\ubf}_\ell\}_{\ell\in\Lcal(\Wlk_t)}$ such that $\|\hat{\hbf}_{\Mol}\|_\infty\leq 1-\gamma$ and
\begin{equation}
    \Abf_{\Mol,\Xlk}^\top\hat{\hbf}_{\Mol}+\Abf_{\Mad,\Xlk}^\top\hat{\hbf}_{\Mad}+\sum_{\ell\in\Lad(\Wlk_t)\cup\Lcal(\Wlk_t)}\left[\hat{\nu}_\ell\cbf_\ell+\Dbf_\ell^\top\hat{\ubf}_\ell\right]_{\Xlk}=\bfzero.
    \label{equ:small_exist_h_socp_tree}
\end{equation}
    This is immediate by simply choosing $\hat{\hbf}_{\Msf\cup\Mol}=\begin{bmatrix}\bfzero^\top&\hat{\hbf}_{\Mol}^\top\end{bmatrix}^\top$ and $\hat{\nu}_\ell=0$ and $\hat{\ubf}_\ell=\bfzero$ for $\ell\in\LBsf$. Since it is guaranteed that there exist $\hat{\hbf}_{\Mol}$ and $\{\hat{\nu}_\ell,\hat{\ubf}_\ell\}_{\ell\in\Lcal(\Wlk_t)}$ to satisfy \eqref{equ:small_exist_h_socp_tree} under the condition that $\alpha^{\mathrm{SOCP}}_{\Wif_t\rightarrow\Wlk_t}\leq 1-\gamma$, the claim is proved.
\end{proof}

\begin{theorem}[Robust SE with \eqref{equ:primal_noise} for tree decomposition]
Consider a tree decomposition $\Tcal$ and the associated set partitioning (c.f., Def. \ref{def:bound_var_mea_all_tree}). Suppose that the infected bags form a subtree of $\Tcal$, and that there exists a link bag $\Wlk_t$ that is adjacent to only one infected bag $\Wif_t$. Given the measurements  $\ybf=\Abf\xbf_\natural+\wbf_\natural+\bbf_\natural$, where $\wbf_\natural$ has independent entries with zero mean and subgaussian parameter $\sigma$ and $\supp(\bbf_\natural)\subseteq\Mif$, suppose that the rows of $\Abf$ are normalized (c.f., Def. \ref{def:norm_measure}) and the regularization parameter $\lambda$ is chosen such that
\begin{equation}
    \lambda> \frac{2}{n_m\gamma} \sqrt{2\sigma^2\log n_m}.
\end{equation}
In addition, assume that the following conditions hold:
\begin{itemize}
    \item (Full column rank for the safe and boundary region) $\Abf_{\Msf\cup\Mlk,\Xsf\cup\Xlk}$ and $$\Qbf_{\Mlk,\Xlk}=\begin{bmatrix}\Abf_{\Mlk,\Xlk}&\Ibf_{\Mad}^{(|\Mlk|)\top}\end{bmatrix}$$ have full column rank.
    \item (Localized mutual incoherence for bags) for the link bag $\Wlk_t$ that is adjacent to only one infected bag $\Wif_t$, we have $\alpha_{\Wif_t\rightarrow\Wlk_t}\leq 1-\gamma$ for some $\gamma>0$.
\end{itemize}

Then, the following properties hold for the solution to \eqref{equ:primal_noise}, denoted as $(\hat{\xbf},\hat{\bbf})$:
\begin{enumerate}
    \item (No false inclusion) The solution $(\hat{\xbf},\hat{\bbf})$ has no false bad data inclusion (i.e., $\supp(\hat{\bbf})\subset \supp(\bbf_\natural)$) with probability greater than $1-\frac{c_0}{n_m}$, for some constant $c_0>0$.
    \item (Large bad data detection)  Let ${\Abfo}\coloneqq\begin{bmatrix}\Abf_{\Msf,\Xsf}&\Abf_{\Msf,\Xlk}\\\bfzero&\Abf_{\Mol,\Xlk}\\
        \bfzero&\Abf_{\Mad,\Xlk}\end{bmatrix}$ and $\Qbfo_{\Mad}=\begin{bmatrix}\Abfo&\Ibfot_{\Mad}\end{bmatrix}$, and $$g(\lambda)=n_m\lambda\left(\frac{1}{2\sqrt{{C_{\text{min}}}}}+\|\Ibf_b(\Qbfot_{\Mad}\Qbfo_{\Mad})^{-1}\Ibf_b^\top\|_\infty\right)$$ be a threshold value, and let $\tbbf_{\Mad}=\Abf_{\Mad,\Xif}({\xbf}_{\natural\Iif}-\hat{\xbf}_{\Iif})$ be the error at the boundary. Then, all bad data with magnitude greater than $g(\lambda)$ will be detected (i.e., if $|\tilde{b}_{i}|>g(\lambda)$, then $|\hat{b}_{i}|>0$) with probability greater than $1-\frac{c_2}{m}$.
    \item (Bounded error) The estimator error is bounded by \begin{align*}
    \|\xbf_{\natural\Xsf\cup\Xlk}-\hat{\xbf}_{\Xsf\cup\Xlk}\|_2\leq t\frac{\sqrt{|\Xsf|+|\Xlk|+|\Mad|}}{C_{\text{min}}}+n_m\lambda\|\Ibf_x(\Qbfot_{\Mad}\Qbfo_{\Mad})^{-1}\Ibf_b^\top\|_{\infty,2}
\end{align*}
with probability greater than $1-\exp\left(-\frac{c_1t^2}{\sigma^4}\right)$.
\end{enumerate}
\label{thm:lasso_est_cyber_tree}
\end{theorem}

\begin{theorem}[SE robustness with \eqref{equ:primal_lasso_socp} for tree decomposition]
Given a tree decomposition $\Tcal$ and the associated set partitioning (c.f., Def. \ref{def:bound_var_mea_all_tree}), suppose that the infected bags form a subtree of $\Tcal$ and that there exists a link bag $\Wlk_t$ that is adjacent to only one infected bag $\Wif_t$. Consider the measurements  $\ybf=\Abf\xbf_\natural+\wbf_\natural+\bbf_\natural$, where $\wbf_\natural$ has independent entries with zero mean and subgaussian parameter $\sigma$ and $\supp(\bbf_\natural)\subseteq\Mif$.  Let $\Klk$ and $\Kif$ be the subsets of SOCP constraints $\Kcal$ restricted to the variables $\xbf_\Ilk$ and $\xbf_\Iif$, respectively, and let
\begin{align*}
    \tKif(\hat{\xbf}_\Ilk)=&\Bigg\lbrace\xbf_\Iif\Big\vert\begin{bmatrix}x^{\text{mg}}_i&x^{\text{re}}_\ell+jx^{\text{im}}_\ell\\
x^{\text{re}}_\ell-jx^{\text{im}}_\ell&x^{\text{mg}}_j\end{bmatrix}\succeq 0,\\
&\qquad\qquad\qquad\quad\forall\ell=(i,j)\in\Lif\cup\Lad(\Wlk_t),\text{ where }x^{\text{mg}}_i=\hat{x}^{\text{mg}}_i\quad\forall i\in\Nad(\Wlk_t,\Wif_t)\Bigg\rbrace,
\end{align*}
be the confined feasible set for $\xbf_\Iif$, which fixes the boundary variables $\hat{\xbf}_\Ilk$ in the SOCP constraints. Suppose that rows of $\Abf$ are normalized (c.f., Def. \ref{def:norm_measure}), and the regularization parameter $\lambda$ is chosen such that
\begin{equation}
    \lambda> \frac{2}{n_m\gamma} \sqrt{2\sigma^2\log n_m}.
\end{equation}
In addition, suppose that the following conditions hold:
\begin{itemize}
    \item (Full column rank for the safe and boundary region) $\Abf_{\Msf\cup\Mlk,\Xsf\cup\Xlk}$ and $$\Qbf_{\Mlk,\Xlk}=\begin{bmatrix}\Abf_{\Mlk,\Xlk}&\Ibf_{\Mad}^{(|\Mlk|)\top}\end{bmatrix}$$ have full column rank.
    \item (Localized mutual incoherence for bags) for the link bag $\Wlk_t$ that is adjacent to only one infected bag $\Wif_t$, we have $\alpha^{\mathrm{SOCP}}_{\Wif_t\rightarrow\Wlk_t}\leq 1-\gamma$ for some $\gamma>0$.
    \item (Nonbinding SOCP constraints in the boundary) the solution for the attacked states satisfies $\hat{\xbf}_\Iif\in\tKif({\xbf}_{\natural\Ilk})$.
\end{itemize}
Then, the following properties hold for the solution to \eqref{equ:primal_noise}, denoted as $(\hat{\xbf},\hat{\bbf})$:
\begin{enumerate}
    \item (No false inclusion) The solution $(\hat{\xbf},\hat{\bbf})$ has no false bad data inclusion (i.e., $\supp(\hat{\bbf})\subset \supp(\bbf_\natural)$) with probability greater than $1-\frac{c_0}{n_m}$, for some constant $c_0>0$.
    \item (Large bad data detection)  Let ${\Abfo}\coloneqq\begin{bmatrix}\Abf_{\Msf,\Xsf}&\Abf_{\Msf,\Xlk}\\\bfzero&\Abf_{\Mol,\Xlk}\\
        \bfzero&\Abf_{\Mad,\Xlk}\end{bmatrix}$ and $\Qbfo_{\Mad}=\begin{bmatrix}\Abfo&\Ibfot_{\Mad}\end{bmatrix}$, and $$g(\lambda)=n_m\lambda\left(\frac{1}{2\sqrt{{C_{\text{min}}}}}+\|\Ibf_b(\Qbfot_{\Mad}\Qbfo_{\Mad})^{-1}\Qbfot_{\Mad}\|_\infty\right)$$ be a threshold value, and let $\tbbf_{\Mad}=\Abf_{\Mad,\Xif}({\xbf}_{\natural\Iif}-\hat{\xbf}_{\Iif})$ be the error at the boundary. Then, all bad data with magnitude greater than $g(\lambda)$ will be detected (i.e., if $|\tilde{b}_{i}|>g(\lambda)$, then $|\hat{b}_{i}|>0$) with probability greater than $1-\frac{c_2}{m}$.
    \item (Bounded error) The estimator error is bounded by \begin{align*}
    \|\xbf_{\natural\Xsf\cup\Xlk}-\hat{\xbf}_{\Xsf\cup\Xlk}\|_2\leq t\frac{\sqrt{|\Xsf|+|\Xlk|+|\Mad|}}{C_{\text{min}}}+n_m\lambda\|\Ibf_x(\Qbfot_{\Mad}\Qbfo_{\Mad})^{-1}\Qbfot_{\Mad}\|_{\infty,2}
\end{align*}
with probability greater than $1-\exp\left(-\frac{c_1t^2}{\sigma^4}\right)$.
\end{enumerate}
\label{thm:lasso_est_cyber_socp_tree}
\end{theorem}

The proofs of Theorems \ref{thm:lasso_est_cyber_tree} and \ref{thm:lasso_est_cyber_socp_tree} are similar to those of Theorems \ref{thm:lasso_est_cyber} and \ref{thm:lasso_est_cyber_socp} in Section \ref{sec:proofs} and are omitted for brevity. As shown in our analysis, tree decomposition provides an efficient way to define the boundary between infected and safe nodes. Tree decomposition has been employed in semidefinite programming (SDP) to efficiently deal with network with chordal sparsity \cite{vandenberghe2015chordal}. The smaller the treewidth, the faster it is to solve SDP \cite{zohrizadehbconic}. Our analysis shows that with smaller treewidth, it is generally easier to certify robustness for SE. This is mainly due to the fact that the adhesion set is bounded by the treewidth, which limits the number of nodes that an infected bag can influence.

\section{Experimental details}
\label{sec:add_exp}
\textbf{Noisy measurements:} For each simulation, we randomly generate dense noise $\wbf$ and sparse bad data $\bbf$, and add them to the clean data according to \eqref{equ:mv}. The dense noise for each measurement is zero-mean Gaussian variable, with standard deviation of 1e-5 (per unit) for voltage magnitude measurements and 0.005 (per unit) for all the other measurements. The difference in standard deviation is due to the fact that voltage magnitude sensors have higher standards of accuracy compared to power meters. For the sparse bad data, its support is randomly selected among the line measurements. We randomly select a set of lines, whose branch flow measurements are all compromised accordingly. The values for the sparse noise can be arbitrarily large, and we assume these parameters are uniformly chosen from the set $[-4.25,-3.75]\cup[3.75,4.25]$ (per unit).

\textbf{Performance metrics:} We use the root-mean-square error (RMSE) as the metric for estimation accuracy, which is  defined as $\sqrt{\tfrac{1}{{n_b}}\sum_{i\in\Ncal}|v_i-\hat{v}_i|^2}$, where $v_i$ and $\hat{v}_i$ are the true and estimated complex voltage at bus $i\in\Ncal$. To evaluate the bad data detection accuracy, we use the F1 score, which is defined as $\tfrac{2*\text{precision}\times\text{recall}}{\text{precision}+\text{recall}}$, where {\it{precision}} is given by $\tfrac{\text{\#True positives } |\Jcal\cap\hat{\Jcal}|}{\text{\#Conditional positives } |{\hat{\Jcal}}|}$, and {\it{recall}} is given by $\tfrac{\text{\#True positives } |\Jcal\cap\hat{\Jcal}|}{\text{\#Conditional positives } |{\Jcal}|}$, and $\Jcal$ and $\hat{\Jcal}$ denote the true and estimated supports of bad data (\# indicates the number of elements). The F1 score is the harmonic average of the precision and recall, which reaches its best value at 1 (perfect precision and recall) and worst at 0. 

\textbf{Experimental setup:} We evaluate the proposed method (step-1 estimators include \eqref{equ:primal_l1}, \eqref{equ:primal_noise}, \eqref{equ:primal_socp} or  \eqref{equ:primal_lasso_socp}) combined with step-2 recovery method \eqref{equ:s2-phase} or \eqref{equ:s2-phase-lasso}), and compare it with the current practice of nonlinear least square (NLS) method based on Newton's algorithm. We use SeDuMi \cite{sturm1999using} as the optimization solver and the MATPOWER implementation of NLS. Throughout the experiment, we choose $\lambda$ in \eqref{equ:primal_noise} to be $3\times10^{-4}/{n_m}$, $\lambda_2$ in \eqref{equ:primal_lasso_socp} to be 0.1, and a bad data detection threshold of 0.01 for stage-1 estimators. After the removal of bad data (i.e., cleaning step), we perform the estimation with the remaining data. All the experiments are performed on a standard laptop with 3.3GHz Intel Core i7 and 16GB memory.

\textbf{Convergence issue of Newton's method:} We performed a simple experiment, where there is no noise in the measurements, and we use both Newton's method and our proposed method to estimate the state for the IEEE 300-bus system \cite{zimmerman2010matpower}. Since Newton's method depends on the initial point, we randomly generate an initial point, where we add a complex vector on top of the ground truth. The magnitude of each entry is uniformly chosen from $[1-\tau,1+\tau]$, and angle (in degrees) uniformly chosen from $[-100\times\tau,100\times\tau]$. We increase $\tau$ to enlarge the initialization distance. As shown in Figure \ref{fig:newton-comp-init}, as we increase $\tau$, Newton's method becomes less and less reliable. This can be due to several factors, for example, if the initial point is far from the ground truth, the algorithm can become stuck at a local optimal. On the contrary, our proposed method based on \eqref{equ:primal_noise} does not depend on the initial point and can recover the ground truth for all the experiments.

\begin{figure}[h!]
  \centering
  \begin{center}
    \includegraphics[width=.58\columnwidth,trim=0mm 0mm 0mm 0mm,clip]{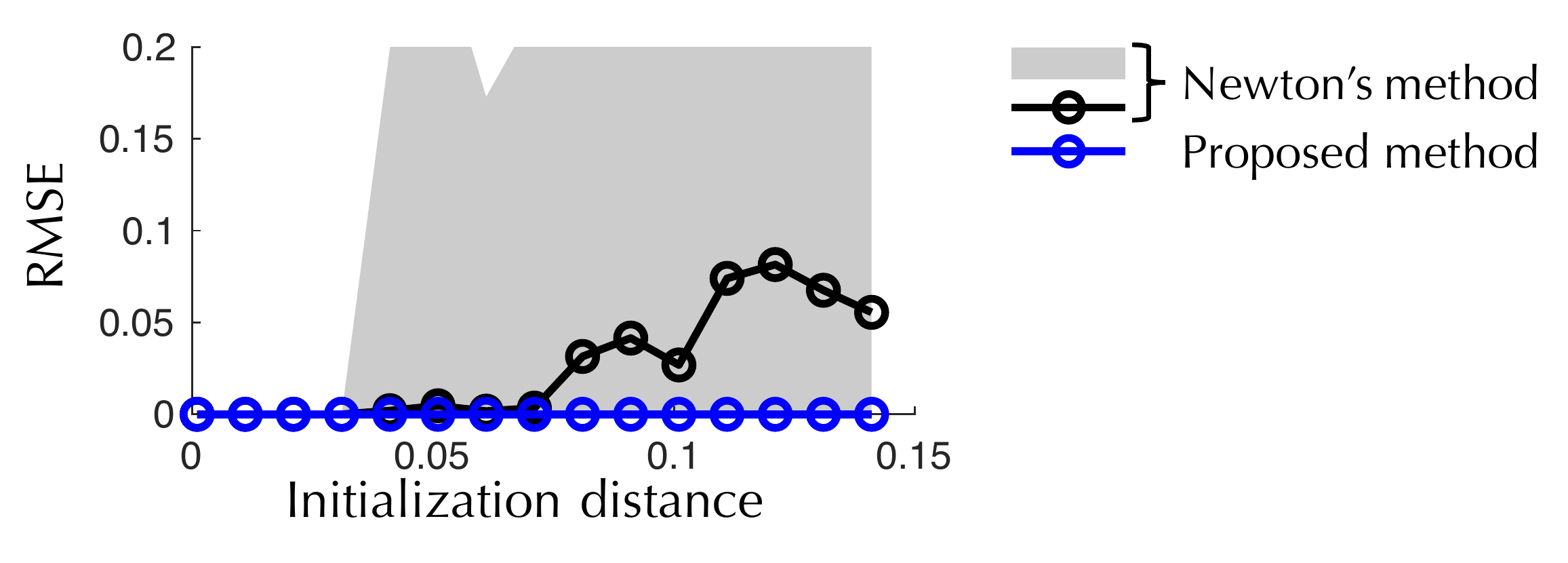}
  \end{center}
  \caption{Plots of RMSE against initilization distance $\tau$ for Newton's method. The RMSE is averaged over 20 simulations. For the Newton's method, we show both the mean performance (circled line) and the min/max range (black shades). }
  \label{fig:newton-comp-init}
\end{figure}

\textbf{Simulations on measurement redundancy:} In the main paper, we demonstrated the performance of \eqref{equ:primal_lasso_socp} for different sensor measurement profiles. We have tested three different methods to add additional sensors: the first method (Method 1) starts from a spanning tree of the network and incrementally adds a set of lines to the tree. In this method, each bus is equipped with only voltage magnitude measurements, and each line has 3 out of 4 branch flow measurements. The second method (Method 2) starts with the full network, where each node has voltage magnitude measurements and each line has one real and one reactive power measurements, and it grows the set of sensors by randomly adding branch measurements. The third method (Method 3) differs from Method 2 only in that it grows the set of sensors by randomly adding branch measurements as well as nodal power injections. In Figure \ref{fig:meas_profile_comp}, we compared the performance of \eqref{equ:primal_noise} with \eqref{equ:primal_lasso_socp} in terms of both estimation accuracy and bad data detection rates. It can be seen that \eqref{equ:primal_lasso_socp} consistently outperforms \eqref{equ:primal_noise} at different redundancy rates. We can also observe that Method 1 is more efficient in terms of improvement of performance with additional sensors.

\begin{figure}[h!]
  \centering
  \begin{center}
    \includegraphics[width=.75\columnwidth,trim=0mm 0mm 0mm 0mm,clip]{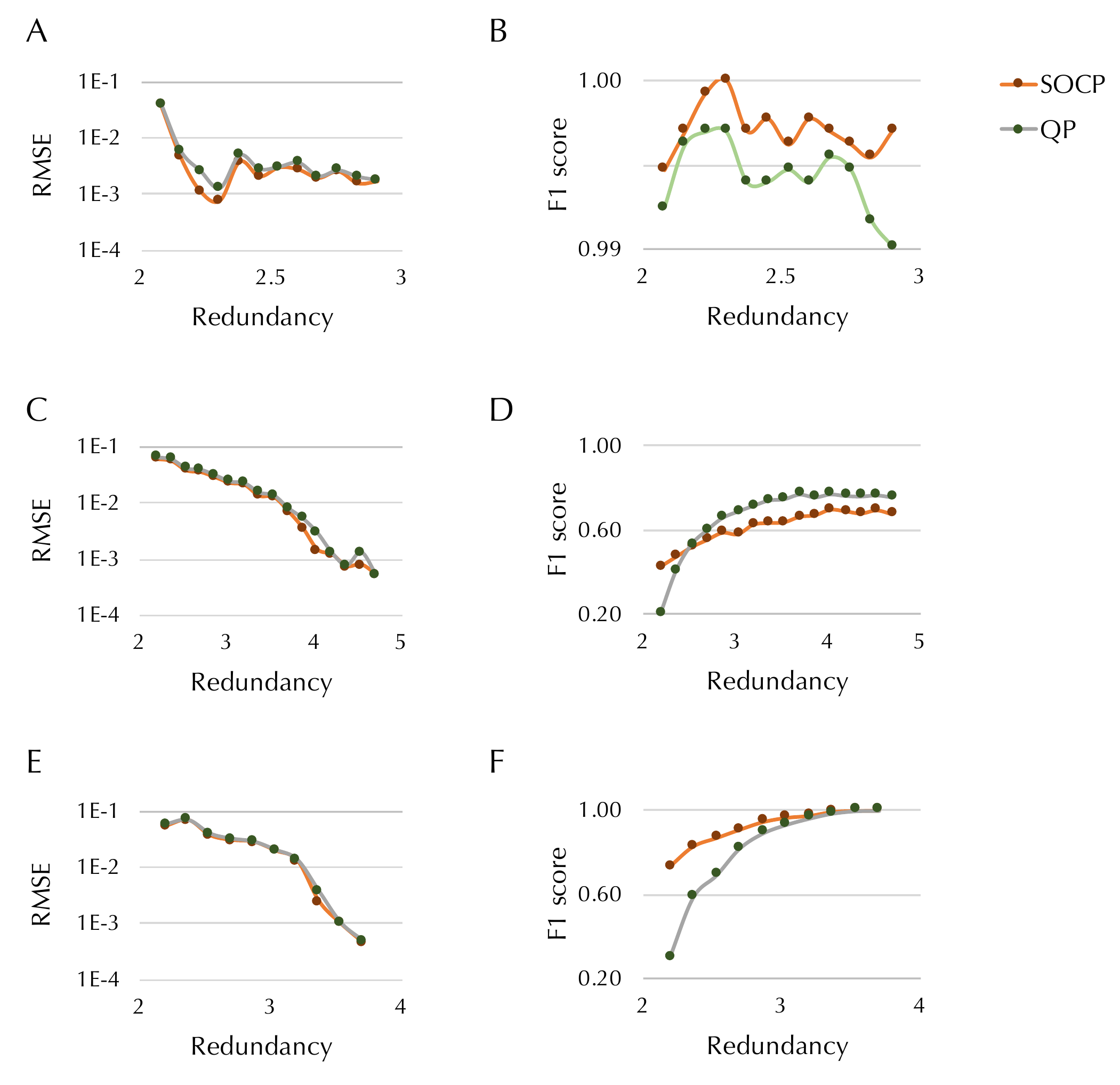}
  \end{center}
  \caption{Performance of proposed algorithms with different rates of measurement redundancy. Plots for different methods to add measurements: \textbf{(A, B)} Method 1, \textbf{(C, D)} Method 2, \textbf{(E, F)} Method 3. Results for \eqref{equ:primal_lasso_socp} (red) and \eqref{equ:primal_noise} (green) are shown, which are averaged over 20 independent simulations. }
  \label{fig:meas_profile_comp}
\end{figure}

\textbf{Visualization of vulnerability maps for different measurement profiles:} In Figure 9 from the main text, we show statistics regarding vulnerability index and critical index for different measurement profiles. Figures \ref{fig:inco_comp} and \ref{fig:ci_comp} show the geographical distributions of VI and CI, respectively. It can be seen that \eqref{equ:primal_lasso_socp} is consistently more robust in terms of VI and CI than \eqref{equ:primal_noise}. This is also theoretically proven in Proposition \ref{prop:soc_better}. We also see that the more vulnerable lines exist, the higher the bus critical index tends to be. By comparing Figure \textbf{(B, G)} with Figure \textbf{(C, H)}, we see that the inclusion of nodal power injections is likely to cause vulnerable lines. By including more branch flow measurements, as shown in Figure \textbf{(A, C, D)} and Figure \textbf{(F, H, I)}, or more voltage magnitude measurements, as shown in Figure \textbf{(B, C, E)} or Figure \textbf{(G, H, J)}, it is more likely to robustify the network.

\begin{figure}[h!]
  \centering
  \begin{center}
    \includegraphics[width=.45\columnwidth,trim=0mm 0mm 0mm 0mm,clip]{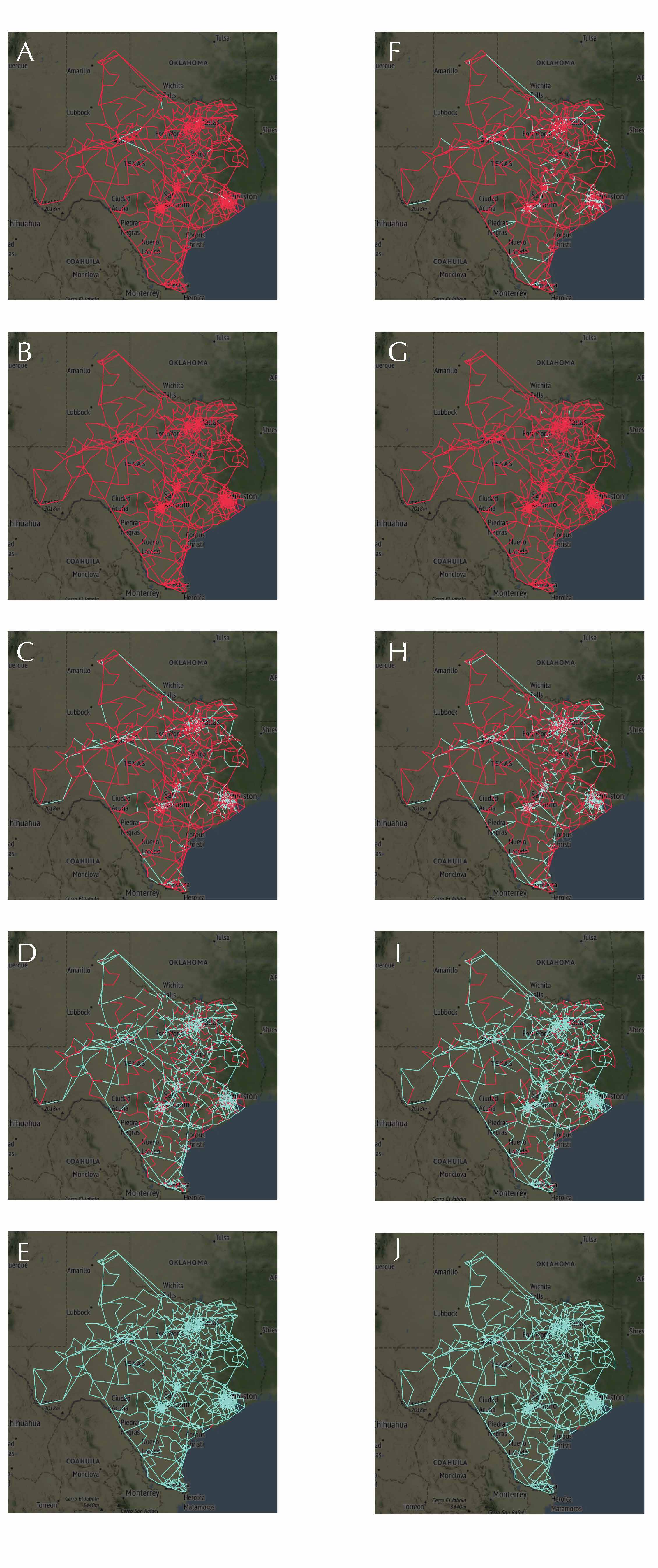}
  \end{center}
  \caption{Vulnerability maps for different measurement profiles and optimization techniques. \textbf{(A--E)} and \textbf{(F--J)} are series of maps without and with the SOCs, respectively. \textbf{(A, F)}, \textbf{(C, H)} and \textbf{(D, I)} correspond to PV or PQ nodal measurements together with 2, 3, and 4 branch power flows, respectively. \textbf{(B, G)} and \textbf{(E, J)} correspond to only PQ or only voltage magnitude nodal measurements with 3 branch power flows, respectively. }
  \label{fig:inco_comp}
\end{figure}

\begin{figure}[h!]
  \centering
  \begin{center}
    \includegraphics[width=.63\columnwidth,trim=-60mm 0mm 5mm 0mm,clip]{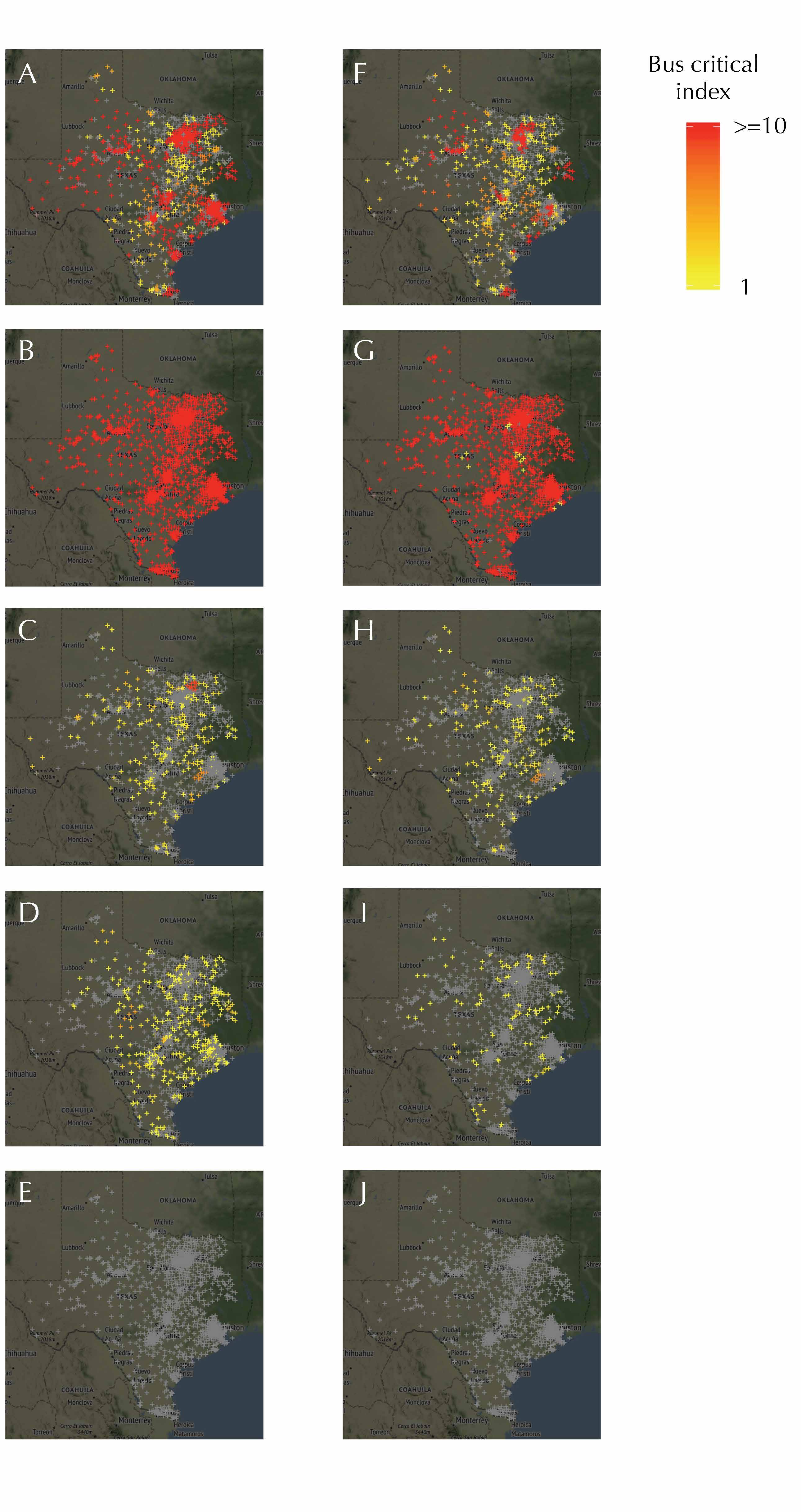}
  \end{center}
  \caption{Bus critical index maps for different measurement profiles and optimization techniques. \textbf{(A--E)} and \textbf{(F--J)} are series of maps without and with the SOCs, respectively. \textbf{(A, F)}, \textbf{(C, H)} and \textbf{(D, I)} correspond to PV or PQ nodal measurements together with 2, 3, and 4 branch power flows, respectively. \textbf{(B, G)} and \textbf{(E, J)} correspond to only PQ or only voltage magnitude nodal measurements with 3 branch power flows, respectively. Color indicates low (yellowish) to high (reddish) critical index. If the critical index is 0, which occurs when attacking the bus does not affect any of its neighbors, the grey color is shown.}
  \label{fig:ci_comp}
\end{figure}

\section{Proofs}
\label{sec:proofs}

\subsection{Proof of Theorem \ref{thm:lasso_est_cyber}}
\label{sec:proof_thm_qp}

For an arbitrary set of attacked measurements $\Mat$, their boundary $\Mbd\coloneqq\Mbi\cup\Mbo$ and unaffected measurements $\Msf$, as well as the associated variables $\xbf_{\Iat}$, $\xbf_\Ibd$ and $\xbf_\Isf$, respectively, we design the primal-dual witness (PDW) process as follows:

(1) Set $\hat{\bbf}_{\Msf}=\bfzero$ and $\hat{\bbf}_{\Mbo}=\bfzero$;

(2) Determine $\hat{\xbf}=\begin{bmatrix}\hat{\xbf}_\Isf^\top&\hat{\xbf}_{\Ibd}^\top&\hat{\xbf}_{\Iat}^\top\end{bmatrix}^\top$ and $\hat{\bbf}=\begin{bmatrix}\bfzero^\top&\bfzero^\top&\hat{\bbf}_{\Mbi}^\top&\hat{\bbf}_{\Mat}^\top\end{bmatrix}^\top$ by solving the following program:
    \begin{equation}
        \min_{{\bbf}\in\Rbb^{n_m},\xbf\in\Rbb^{n_x}}\frac{1}{2n_m}\left\|\begin{bmatrix}\ybf_{\Msf}\\\ybf_{\Mbo}\\\ybf_{\Mbi}\\\ybf_{\Mat}\end{bmatrix}\!\!-\!\!\begin{bmatrix}\Abf_{\Msf,\Xsf}&\Abf_{\Msf,\Xbd}&\bfzero\\\bfzero&\Abf_{\Mbo,\Xbd}&\bfzero\\
        \bfzero&\Abf_{\Mbi,\Xbd}&\Abf_{\Mbi,\Xat}\\\bfzero&\bfzero&\Abf_{\Mat,\Xat}\end{bmatrix}\begin{bmatrix}\xbf_{\Isf}\\\xbf_{\Ibd}\\\xbf_{\Iat}\end{bmatrix}\!\!-\!\!\begin{bmatrix}\bfzero\\\bfzero\\\bbf_{\Mbi}\\\bbf_{\Mat}\end{bmatrix}\right\|_2^2\!\!+\lambda\left\|\begin{bmatrix}\bbf_{\Mbi}\\\bbf_{\Mat}\end{bmatrix}\right\|_1,
        \label{equ:pdw-all}
    \end{equation}
    and $\hat{\hbf}_{\Mbi}\in\partial\|\hat{\bbf}_{\Mbi}\|_1$ and $\hat{\hbf}_{\Mat}\in\partial\|\hat{\bbf}_{\Mat}\|_1$ satisfying the optimality conditions
    \begin{subequations}
    \label{equ:pdw_opt_h}
    \begin{align}
        -\frac{1}{n_m}(\ybf_{\Mat}-\Abf_{\Mat,\Xat}\hat{\xbf}_{\Iat}-\hat{\bbf}_{\Mat})+\lambda\hat{\hbf}_{\Mat}&=\bfzero,\\
        -\frac{1}{n_m}\left(\ybf_{\Mbi}-\Abf_{\Mbi,\Xbd}\hat{\xbf}_{\Ibd}-\Abf_{\Mbi,\Xat}\hat{\xbf}_{\Iat}-\hat{\bbf}_{\Mbi}\right)+\lambda\hat{\hbf}_{\Mbi}&=\bfzero.
    \end{align}
    \end{subequations}

(3) Solve $(\hat{\hbf}_{\Msf},\hat{\hbf}_{\Mbo})$ via the zero-subgradient equation:
    \begin{equation}
        -\frac{1}{n_m}\left(\ybf-\Abf\hat{\xbf}-\hat{\bbf}\right)+\lambda\hat{\hbf}=\bfzero,
        \label{equ:pdw_opt_h_all}
    \end{equation}
    where $\hat{\xbf}=\begin{bmatrix}\hat{\xbf}_{\Bsf}^\top&\hat{\xbf}_{\Bbd}^\top&\hat{\xbf}_{\Bat}^\top\end{bmatrix}^\top$ and $\hat{\bbf}=\begin{bmatrix}\bfzero^\top&\bfzero^\top&\hat{\bbf}_{\Mbi}^\top&\hat{\bbf}_{\Mat}^\top\end{bmatrix}^\top$ are solutions obtained in \eqref{equ:pdw-all}, and $\hat{\hbf}=\begin{bmatrix}\hat{\hbf}_{\Msf}^\top&\hat{\hbf}_{\Mbo}^\top&\hat{\hbf}_{\Mbi}^\top&\hat{\hbf}_{\Mat}^\top\end{bmatrix}^\top$ where $(\hat{\hbf}_{\Mbi},\hat{\hbf}_{\Mat})$ are given in \eqref{equ:pdw_opt_h}. Check whether strict feasibility conditions $\|\hat{\hbf}_{\Msf}\|_\infty<1$ and $\|\hat{\hbf}_{\Mbo}\|_\infty<1$ hold.

\begin{lemma}
If the PDW procedure succeeds, then $(\hat{\xbf},\hat{\bbf})$ that is optimal for \eqref{equ:pdw-all} is also optimal for \eqref{equ:primal_noise}. Furthermore, for any optimal solution $(\tilde{\xbf},\tilde{\bbf})$, if $\hat{\xbf}_\Iat=\tilde{\xbf}_\Iat$, we must have $\hat{\xbf}_\Isf=\tilde{\xbf}_\Isf$ and $\hat{\xbf}_\Ibd=\tilde{\xbf}_\Ibd$ (i.e., uniqueness property in the weak sense).
\label{lem:unique_pdw_sol}
\end{lemma}
\begin{proof}
The KKT conditions of \eqref{equ:primal_noise} for a given primal-dual pair $(\hat{\xbf},\hat{\bbf})$ and $\hat{\hbf}$ are given by:
\begin{subequations}
\label{equ:s1-lasso-kkt}
\begin{align}
    \Abf^\top(\ybf-\Abf\hat{\xbf}-\hat{\bbf})&=\bfzero,\label{equ:station_lasso_x}\\
    -\frac{1}{n_m}\left(\ybf-\Abf\hat{\xbf}-\hat{\bbf}\right)+\lambda\hat{\hbf}&=\bfzero,\label{equ:station_lasso_b}\\
    \|\hat{\hbf}\|_\infty&\leq 1
\end{align}
\end{subequations}
If PDW succeeds, then the optimality conditions \eqref{equ:s1-lasso-kkt} are satisfied, which certify the optimality of  $(\hat{\xbf},\hat{\bbf})$. The subgradient $\hat{\hbf}$ satisfies $\|\hat{\hbf}_{\Msf}\|_\infty<1$,  $\|\hat{\hbf}_{\Mbo}\|_\infty<1$ and $\left\langle \hat{\hbf},\hat{\bbf}\right\rangle=\|\hat{\bbf}\|_1$. Now, let $(\tilde{\xbf},\tilde{\bbf})$ be any other optimal, and let $F(\xbf,\bbf)=\frac{1}{2n_m}\|\ybf-\Abf\xbf-\bbf\|_2^2$, then we have
\begin{equation*}
    F(\hat{\xbf},\hat{\bbf})+\lambda\left\langle \hat{\hbf},\hat{\bbf}\right\rangle=F(\tilde{\xbf},\tilde{\bbf})+\lambda\|\tilde{\bbf}\|_1,
\end{equation*}
and hence,
\begin{equation*}
    F(\hat{\xbf},\hat{\bbf})+\lambda\left\langle \hat{\hbf},\hat{\bbf}-\tilde{\bbf}\right\rangle=F(\tilde{\xbf},\tilde{\bbf})+\lambda\left(\|\tilde{\bbf}\|_1-\left\langle \hat{\hbf},\tilde{\bbf}\right\rangle\right).
\end{equation*}
By the optimality conditions in \eqref{equ:s1-lasso-kkt}, we have 
$\lambda\hat{\hbf}=-\nabla_b F(\hat{\xbf},\hat{\bbf})=\frac{1}{n_m}(\ybf-\Abf\hat{\xbf}-\hat{\bbf})$, which implies that
\begin{align*}
    F(\hat{\xbf},\hat{\bbf})-\left\langle \nabla_b F(\hat{\xbf},\hat{\bbf}),\hat{\bbf}-\tilde{\bbf}\right\rangle-F(\tilde{\xbf},\tilde{\bbf})=\lambda\left(\|\tilde{\bbf}\|_1-\left\langle \hat{\hbf},\tilde{\bbf}\right\rangle\right)\leq 0
\end{align*}
due to convexity. Therefore, $\|\tilde{\bbf}\|_1\leq\left\langle \hat{\hbf},\tilde{\bbf}\right\rangle$. Since by Holder's inequality, we also have $\left\langle \hat{\hbf},\tilde{\bbf}\right\rangle\leq\|\hat{\hbf}\|_\infty\|\tilde{\bbf}\|_1$, and $\|\hat{\hbf}\|_\infty\leq 1$, it holds that $\|\tilde{\bbf}\|_1=\left\langle \hat{\hbf},\tilde{\bbf}\right\rangle$. Since by the success of PDW, $\|\hat{\hbf}_{\Msf}\|_\infty<1$ and  $\|\hat{\hbf}_{\Mbo}\|_\infty<1$, we have $\tilde{\bbf}_j=0$ for all $j\in \Msf\cup\Mbo$. To show the weak uniqueness, let $(\tilde{\xbf},\tilde{\bbf})$ be another optimal solution, and assume that $\hat{\xbf}_{\Iat}=\tilde{\xbf}_{\Iat}$. Then, by fixing ${\xbf}_\Iat$ in the optimization \eqref{equ:pdw-all} as $\hat{\xbf}_{\Iat}$ and by the lower eigenvalue condition, the the function is strictly convex in $\xbf_{\Isf}$, $\xbf_{\Ibd}$ and $\bbf_{\Mbi}$.
\end{proof}

\subsection*{Proof of Theorem \ref{thm:lasso_est_cyber}}
\begin{proof}
\textbf{Part 1)}: By the construction of PDW, we have $\hat{\bbf}_{\Msf}=\bbf_{\natural \Msf}=\bfzero$ and $\hat{\bbf}_{\Mbo}=\bbf_{\natural \Mbo}=\bfzero$. In the following, we allow the optimal solution $\hat{\xbf}_\Iat$ and $\hat{\bbf}_\Mat$ of \eqref{equ:pdw-all}  to take any value. Thus, for any given $\hat{\xbf}_\Iat$ and $\hat{\bbf}_\Mat$, we can fix $\xbf_\Iat$ and ${\bbf}_\Mat$ in \eqref{equ:pdw-all} and solve the following smaller program:
\begin{equation}
        \min_{{\bbf}_{\Mbi},\xbf_{\Isf},\xbf_{\Ibd}}\frac{1}{2n_m}\Bigg\|\underbrace{\begin{bmatrix}\ybf_{\Msf}\\\ybf_{\Mbo}\\\zbf_{\Mbi}\end{bmatrix}}_{\zbfo}-\underbrace{\begin{bmatrix}\Abf_{\Msf,\Xsf}&\Abf_{\Msf,\Xbd}\\\bfzero&\Abf_{\Mbo,\Xbd}\\
        \bfzero&\Abf_{\Mbi,\Xbd}\end{bmatrix}}_{\Abfo}\underbrace{\begin{bmatrix}\xbf_{\Isf}\\\xbf_{\Ibd}\end{bmatrix}}_{\xbfo}-\begin{bmatrix}\bfzero\\\bfzero\\\bbf_{\Mbi}\end{bmatrix}\Bigg\|_2^2+\lambda\left\|\bbf_{\Mbi}\right\|_1,
        \label{equ:pdw-all-small}
    \end{equation}
where $\zbf_{\Mbi}=\ybf_{\Mbi}-\Abf_{\Mbi,\Xat}\hat{\xbf}_\Iat=\Abf_{\Mbi,\Xbd}\xbf_{\natural\Ibd}+\tbbf_{\Mbi}$ and $\tbbf_{\Mbi}=\Abf_{\Mbi,\Xat}(\xbf_{\natural\Iat}-\hat{\xbf}_{\Iat})$. Let $\Ibfo$ be an identity matrix of size $n_m-|\Mat|$, and $\xbfo$ and $\wbfo$ be the subvectors of $\xbf$ and $\wbf$ indexed by $\Msf\cup\Mbo\cup\Mbi$, respectively. Thus, we have $\zbfo=\Abfo\xbfo_\natural+\wbfo_\natural+\Ibfot_{\Mbi}\tbbf_{\Mbi}$. The solution $(\xbf_\Isf,\xbf_\Ibd,\bbf_{\Mbi})$ of \eqref{equ:pdw-all-small} is unique and coincides with that of \eqref{equ:pdw-all} due to the lower eigenvalue condition. Thus, the zero-subgradient condition \eqref{equ:pdw_opt_h} is satisfied, which together with \eqref{equ:pdw_opt_h_all} can be written as:
\begin{equation}
    -\frac{1}{n_m}\left(\begin{bmatrix}\Abf_{\Msf,\Xsf}&\Abf_{\Msf,\Xbd}\\\bfzero&\Abf_{\Mbo,\Xbd}\\
        \bfzero&\Abf_{\Mbi,\Xbd}\end{bmatrix}\begin{bmatrix}\xbf_{\natural\Isf}-\hat{\xbf}_{\Isf}\\\xbf_{\natural\Ibd}-\hat{\xbf}_{\Ibd}\end{bmatrix}+\begin{bmatrix}
    \bfzero\\\bfzero\\\tbbf_{\Mbi}-\hat{\bbf}_{\Mbi}\end{bmatrix}\right)-\frac{1}{n_m}\begin{bmatrix}
    \wbf_{\natural\Msf}\\
    \wbf_{\natural\Mbo}\\
    \wbf_{\natural\Mbi}
    \end{bmatrix}+\lambda\begin{bmatrix}
    \hat{\hbf}_\Msf\\\hat{\hbf}_{\Mbo}\\\hat{\hbf}_{\Mbi}
    \end{bmatrix}=\bfzero.
\end{equation}
We can partition the above relation into equations indexed by $\Mbi$, which can be rearranged as:
\begin{equation}
    \hat{\hbf}_{\Mbi}=\frac{1}{n_m\lambda}\begin{bmatrix}
    \Ibfo_{\Mbi}\Abfo& \Ibfo_{\Mbi}\Ibfot_{\Mbi}
    \end{bmatrix}\begin{bmatrix}
    \xbfo_\natural-{\hxbfo}\\ \tbbf_{\Mbi}-\hat{\bbf}_{\Mbi}
    \end{bmatrix}
    +\frac{1}{n_m\lambda}\Ibfo_{\Mbi}\wbfo_\natural,
    \label{equ:cond_b_Mbi}
\end{equation}
as well as those indexed by $\Msf\cup\Mbo$, which can be solved for $\hat{\hbf}_{\Msf\cup\Mbo}=\begin{bmatrix}\hat{\hbf}_{\Msf}^\top&\hat{\hbf}_{\Mbo}^\top\end{bmatrix}^\top$:
\begin{equation}
    \hat{\hbf}_{\Msf\cup\Mbo}=\frac{1}{n_m\lambda}\Ibfo_{\Msf\cup\Mbo}\left(\Abfo(\xbfo_\natural-{\hxbfo})+\wbfo_\natural\right).
    \label{equ:h_j_c_equ}
\end{equation}
Since $\hxbfo$ is the optimal solution of \eqref{equ:pdw-all-small}, it satisfies the optimality condition:
\begin{equation}
    \Abfot\left(\Abfo(\xbfo_\natural-\hxbfo)+\wbfo_\natural+\Ibfot_{\Mbi}(\tbbf_{\Mbi}-\hat{\bbf}_{\Mbi})\right)=\bfzero
    \label{equ:opt_x_small_cond}
\end{equation}
Combining \eqref{equ:cond_b_Mbi}, \eqref{equ:h_j_c_equ} and \eqref{equ:opt_x_small_cond} and after some elementary operations, we have
\begin{equation}
    \Abfot_{\Msf\cup\Mbo}\hat{\hbf}_{\Msf\cup\Mbo}+\Abfot_{\Mbi}\hat{\hbf}_{\Mbi}=\bfzero.
\end{equation}
By Lemma \ref{lemma:local_sufficient_cond}, for any $\hat{\hbf}_{\Mbi}\in\partial\|\hat{\bbf}_{\Mbi}\|_1$, there always exists $\hat{\hbf}_{\Msf\cup\Mbo}$ such that $\|\hat{\hbf}_{\Msf\cup\Mbo}\|_\infty<1$. Thus, the strict feasibility condition is satisfied deterministically.

\textbf{Part 2)}:  By the lower eigenvalue condition  the and definition of $\Qbfo_{\Mbi}=\begin{bmatrix}\Abfo&\Ibfot_{\Mbi}\end{bmatrix}$, we can solve \eqref{equ:cond_b_Mbi} and \eqref{equ:opt_x_small_cond}:
\begin{align}
    \bDelta&\coloneqq\begin{bmatrix}
    \xbfo_\natural-{\hxbfo}\\ \tbbf_{\Mbi}-\hat{\bbf}_{\Mbi}
    \end{bmatrix}={-(\Qbfot_{\Mbi}\Qbfo_{\Mbi})^{-1}\Qbfot_{\Mbi}\wbfo_\natural+n_m\lambda(\Qbfot_{\Mbi}\Qbfo_{\Mbi})^{-1}\begin{bmatrix}\bfzero\\
    \hat{\hbf}_{\Mbi}\end{bmatrix}}\label{equ:delta_err}
\end{align}
Let $\Ibf_x$ and $\Ibf_b$ denote the matrices that consist of the first $|\Xsf|+|\Xbd|$ rows and the last $|\Mbi|$ rows of the identity matrix of size $|\Xsf|+|\Xbd|+|\Mbi|$, respectively. Then, we can bound the estimation error $\bDelta$ in \eqref{equ:delta_err}. First, we bound the infinity norm of $\tbbf_{\Mbi}-\hat{\bbf}_{\Mbi}=\Ibf_b\bDelta$. By triangle inequality,
\begin{align}
    \|\Ibf_b\bDelta\|_\infty\leq \|\Ibf_b(\Qbfot_{\Mbi}\Qbfo_{\Mbi})^{-1}\Qbfot_{\Mbi}\wbfo_\natural\|_\infty+n_m\lambda\|\Ibf_b(\Qbfot_{\Mbi}\Qbfo_{\Mbi})^{-1}\Ibf_b^\top\|_\infty.
\end{align}
Since the second term is deterministic, we will now bound the first term. By the normalized measurement condition \eqref{def:norm_measure} (we assume all measurement vectors are normalized by 1 without loss of generality) and the lower eigenvalue condition, each entry of $(\Qbfot_{\Mbi}\Qbfo_{\Mbi})^{-1}\Qbfot_{\Mbi}\wbfo_\natural$ is zero-mean sub-Gaussian with parameter at most
\begin{equation}
    {\sigma^2}\|(\Qbfot_{\Mbi}\Qbfo_{\Mbi})^{-1}\|_2\leq\frac{\sigma^2}{C_{\text{min}}}.
\end{equation}
Thus, by the union bound, we have
\begin{equation}
    \pr\left(\|\Ibf_b(\Qbfot_{\Mbi}\Qbfo_{\Mbi})^{-1}\Qbfot_{\Mbi}\wbfo_\natural\|_\infty>t\right)\leq 2\exp\left(-\frac{C_{\text{min}}t^2}{2\sigma^2}+\log |\Mbi|\right).
\end{equation}
Then, set $t=\frac{n_m\lambda}{2\sqrt{{C_{\text{min}}}}}$, and note that by our choice of $\lambda$, we have $\frac{C_{\text{min}}t^2}{2\sigma^2}>\log |\Mbi|$. Thus, we conclude that
\begin{equation}
    \|\tbbf_{\Mbi}-\hat{\bbf}_{\Mbi}\|_\infty\leq n_m\lambda\left(\frac{1}{2\sqrt{{C_{\text{min}}}}}+\|\Ibf_b(\Qbfot_{\Mbi}\Qbfo_{\Mbi})^{-1}\Ibf_b^\top\|_\infty\right)
\end{equation}
with probability greater than $1-2\exp(-c_2n_m^2\lambda^2)$. This indicates that all bad data entries greater than
\begin{equation}
    g(\lambda)=n_m\lambda\left(\frac{1}{2\sqrt{{C_{\text{min}}}}}+\|\Ibf_b(\Qbfot_{\Mbi}\Qbfo_{\Mbi})^{-1}\Ibf_b^\top\|_\infty\right)
\end{equation}
will be detected by $\hat{\bbf}_{\Mbi}$.

\textbf{Part 3)}: Now, we bound the $\ell_2$ norm of the signal error $\xbfo_\natural-{\hxbfo}=\Ibf_x\bDelta$,
\begin{align}
    \|\Ibf_x\bDelta\|_2\leq \|\Ibf_x(\Qbfot_{\Mbi}\Qbfo_{\Mbi})^{-1}\Qbfot_{\Mbi}\wbfo_\natural\|_2+n_m\lambda\|\Ibf_x(\Qbfot_{\Mbi}\Qbfo_{\Mbi})^{-1}\Ibf_b^\top\|_{\infty,2}.
\end{align}
For the first term, by the application of standard sub-gaussian concentration, 
\begin{equation*}
    \pr\left(\|\Ibf_x(\Qbfot_{\Mbi}\Qbfo_{\Mbi})^{-1}\Qbfot_{\Mbi}\wbfo_\natural\|_2>\|\Ibf_x(\Qbfot_{\Mbi}\Qbfo_{\Mbi})^{-1}\Qbfot_{\Mbi}\|_F+t\|\Ibf_x(\Qbfot_{\Mbi}\Qbfo_{\Mbi})^{-1}\Qbfot_{\Mbi}\|_2\right)
\end{equation*}
is upper bounded by $\exp\left(-\frac{c_1t^2}{\sigma^4}\right)$. Since $$\|\Ibf_x(\Qbfot_{\Mbi}\Qbfo_{\Mbi})^{-1}\Qbfot_{\Mbi}\|_F\leq \|\Ibf_x\|_2\|(\Qbfot_{\Mbi}\Qbfo_{\Mbi})^{-1}\|_2\|\Qbfot_{\Mbi}\|_F\leq \frac{\sqrt{|\Xsf|+|\Xbd|+|\Mbi|}}{C_{\text{min}}}$$ due to the lower eigenvalue condition and the normalized measurement condition, and similarly it holds that $$\|\Ibf_x(\Qbfot_{\Mbi}\Qbfo_{\Mbi})^{-1}\Qbfot_{\Mbi}\|_2\leq \|\Ibf_x\|_2\|(\Qbfot_{\Mbi}\Qbfo_{\Mbi})^{-1}\|_2\|\Qbfot_{\Mbi}\|_F\leq\frac{\sqrt{|\Xsf|+|\Xbd|+|\Mbi|}}{C_{\text{min}}}.$$
Moreover,
\begin{equation*}
    \pr\left(\|\Ibf_x(\Qbfot_{\Mbi}\Qbfo_{\Mbi})^{-1}\Qbfot_{\Mbi}\wbfo_\natural\|_2>t\frac{\sqrt{|\Xsf|+|\Xbd|+|\Mbi|}}{C_{\text{min}}}\right)\leq \exp\left(-\frac{c_1t^2}{\sigma^4}\right).
\end{equation*}
Together, we conclude that 
\begin{align}
    \|\xbf_\natural-\hat{\xbf}\|_2\leq t\frac{\sqrt{|\Xsf|+|\Xbd|+|\Mbi|}}{C_{\text{min}}}+n_m\lambda\|\Ibf_x(\Qbfot_{\Mbi}\Qbfo_{\Mbi})^{-1}\Ibf_b^\top\|_{\infty,2}
\end{align}
with probability greater than $1-\exp\left(-\frac{c_1t^2}{\sigma^4}\right)$.

\end{proof}

\begin{lemma}
Suppose that $\Qbfot_{\Mbi}\Qbfo_{\Mbi}$ is invertible, where $\Qbfo_{\Mbi}=\begin{bmatrix}\Abfo&\Ibfot_{\Mbi}\end{bmatrix}$. Then, it holds that
\begin{equation}
    \Ibf_{\Msf\cup\Mbo}\Abfo \Ibf_x(\Qbfot_{\Mbi}\Qbfo_{\Mbi})^{-1}\Ibf_b^\top=-\Abf_{\Msf\cup\Mbo}^{\top+}\Abf_{\Mbi}^\top.
\end{equation}
\label{lem:mutual_inco_lasso_trans}
\end{lemma}
\begin{proof}
By the definition of $\Qbfo_{\Mbi}$ and block matrix inversion formula, we have
\begin{align*}
    &\Ibf_x(\Qbfot_{\Mbi}\Qbfo_{\Mbi})^{-1}\Ibf_b^\top\\
    &=-(\Abfot\Abfo)^{-1}\Abf_{\Mbi}^\top(\Ibf-\Abf_{\Mbi}(\Abfot\Abfo)^{-1}\Abf_{\Mbi}^\top)^{-1}\\
    &=-(\Abfot\Abfo)^{-1}\Abf_{\Mbi}^\top(\Ibf+\Abf_{\Mbi}(\Abf_{\Msf\cup\Mbo}^\top\Abf_{\Msf\cup\Mbo})^{-1}\Abf_{\Mbi}^\top)\\
    &=-(\Abfot\Abfo)^{-1}(\Ibf+\Abf_{\Mbi}^\top\Abf_{\Mbi}(\Abf_{\Msf\cup\Mbo}^\top\Abf_{\Msf\cup\Mbo})^{-1})\Abf_{\Mbi}^\top\\
    &=-(\Abf_{\Msf\cup\Mbo}^\top\Abf_{\Msf\cup\Mbo})^{-1}\Abf_{\Mbi}^\top,
\end{align*}
where the first equation follows from the Sherman-Morrison-Woodbury formula and the rest are elementary operations.
\end{proof}

\begin{lemma}
Suppose that $\Qbfot_{\Mbi}\Qbfo_{\Mbi}$ is invertible. Then, it holds that
\begin{equation}
    \Ibf_b(\Qbfot_{\Mbi}\Qbfo_{\Mbi})^{-1}\Ibf_b^\top=\Ibf+\Abf_{\Mbi}(\Abf_{\Msf\cup\Mbo}^\top\Abf_{\Msf\cup\Mbo})^{-1}\Abf_{\Mbi}^\top
\end{equation}
\end{lemma}
\begin{proof}
By the definition of $\Qbfo_{\Mbi}$ and block matrix inversion formula, we have
\begin{align*}
    \Ibf_b(\Qbfot_{\Mbi}\Qbfo_{\Mbi})^{-1}\Ibf_b^\top&=(\Ibf-\Abf_{\Mbi}(\Abfot\Abfo)^{-1}\Abf_{\Mbi}^\top)^{-1}\\
    &=\Ibf+\Abf_{\Mbi}(\Abf_{\Msf\cup\Mbo}^\top\Abf_{\Msf\cup\Mbo})^{-1}\Abf_{\Mbi}^\top,
\end{align*}
where the second equation follows from the Sherman-Morrison-Woodbury formula. 
\end{proof}

\subsection{Proof of Theorem \ref{thm:lasso_est_cyber_socp}}
\label{sec:proof_thm_socp}

For an arbitrary set of attacked measurements $\Mat$, their boundary $\Mbd\coloneqq\Mbi\cup\Mbo$ and unaffected measurements $\Msf$, as well as the associated variables $\xbf_{\Iat}$, $\xbf_\Ibd$ and $\xbf_\Isf$, respectively, we design the primal-dual witness process as follows:

1) Set $\hat{\bbf}_{\Msf}=\bfzero$ and $\hat{\bbf}_{\Mbo}=\bfzero$;

2) Determine $\hat{\xbf}=\begin{bmatrix}\hat{\xbf}_\Isf^\top&\hat{\xbf}_{\Ibd}^\top&\hat{\xbf}_{\Iat}^\top\end{bmatrix}^\top$ and $\hat{\bbf}=\begin{bmatrix}\bfzero^\top&\bfzero^\top&\hat{\bbf}_{\Mbi}^\top&\hat{\bbf}_{\Mat}^\top\end{bmatrix}^\top$ by solving the following program:
    \begin{subequations}
    \label{equ:pdw-all-socp}
	\begin{align}
    &\underset{{\bbf}\in\Rbb^{n_m},\xbf\in\Rbb^{n_x}}{\text{min~~~}}
	&& \frac{1}{2n_m}\left\|\begin{bmatrix}\ybf_{\Msf}\\\ybf_{\Mbo}\\\ybf_{\Mbi}\\\ybf_{\Mat}\end{bmatrix}\!\!-\!\!\begin{bmatrix}\Abf_{\Msf,\Xsf}&\Abf_{\Msf,\Xbd}&\bfzero\\\bfzero&\Abf_{\Mbo,\Xbd}&\bfzero\\
        \bfzero&\Abf_{\Mbi,\Xbd}&\Abf_{\Mbi,\Xat}\\\bfzero&\bfzero&\Abf_{\Mat,\Xat}\end{bmatrix}\begin{bmatrix}\xbf_{\Isf}\\\xbf_{\Ibd}\\\xbf_{\Iat}\end{bmatrix}\!\!-\!\!\begin{bmatrix}\bfzero\\\bfzero\\\bbf_{\Mbi}\\\bbf_{\Mat}\end{bmatrix}\right\|_2^2\!\!\!+\!\lambda\left\|\begin{bmatrix}\bbf_{\Mbi}\\\bbf_{\Mat}\end{bmatrix}\right\|_1,\\
	& \;\;\text{subject to~~~}
	&&\hspace{-0.5cm}\cbf_\ell^\top\xbf\geq\left\|\Dbf_\ell\xbf\right\|_2,\qquad\qquad\forall \ell\in\Lcal,
	\end{align}
\end{subequations}
    and $\hat{\hbf}_{\Mbi}\in\partial\|\hat{\bbf}_{\Mbi}\|_1$ and $\hat{\hbf}_{\Mat}\in\partial\|\hat{\bbf}_{\Mat}\|_1$ satisfying the optimality conditions
    \begin{subequations}
    \label{equ:pdw_opt_h_socp}
    \begin{align}
        -\frac{1}{n_m}(\ybf_{\Mat}-\Abf_{\Mat,\Xat}\hat{\xbf}_{\Iat}-\hat{\bbf}_{\Mat})+\lambda\hat{\hbf}_{\Mat}&=\bfzero,\\
        -\frac{1}{n_m}\left(\ybf_{\Mbi}-\Abf_{\Mbi,\Xbd}\hat{\xbf}_{\Ibd}-\Abf_{\Mbi,\Xat}\hat{\xbf}_{\Iat}-\hat{\bbf}_{\Mbi}\right)+\lambda\hat{\hbf}_{\Mbi}&=\bfzero.
    \end{align}
    \end{subequations}
3) Solve $(\hat{\hbf}_{\Msf},\hat{\hbf}_{\Mbo})$ via the zero-subgradient equation:
    \begin{equation}
        -\frac{1}{n_m}\left(\ybf-\Abf\hat{\xbf}-\hat{\bbf}\right)+\lambda\hat{\hbf}=\bfzero,
        \label{equ:pdw_opt_h_all_socp}
    \end{equation}
    where $\hat{\xbf}=\begin{bmatrix}\hat{\xbf}_{\Bsf}^\top&\hat{\xbf}_{\Bbd}^\top&\hat{\xbf}_{\Bat}^\top\end{bmatrix}^\top$ and $\hat{\bbf}=\begin{bmatrix}\bfzero^\top&\bfzero^\top&\hat{\bbf}_{\Mbi}^\top&\hat{\bbf}_{\Mat}^\top\end{bmatrix}^\top$ are solutions obtained in \eqref{equ:pdw-all}, and $\hat{\hbf}=\begin{bmatrix}\hat{\hbf}_{\Msf}^\top&\hat{\hbf}_{\Mbo}^\top&\hat{\hbf}_{\Mbi}^\top&\hat{\hbf}_{\Mat}^\top\end{bmatrix}^\top$ where $(\hat{\hbf}_{\Mbi},\hat{\hbf}_{\Mat})$ are given in \eqref{equ:pdw_opt_h}. Check whether strict feasibility conditions $\|\hat{\hbf}_{\Msf}\|_\infty<1$ and $\|\hat{\hbf}_{\Mbo}\|_\infty<1$ hold.

\begin{lemma}
If the PDW procedure succeeds, then $(\hat{\xbf},\hat{\bbf})$ that is optimal for \eqref{equ:pdw-all-socp} is also optimal for \eqref{equ:primal_lasso_socp}. Furthermore, for any optimal solution $(\tilde{\xbf},\tilde{\bbf})$, if $\hat{\xbf}_\Iat=\tilde{\xbf}_\Iat$, it holds that $\hat{\xbf}_\Isf=\tilde{\xbf}_\Isf$ and $\hat{\xbf}_\Ibd=\tilde{\xbf}_\Ibd$ (i.e., uniqueness property in the weak sense).
\label{lem:unique_pdw_sol_socp}
\end{lemma}
\begin{proof}
The KKT conditions of \eqref{equ:primal_lasso_socp} for a given primal-dual pair $(\hat{\xbf},\hat{\bbf})$ and $\hat{\hbf}$ are given by:
\begin{subequations}
\label{equ:s1-lasso-kkt-socp}
\begin{align}
    \frac{1}{n_m}\Abf^\top(\ybf-\Abf\hat{\xbf}-\hat{\bbf})+\lambda\sum_{\ell\in\Lcal}(\nu_\ell\cbf_\ell+\Dbf_\ell^\top\bmu_\ell)&=\bfzero,\label{equ:station_lasso_x_socp}\\
    -\frac{1}{n_m}\left(\ybf-\Abf\hat{\xbf}-\hat{\bbf}\right)+\lambda\hat{\hbf}&=\bfzero,\label{equ:station_lasso_b_socp}\\
    \hat{\hbf}\in\partial\|\hat{\bbf}\|_1,\quad\|\hat{\hbf}\|_\infty&\leq 1
\end{align}
\end{subequations}
If PDW succeeds, then the optimality conditions \eqref{equ:s1-lasso-kkt-socp} are satisfied, which certify the optimality of  $(\hat{\xbf},\hat{\bbf})$. The subgradient $\hat{\hbf}$ satisfies $\|\hat{\hbf}_{\Msf}\|_\infty<1$,  $\|\hat{\hbf}_{\Mbo}\|_\infty<1$ and $\left\langle \hat{\hbf},\hat{\bbf}\right\rangle=\|\hat{\bbf}\|_1$. Now, let $(\tilde{\xbf},\tilde{\bbf})$ be any other optimal, and let $F(\xbf,\bbf)=\frac{1}{2n_m}\|\ybf-\Abf\xbf-\bbf\|_2^2$; then,
\begin{equation*}
    F(\hat{\xbf},\hat{\bbf})+\lambda\left\langle \hat{\hbf},\hat{\bbf}\right\rangle=F(\tilde{\xbf},\tilde{\bbf})+\lambda\|\tilde{\bbf}\|_1,
\end{equation*}
and hence,
\begin{equation*}
    F(\hat{\xbf},\hat{\bbf})+\lambda\left\langle \hat{\hbf},\hat{\bbf}-\tilde{\bbf}\right\rangle=F(\tilde{\xbf},\tilde{\bbf})+\lambda\left(\|\tilde{\bbf}\|_1-\left\langle \hat{\hbf},\tilde{\bbf}\right\rangle\right).
\end{equation*}
By the optimality conditions in \eqref{equ:s1-lasso-kkt-socp}, we have 
$\lambda\hat{\hbf}=-\nabla_b F(\hat{\xbf},\hat{\bbf})=\frac{1}{n_m}(\ybf-\Abf\hat{\xbf}-\hat{\bbf})$, which implies that
\begin{align*}
    F(\hat{\xbf},\hat{\bbf})-\left\langle \nabla_b F(\hat{\xbf},\hat{\bbf}),\hat{\bbf}-\tilde{\bbf}\right\rangle-F(\tilde{\xbf},\tilde{\bbf})=\lambda\left(\|\tilde{\bbf}\|_1-\left\langle \hat{\hbf},\tilde{\bbf}\right\rangle\right)\leq 0
\end{align*}
due to convexity. We thus have $\|\tilde{\bbf}\|_1\leq\left\langle \hat{\hbf},\tilde{\bbf}\right\rangle$. Since by Holder's inequality, we also have $\left\langle \hat{\hbf},\tilde{\bbf}\right\rangle\leq\|\hat{\hbf}\|_\infty\|\tilde{\bbf}\|_1$, and $\|\hat{\hbf}\|_\infty\leq 1$, it holds that $\|\tilde{\bbf}\|_1=\left\langle \hat{\hbf},\tilde{\bbf}\right\rangle$. Since by the success of PDW, $\|\hat{\hbf}_{\Msf}\|_\infty<1$,  $\|\hat{\hbf}_{\Mbo}\|_\infty<1$, we have $\tilde{\bbf}_j=0$ for $j\in \Msf\cup\Mbo$. To show the weak uniqueness, let $(\tilde{\xbf},\tilde{\bbf})$ be another optimal solution, and assume that $\hat{\xbf}_{\Iat}=\tilde{\xbf}_{\Iat}$. Then, by fixing ${\xbf}_\Iat$ in the optimization \eqref{equ:pdw-all-socp} at $\hat{\xbf}_{\Iat}$ and by the lower eigenvalue condition, the the function is strictly convex in $\xbf_{\Isf}$, $\xbf_{\Ibd}$ and $\bbf_{\Mbi}$.
\end{proof}

\subsection*{Proof of Theorem \ref{thm:lasso_est_cyber_socp}}
\begin{proof}
\textbf{Part 1)}: By the construction of PDW, we have $\hat{\bbf}_{\Msf}=\bbf_{\natural \Msf}=\bfzero$ and $\hat{\bbf}_{\Mbo}=\bbf_{\natural \Mbo}=\bfzero$. In the following, we allow the optimal solution $\hat{\xbf}_\Iat$ and $\hat{\bbf}_\Mat$ of \eqref{equ:pdw-all-socp}  to take any value as long as the nonbinding SOC constraints assumption is satisfied. Thus, for any given $\hat{\xbf}_\Iat$ and $\hat{\bbf}_\Mat$, we can fix $\xbf_\Iat$ and ${\bbf}_\Mat$ in \eqref{equ:pdw-all-socp} and solve the following smaller program:

\begin{subequations}
	\begin{align}
    &\underset{{\bbf}_{\Mbi},\xbf_{\Isf},\xbf_{\Ibd}}{\text{min~~~}}
	&& \hspace{-0.5cm}\frac{1}{2n_m}\Bigg\|\underbrace{\begin{bmatrix}\ybf_{\Msf}\\\ybf_{\Mbo}\\\zbf_{\Mbi}\end{bmatrix}}_{\zbfo}-\underbrace{\begin{bmatrix}\Abf_{\Msf,\Xsf}&\Abf_{\Msf,\Xbd}\\\bfzero&\Abf_{\Mbo,\Xbd}\\
        \bfzero&\Abf_{\Mbi,\Xbd}\end{bmatrix}}_{\Abfo}\underbrace{\begin{bmatrix}\xbf_{\Isf}\\\xbf_{\Ibd}\end{bmatrix}}_{\xbfo}-\begin{bmatrix}\bfzero\\\bfzero\\\bbf_{\Mbi}\end{bmatrix}\Bigg\|_2^2+\lambda\left\|\bbf_{\Mbi}\right\|_1,\\
	& \;\;\text{subject to~~~}
	&&\hspace{-0.5cm}\cbf_\ell^\top\xbf\geq\left\|\Dbf_\ell\xbf\right\|_2,\qquad\qquad\forall \ell\in\Lcal\setminus\LBat,
	\end{align}
	\label{equ:pdw-all-small-socp}
\end{subequations}
where $\zbf_{\Mbi}=\ybf_{\Mbi}-\Abf_{\Mbi,\Xat}\hat{\xbf}_\Iat=\Abf_{\Mbi,\Xbd}\xbf_{\natural\Ibd}+\tbbf_{\Mbi}$ and $\tbbf_{\Mbi}=\Abf_{\Mbi,\Xat}(\xbf_{\natural\Iat}-\hat{\xbf}_{\Iat})$. Let $\Ibfo$ be an identity matrix of size $n_m-|\Mat|$, and $\xbfo$, $\cbf_\ell^\circ$ and $\Dbf_\ell^\circ$ be the subvector and submatrix of $\xbf$, $\cbf_\ell$ and $\Dbf_\ell$ indexed by $\Xsf$ and $\Xbd$, respectively, and $\wbfo$ be the subvector of $\wbf$ indexed by $\Msf\cup\Mbo\cup\Mbi$. Thus, we have $\zbfo=\Abfo\xbfo_\natural+\wbfo_\natural+\Ibfot_{\Mbi}\tbbf_{\Mbi}$. The solution $(\xbf_\Isf,\xbf_\Ibd,\bbf_{\Mbi})$ of \eqref{equ:pdw-all-small-socp} is unique and coincides with that of \eqref{equ:pdw-all-socp} due to the lower eigenvalue condition. Thus, the zero-subgradient condition \eqref{equ:pdw_opt_h_socp} is satisfied, which together with \eqref{equ:pdw_opt_h_all_socp} can be written as:
\begin{equation}
    -\frac{1}{n_m}\left(\begin{bmatrix}\Abf_{\Msf,\Xsf}&\Abf_{\Msf,\Xbd}\\\bfzero&\Abf_{\Mbo,\Xbd}\\
        \bfzero&\Abf_{\Mbi,\Xbd}\end{bmatrix}\begin{bmatrix}\xbf_{\natural\Isf}-\hat{\xbf}_{\Isf}\\\xbf_{\natural\Ibd}-\hat{\xbf}_{\Ibd}\end{bmatrix}+\begin{bmatrix}
    \bfzero\\\bfzero\\\tbbf_{\Mbi}-\hat{\bbf}_{\Mbi}\end{bmatrix}\right)-\frac{1}{n_m}\begin{bmatrix}
    \wbf_{\natural\Msf}\\
    \wbf_{\natural\Mbo}\\
    \wbf_{\natural\Mbi}
    \end{bmatrix}+\lambda\begin{bmatrix}
    \hat{\hbf}_\Msf\\\hat{\hbf}_{\Mbo}\\\hat{\hbf}_{\Mbi}
    \end{bmatrix}=\bfzero.
\end{equation}
We can partition the above relation into equations indexed by $\Mbi$, which can be rearranged as:
\begin{equation}
    \hat{\hbf}_{\Mbi}=\frac{1}{n_m\lambda}\begin{bmatrix}
    \Ibfo_{\Mbi}\Abfo& \Ibfo_{\Mbi}\Ibfot_{\Mbi}
    \end{bmatrix}\begin{bmatrix}
    \xbfo_\natural-{\hxbfo}\\ \tbbf_{\Mbi}-\hat{\bbf}_{\Mbi}
    \end{bmatrix}
    +\frac{1}{n_m\lambda}\Ibfo_{\Mbi}\wbfo_\natural,
    \label{equ:cond_b_Mbi_socp}
\end{equation}
as well as those indexed by $\Msf\cup\Mbo$, which can be solved for $\hat{\hbf}_{\Msf\cup\Mbo}=\begin{bmatrix}\hat{\hbf}_{\Msf}^\top&\hat{\hbf}_{\Mbo}^\top\end{bmatrix}^\top$:
\begin{equation}
    \hat{\hbf}_{\Msf\cup\Mbo}=\frac{1}{n_m\lambda}\Ibfo_{\Msf\cup\Mbo}\left(\Abfo(\xbfo_\natural-{\hxbfo})+\wbfo_\natural\right).
    \label{equ:h_j_c_equ_socp}
\end{equation}
Since $\hxbfo$ is the optimal solution of \eqref{equ:pdw-all-small-socp}, it satisfies the optimality condition:
\begin{equation}
    \frac{1}{n_m}\Abfot\left(\Abfo(\xbfo_\natural-\hxbfo)+\wbfo_\natural+\Ibfot_{\Mbi}(\tbbf_{\Mbi}-\hat{\bbf}_{\Mbi})\right)+\sum_{\ell\in\Latbi\cup\LBbd\cup\LBsf}\hat{\nu}_\ell\cbf_\ell^\circ+\Dbf_\ell^{\circ\top}\hat{\ubf}_\ell=\bfzero
    \label{equ:opt_x_small_cond_socp}
\end{equation}
Combining \eqref{equ:cond_b_Mbi_socp}, \eqref{equ:h_j_c_equ_socp} and \eqref{equ:opt_x_small_cond_socp} and after some elementary operations, it yields that
\begin{equation}
    \lambda\Abfot_{\Msf\cup\Mbo}\hat{\hbf}_{\Msf\cup\Mbo}+\lambda\Abfot_{\Mbi}\hat{\hbf}_{\Mbi}+\sum_{\ell\in\Latbi\cup\LBbd\cup\LBsf}\hat{\nu}_\ell\cbf_\ell^\circ+\Dbf_\ell^{\circ\top}\hat{\ubf}_\ell=\bfzero.
    \label{equ:opt_x_small_cond_socp2}
\end{equation}
By Lemma \ref{lemma:local_sufficient_cond_socp}, for any $\hat{\hbf}_{\Mbi}\in\partial\|\hat{\bbf}_{\Mbi}\|_1$, there always exist $\hat{\hbf}_{\Msf\cup\Mbo}$  and $\{\hat{\nu}_\ell,\hat{\ubf}_\ell\}_{\Latbi\cup\LBbd\cup\LBsf}$ such that $\|\hat{\hbf}_{\Msf\cup\Mbo}\|_\infty<1$. Thus, the strict feasibility condition is satisfied deterministically.

\textbf{Part 2)}: Thus, by the lower eigenvalue condition and definition of $\Qbfo_{\Mbi}=\begin{bmatrix}\Abfo&\Ibfot_{\Mbi}\end{bmatrix}$ and $\hat{\hbf}=\begin{bmatrix}\hat{\hbf}_{\Msf\cup\Mbo}^\top&\hat{\hbf}_{\Mbi}^\top\end{bmatrix}^\top$, we can solve \eqref{equ:cond_b_Mbi_socp}, \eqref{equ:opt_x_small_cond_socp} and \eqref{equ:opt_x_small_cond_socp2}:
\begin{align}
    \bDelta&\coloneqq\begin{bmatrix}
    \xbfo_\natural-{\hxbfo}\\ \tbbf_{\Mbi}-\hat{\bbf}_{\Mbi}
    \end{bmatrix}\nonumber\\
    &={
    -(\Qbfot_{\Mbi}\Qbfo_{\Mbi})^{-1}\Qbfot_{\Mbi}\wbfo_\natural+n_m\lambda(\Qbfot_{\Mbi}\Qbfo_{\Mbi})^{-1}\begin{bmatrix}\Abfot_{\Msf\cup\Mbo}\hat{\hbf}_{\Msf\cup\Mbo}+\Abfot_{\Mbi}\hat{\hbf}_{\Mbi}\\
    \hat{\hbf}_{\Mbi}\end{bmatrix}}\nonumber\\
    &=
    -(\Qbfot_{\Mbi}\Qbfo_{\Mbi})^{-1}\Qbfot_{\Mbi}\wbfo_\natural+n_m\lambda(\Qbfot_{\Mbi}\Qbfo_{\Mbi})^{-1}\Qbfot_{\Mbi}\hat{\hbf},
    \label{equ:delta_err_socp}
\end{align}
Let $\Ibf_x$ and $\Ibf_b$ denote the matrices that consist of the first $|\Xsf|+|\Xbd|$ rows and the last $|\Mbi|$ rows of the identity matrix of size $|\Xsf|+|\Xbd|+|\Mbi|$, respectively. Then, we can bound the estimation error $\bDelta$ in \eqref{equ:delta_err}. First, we bound the infinity norm of $\tbbf_{\Mbi}-\hat{\bbf}_{\Mbi}=\Ibf_b\bDelta$. By triangle inequality,
\begin{align}
    \|\Ibf_b\bDelta\|_\infty\leq \|\Ibf_b(\Qbfot_{\Mbi}\Qbfo_{\Mbi})^{-1}\Qbfot_{\Mbi}\wbfo_\natural\|_\infty+n_m\lambda\|\Ibf_b(\Qbfot_{\Mbi}\Qbfo_{\Mbi})^{-1}\Qbfot_{\Mbi}\|_\infty.
\end{align}
Since the second term is deterministic, we will bound the first term similar to Theorem \ref{thm:l1_est_cyber}. This concludes the proof

\textbf{Part 3)}: Now, we bound the $\ell_2$ norm of the signal error $\xbfo_\natural-{\hxbfo}=\Ibf_x\bDelta$,
\begin{align}
    \|\Ibf_x\bDelta\|_2\leq \|\Ibf_x(\Qbfot_{\Mbi}\Qbfo_{\Mbi})^{-1}\Qbfot_{\Mbi}\wbfo_\natural\|_2+n_m\lambda\|\Ibf_x(\Qbfot_{\Mbi}\Qbfo_{\Mbi})^{-1}\Qbfot_{\Mbi}\|_{\infty,2}.
\end{align}
For the first term, we can apply standard sub-gaussian concentration. The second term is deterministic. Combining them together yields the results.

\end{proof}


\end{document}